\documentclass[12pt]{article}
\usepackage[a4paper,tmargin=2cm,bmargin=2cm,lmargin=2cm,rmargin=2cm]{geometry}
\usepackage[T1]{fontenc}
\usepackage[utf8]{inputenc}
\usepackage{amssymb}
\usepackage{amsmath}
\usepackage{times}
\usepackage{graphicx}
\usepackage{mathtools}
\usepackage{multicol}
\usepackage[font=small]{caption} 
\usepackage{abraces} 
\usepackage[english]{babel}
\usepackage[colorlinks=true,linkcolor=blue,urlcolor=blue,citecolor=blue]{hyperref}

\newlength{\highlightlength}
\newlength{\highlightlengthbis}

\newcommand{\bchi}{{{\boldsymbol{\chi}}}}

\newcommand{\bea}{\begin{eqnarray}}
\newcommand{\eea}{\end{eqnarray}}
\newcommand{\be}{\begin{equation}}
\newcommand{\ee}{\end{equation}}
\newcommand{\Eq}[1]{Eq.~(\ref{#1})}
\newcommand{\Eqs}[1]{Eqs.~(\ref{#1})}
\newcommand{\eq}[1]{(\ref{#1})}
\newcommand{\rme}{\mathrm{e}}

\newcommand{\nn}{\nonumber}
\newcommand{\ca}[1]{\mathcal{#1}}
\graphicspath{{./figures/},{./},{../figures/}}

\newcommand{\fig}[2]{\includegraphics[width=#1\columnwidth]{./#2}}

\newcommand{\ds}{\displaystyle}
\renewcommand{\a}{\alpha}

\newcommand{\xt}{{\mathfrak{t}}}

\newcommand{\xr}{{\mathfrak{r}}}

\newcommand{\cc}[1]{\mathcal{#1}}

\newcommand{\LE}{\operatorname{LE}}

\newcounter{theorem}

\usepackage{array}
\usepackage{multirow}
\usepackage{amsthm}
\newtheorem{theoremm}{Theorem}[section]

\newtheorem{proposition}[theoremm]{Proposition}
\newtheorem{definition}[theoremm]{Definition}
\newtheorem{checking}[theoremm]{Check}

\usepackage{tikz}
\usetikzlibrary{scopes}
\usetikzlibrary{calc}
\usetikzlibrary{arrows,shapes,positioning}
\usetikzlibrary{decorations.markings}
\tikzstyle arrowstyle=[scale=1]
\tikzstyle directed=[postaction={decorate,decoration={markings,
 mark=at position .65 with {\arrow[arrowstyle]{stealth}}}}]
\tikzstyle endreversedirected=[postaction={decorate,decoration={markings,
 mark=at position 1.0 with {\arrow[arrowstyle]{stealth}}}}]
\tikzset{
 enddirected/.style={-stealth}
} 
\tikzstyle reverse directed=[postaction={decorate,decoration={markings,
 mark=at position .65 with {\arrowreversed[arrowstyle]{stealth};}}}]
\newcommand{\drawarc}[4][]{%
 \draw[#1]
 let
 \p1 = ($(#3)-(#2)$),
 \p2 = ($(#4)-(#2)$),
 \n1 = {atan2(\y1,\x1)},
 \n2 = {atan2(\y2,\x2)},
 \n3 = {veclen(\x1,\y1)},
 \n4 = {mod(\n2-\n1+360,360)}
 in
 (#2) ++(\n1:\n3)
 arc[start angle=\n1, delta angle=\n4, radius=\n3];
}

\newcommand{\ifTikzExternal}[1]{}
\usepackage{tikz}
\ifTikzExternal{\usetikzlibrary{external}\tikzexternalize[prefix=tikz/,up to date check=md5]}
\usetikzlibrary{arrows,shapes,positioning,decorations.markings}
\usetikzlibrary{arrows.meta}
\usetikzlibrary{decorations.markings}
\tikzstyle arrowstyle=[scale=1]
\tikzstyle directed=[postaction={decorate,decoration={markings,
    mark=at position .65 with {\arrow[arrowstyle]{stealth}}}}]
\tikzstyle endreversedirected=[postaction={decorate,decoration={markings,
    mark=at position 1.0 with {\arrow[arrowstyle]{stealth}}}}]

\tikzset{
  enddirected/.style={-stealth}
} 
\tikzstyle reverse directed=[postaction={decorate,decoration={markings,
    mark=at position .65 with {\arrowreversed[arrowstyle]{stealth};}}}]
\tikzset{enddirected double/.style={
shorten >=2.8pt,
postaction={line width=\pgflinewidth,enddirected},
postaction={double,draw}}}  
\tikzset{double enddirected/.style={
shorten >=2.8pt,
postaction={line width=\pgflinewidth,enddirected},
postaction={double,draw}}}  
    
\usepackage[dvipsnames]{xcolor}
\usetikzlibrary{calc}

\newcommand{\vertexIntPsiChi}{\ifTikzExternal{\tikzsetnextfilename{vertexIntPsiChi}}{\vcenter{\hbox{{\begin{tikzpicture}[scale=1]
\coordinate (x1a) at (0,0) ;
\coordinate (x1b) at  (0,0.5) ;
\coordinate (x1c) at  (0,1) ;
\coordinate (x2a) at  (0.5,0.5) ;
\coordinate (x2c) at  (-0.5,.5) ;
\draw [thick,blue,enddirected] (x1b)  --(x1c);
\draw [thick,blue] (x1a)  --(x1b);
\draw [thick,Green,enddirected] (x2a) -- (x2c);
\fill ([xshift=-2pt,yshift=-1.5pt]x1b) rectangle ++ (4pt,3pt);
\end{tikzpicture}}}}}}
\newcommand{\vertexIntPhiPsi}{\ifTikzExternal{\tikzsetnextfilename{vertexIntPhiPsi}}{\vcenter{\hbox{{\begin{tikzpicture}[scale=1]
\coordinate (x1a) at (0,0) ;
\coordinate (x1b) at  (0,0.5) ;
\coordinate (x1c) at  (0,1) ;
\coordinate (x2a) at  (-.5,0) ;
\coordinate (x2c) at  (0.5,1) ;
\draw [thick,blue,enddirected] (x1b)  --(x1c);
\draw [thick,blue] (x1a)  --(x1b);
\draw [thick,red,enddirected] (x2a) -- (x2c);
\fill ([xshift=-2pt,yshift=-1.5pt]x1b) rectangle ++ (4pt,3pt);
\end{tikzpicture}}}}}}
\newcommand{\vertexIntPhiChi}{\ifTikzExternal{\tikzsetnextfilename{vertexIntPhiChi}}{\vcenter{\hbox{{\begin{tikzpicture}[scale=1]
\coordinate (x1a) at (-.5,0) ;
\coordinate (x1b) at  (0,0.5) ;
\coordinate (x1c) at  (0.5,1) ;
\coordinate (x2a) at  (0.5,0.5) ;
\coordinate (x2c) at  (-0.5,.5) ;
\draw [thick,red,enddirected] (x1b)  --(x1c);
\draw [thick,red] (x1a)  --(x1b);
\draw [thick,Green,enddirected] (x2a) -- (x2c);
\fill ([xshift=-2pt,yshift=-1.5pt]x1b) rectangle ++ (4pt,3pt);
\end{tikzpicture}}}}}}
\newcommand{\vertexIntChi}{\ifTikzExternal{\tikzsetnextfilename{vertexIntChi}}{\vcenter{\hbox{{\begin{tikzpicture}[scale=1]
\coordinate (x1a) at (0,0) ;
\coordinate (x1b) at  (0,0.5) ;
\coordinate (x1c) at  (0,1) ;
\coordinate (x2a) at  (0.5,0) ;
\coordinate (x2b) at  (0.5,.5) ;
\coordinate (x2c) at  (0.5,1) ;
\coordinate (blank1) at (-0.07,0) ;
\coordinate (blank2) at  (0.57,0) ;
\draw[white] (blank1)--(blank2);
\draw [thick,Green,enddirected] (x1a)  --(x1c);
\draw [thick,dotted, dash expand off] (x1b) -- (x2b);
\draw [thick,Green,enddirected] (x2a) -- (x2c);
\fill ([xshift=-2pt,yshift=-1.5pt]x1b) rectangle ++ (4pt,3pt);
\fill ([xshift=-2pt,yshift=-1.5pt]x2b) rectangle ++ (4pt,3pt);
\end{tikzpicture}}}}}}
\newcommand{\vertexIntPhi}{\ifTikzExternal{\tikzsetnextfilename{vertexIntPhi}}{\vcenter{\hbox{{\begin{tikzpicture}[scale=1]
\coordinate (x1a) at (0,0) ;
\coordinate (x1b) at  (0,0.5) ;
\coordinate (x1c) at  (0,1) ;
\coordinate (x2a) at  (0.5,0) ;
\coordinate (x2b) at  (0.5,.5) ;
\coordinate (x2c) at  (0.5,1) ;
\coordinate (blank1) at (-0.07,0) ;
\coordinate (blank2) at  (0.57,0) ;
\draw[white] (blank1)--(blank2);
\draw [thick,red,enddirected] (x1a)  --(x1c);
\draw [thick,dotted, dash expand off] (x1b) -- (x2b);
\draw [thick,red,enddirected] (x2a) -- (x2c);
\fill ([xshift=-2pt,yshift=-1.5pt]x1b) rectangle ++ (4pt,3pt);
\fill ([xshift=-2pt,yshift=-1.5pt]x2b) rectangle ++ (4pt,3pt);
\end{tikzpicture}}}}}}
\newcommand{\vertexIntPsiChiRel}{\ifTikzExternal{\tikzsetnextfilename{vertexIntPsiChiRel}}{\vcenter{\hbox{{\begin{tikzpicture}[scale=1]
\coordinate (x1a) at (0,0) ;
\coordinate (x1b) at  (0,0.5) ;
\coordinate (x1c) at  (0,1) ;
\coordinate (x2a) at  (0.5,0.5) ;
\coordinate (x2c) at  (-0.5,.5) ;
\coordinate (blank1) at (-0.07,0) ;
\coordinate (blank2) at  (0.57,0) ;
\draw[white] (blank1)--(blank2);
\draw [thick,white] (x1b)  --(x1c);
\draw [thick,blue] (x1a)  --(x1b);
\draw [thick,Green,enddirected] (x2a) -- (x2c);
\fill ([xshift=-2pt,yshift=-1.5pt]x1b) rectangle ++ (4pt,3pt);
\end{tikzpicture}}}}}}

\newcommand{\vertexOneDashed}{\ifTikzExternal{\tikzsetnextfilename{vertexOneDashed}}{\vcenter{\hbox{{\begin{tikzpicture}[scale=1]
\coordinate (x1) at (0,0) ;
\coordinate (x2) at  (0,0.5) ;
\coordinate (x2a) at  (.4,0.9) ;
\coordinate (x2b) at  (0,1) ;
\coordinate (x2c) at  (0.5,.5) ;
\coordinate (blank1) at (-0.07,0) ;
\coordinate (blank2) at  (0.5,0) ;
\draw[white] (blank1)--(blank2);
\draw [thick,red] (x1) -- (x2) ;
\draw [thick,blue,densely dashed,enddirected]  (x2) -- (x2b);
\draw [thick,red,enddirected] (x2) -- (x2a);
\draw [thick,Green,enddirected] (x2) -- (x2c);
\fill (x2) circle (1.5pt);
\end{tikzpicture}}}}}}

\newcommand{\vertexOneDashedmean}{\ifTikzExternal{\tikzsetnextfilename{vertexOneDashedmean}}{\vcenter{\hbox{{\begin{tikzpicture}[scale=1]
\coordinate (x1) at (0,0) ;
\coordinate (blank) at (-0.1,0) ;
\coordinate (x2) at  (0,0.5) ;
\coordinate (x2a) at  (.4,0.9) ;
\coordinate (x2b) at  (0,1) ;
\fill[white] (blank) circle (0.1pt);
\draw [thick,red] (x1) -- (x2) ;
\draw [thick,blue,densely dashed,enddirected]  (x2) -- (x2b);
\draw [thick,red,enddirected] (x2) -- (x2a);
\fill (x2) circle (1.5pt);
\end{tikzpicture}}}}}}
\newcommand{\vertexOneDashedFluctuating}{\ifTikzExternal{\tikzsetnextfilename{vertexOneDashedFluctuating}}{\vcenter{\hbox{{\begin{tikzpicture}[scale=1]
\coordinate (x1) at (0,0) ;
\coordinate (blank) at (-0.1,0) ;
\coordinate (x2) at  (0,0.5) ;
\coordinate (x2a) at  (.4,0.9) ;
\coordinate (x2b) at  (0,1) ;
\coordinate (x2c) at  (0.5,.5) ;
\fill[white] (blank) circle (0.1pt);
\draw [thick,red] (x1) -- (x2) ;
\draw [thick,blue,densely dashed,enddirected]  (x2) -- (x2b);
\draw [thick,red,enddirected] (x2) -- (x2a);
\draw [thick,Green,enddirected,densely dashed] (x2) -- (x2c);
\fill (x2) circle (1.5pt);
\end{tikzpicture}}}}}}

\newcommand{\vertexTwoRel}{\ifTikzExternal{\tikzsetnextfilename{vertexTwoRel}}{\vcenter{\hbox{{\begin{tikzpicture}[scale=1]
\coordinate (x1) at (0,0) ;
\coordinate (x2) at  (0,0.5) ;
\coordinate (x2a) at  (0,1) ;
\coordinate (x3) at  (0.5,0.5) ;
\coordinate (x4) at  (1,.5) ;
\coordinate (x3a) at  (0.5,1) ;
\draw [thick,red] (x1) -- (x2);
\draw [thick,red,enddirected]  (x3) -- (x3a);
\draw [thick,dotted, dash expand off] (x2) -- (x3) node [pos=0.5, sloped, font=\scriptsize\bfseries, black ] {||} ;
\draw [thick,Green,enddirected,densely dashed] (x3) -- (x4);
\fill (x2) circle (1.5pt);
\fill (x3) circle (1.5pt);
\end{tikzpicture}}}}}}

\newcommand{\vertexTwoMinus}{\ifTikzExternal{\tikzsetnextfilename{vertexTwoMinus}}{\vcenter{\hbox{{\begin{tikzpicture}[scale=1]
\coordinate (x1) at (0,0) ;
\coordinate (x2) at  (0,0.5) ;
\coordinate (x2a) at  (0,1) ;
\coordinate (x3) at  (0.5,0.5) ;
\coordinate (x4) at  (1,.5) ;
\coordinate (x3a) at  (0.5,1) ;
\coordinate (blank1) at (-0.07,0) ;
\coordinate (blank2) at  (0.57,0) ;
\draw[white] (blank1)--(blank2);
\draw [thick,red] (x1) -- (x2);
\draw [thick,red,enddirected]  (x2) -- (x2a);
\draw [thick,dotted, dash expand off] (x2) -- (x3) node [pos=0.5, sloped, font=\scriptsize\bfseries, black ] {||} ;
\draw [thick,Green,enddirected,densely dashed] (x3) -- (x4);
\fill  (x2) circle (1.5pt);
\fill (x3) circle (1.5pt);
\end{tikzpicture}}}}}}

\newcommand{\vertexTwoBred}{\ifTikzExternal{\tikzsetnextfilename{vertexTwoBred}}{\vcenter{\hbox{\!\!\begin{tikzpicture}
\coordinate (x1) at (0,0) ;
\coordinate (x2) at  (0,0.5) ;
\coordinate (x4) at  (0.5,.5) ;
\coordinate (x5) at  (0,1) ;
\draw [thick,red,enddirected]  (x1) -- (x5) node [pos=0.7, sloped, font=\scriptsize\bfseries, black ] {||} ;
\draw [thick,Green,enddirected,densely dashed] (x2) -- (x4);
\fill (x2) circle (1.5pt);
\end{tikzpicture}}}}}

\newcommand{\vertexTwoCred}{\ifTikzExternal{\tikzsetnextfilename{vertexTwoCred}}{\vcenter{\hbox{{\!\!\begin{tikzpicture}[scale=1]
\coordinate (x1) at (0,0) ;
\coordinate (x2) at  (0,0.5) ;
\coordinate (x3) at  (-.5,1) ;
\coordinate (x4) at  (0.5,.5) ;
\coordinate (x5) at  (0,1) ;
\draw [thick,red,enddirected]  (x1) -- (x5)  node [pos=0.7, sloped, font=\scriptsize\bfseries, black ] {|} ;
\draw [thick,Green,enddirected,densely dashed] (x2) -- (x4)  node [pos=0.45, sloped, font=\scriptsize\bfseries, black ] {|} ;
\fill (x2) circle (1.5pt);
\end{tikzpicture}}}}}}

\newcommand{\vertexIntPsi}{\ifTikzExternal{\tikzsetnextfilename{vertexIntPsi}}{\vcenter{\hbox{{\begin{tikzpicture}[scale=1]
\coordinate (x1a) at (0,0) ;
\coordinate (x1b) at  (0,0.5) ;
\coordinate (x1c) at  (0,1) ;
\coordinate (x2a) at  (0.5,0) ;
\coordinate (x2b) at  (0.5,.5) ;
\coordinate (x2c) at  (0.5,1) ;
\coordinate (blank1) at (-0.07,0) ;
\coordinate (blank2) at  (0.57,0) ;
\draw[white] (blank1)--(blank2);
\draw [thick,blue,enddirected] (x1a)  --(x1c);
\draw [thick,dotted, dash expand off] (x1b) -- (x2b);
\draw [thick,blue,enddirected] (x2a) -- (x2c);
\fill ([xshift=-2pt,yshift=-1.5pt]x1b) rectangle ++ (4pt,3pt);
\fill ([xshift=-2pt,yshift=-1.5pt]x2b) rectangle ++ (4pt,3pt);
\end{tikzpicture}}}}}}
\newcommand{\propRed}{\ifTikzExternal{\tikzsetnextfilename{propRed}}{\vcenter{\hbox{{\begin{tikzpicture}[scale=1]
\coordinate (x1) at (0,0) ;
\coordinate (x2) at  (0.5,0.5);
\draw [thick,red,enddirected] (x1) -- (x2) ;
\end{tikzpicture}}}}}}
\newcommand{\propBlue}{\ifTikzExternal{\tikzsetnextfilename{propBlue}}{\vcenter{\hbox{{\begin{tikzpicture}[scale=1]
\coordinate (x1) at (0,0) ;
\coordinate (x2) at  (0,0.5);
\coordinate (x3a) at  (-0.07,0);
\coordinate (x3b) at  (0.07,0);
\draw [white] (x3a)--(x3b);
\draw [thick,blue,enddirected] (x1) -- (x2) ;
\end{tikzpicture}}}}}}
\newcommand{\propGreen}{\ifTikzExternal{\tikzsetnextfilename{propGreen}}{\vcenter{\hbox{{\begin{tikzpicture}[scale=1]
\coordinate (x1) at (0,0) ;
\coordinate (x2) at  (0.5,0);
\coordinate (x3a) at  (0,-0.07);
\coordinate (x3b) at  (0,0.07);
\draw [white] (x3a)--(x3b);
\draw [thick,Green,enddirected] (x1) -- (x2) ;
\end{tikzpicture}}}}}}
\newcommand{\diagOne}{\ifTikzExternal{\tikzsetnextfilename{diagOne}}{\vcenter{\hbox{{\begin{tikzpicture}[scale=1]
\coordinate (x1) at (-.25,-.25) ;
\coordinate (x2) at  (0,0) ;
\coordinate (x2b) at  (0,1) ;
\coordinate (x2c) at  (-.3,0) ;
\coordinate (x3) at  (1,1) ;
\coordinate (x3b) at (1.25,1.25) ;
\coordinate (x3c) at  (1.0,1.3) ;
\coordinate (x4) at  (0,1) ;
\coordinate (x4b) at  (-.3,1) ;
\draw [thick,red] (x1) -- (x2) ;
\draw [thick,blue,directed]  (x2) -- (x2b);
\draw [thick,red,directed] (x2) -- (x3);
\draw [thick,red,enddirected] (x3) -- (x3b);
\draw [thick,blue,densely dashed,enddirected] (x3) -- (x3c);
\draw [thick,Green,enddirected] (x2) -- (x2c);
\draw [thick,Green,directed] (x3) -- (x4);
\draw [thick,Green,enddirected] (x4) -- (x4b);
\fill (x2) circle (1.5pt);
\fill (x3) circle (1.5pt);
\fill ([xshift=-2pt,yshift=-1.5pt]x2b) rectangle ++ (4pt,3pt);
\end{tikzpicture}}}}}}

\newcommand{\diagOnesimp}{\ifTikzExternal{\tikzsetnextfilename{diagOnesimp}}{\vcenter{\hbox{{\begin{tikzpicture}[scale=1]
\coordinate (x1) at (-.25,-.25) ;
\coordinate (x2) at  (0,0) ;
\coordinate (x2b) at  (0,1) ;
\coordinate (x2c) at  (-.3,0) ;
\coordinate (x3) at  (1,1) ;
\coordinate (x3b) at  (1.25,1.25) ;
\coordinate (x3c) at  (1.0,1.3) ;
\coordinate (x4) at  (0,1) ;
\coordinate (x4a) at  (0,1.3) ;
\coordinate (x4b) at  (-.3,1) ;
\draw [thick,red] (x1) -- (x2) ;
\draw [thick,blue,directed]  (x2) -- (x2b);
\draw [thick,red,directed] (x2) -- (x3);
\draw [thick,red,enddirected] (x3) -- (x3b);
\draw [thick,blue,densely dashed,enddirected] (x3) -- (x3c);
\draw [thick,Green,directed] (x3) -- (x4);
\draw [thick,Green,enddirected] (x4) -- (x4b);
\fill (x2) circle (1.5pt);
\fill (x3) circle (1.5pt);
\fill ([xshift=-2pt,yshift=-1.5pt]x2b) rectangle ++ (4pt,3pt);
\end{tikzpicture}}}}}}

\newcommand{\diagOneC}{\ifTikzExternal{\tikzsetnextfilename{diagOneC}}{\vcenter{\hbox{{\begin{tikzpicture}[scale=1]
\coordinate (x1) at (-.25,-.25) ;
\coordinate (x2) at  (0,0) ;
\coordinate (x2b) at  (0,1) ;
\coordinate (x2c) at  (-.3,0) ;
\coordinate (x3) at  (1,1) ;
\coordinate (x3b) at (1.25,1.25) ;
\coordinate (x3m1u) at  (.85,1.1) ;
\coordinate (x3m1d) at  (.85,0.9) ;
\coordinate (x3m2u) at  (.8,1.1) ;
\coordinate (x3m2d) at  (.8,0.9) ;
\coordinate (x4) at  (0,1) ;
\coordinate (x4a) at  (0,1.3) ;
\coordinate (x4b) at  (-.3,1) ;
\coordinate (x5) at  (.5,.5) ;
\coordinate (x5b) at  (.5,.8) ;
\coordinate (x5c) at  (.2,.5) ;
\draw [thick,red] (x1) -- (x2) ;
\draw [thick,blue,directed]  (x2) -- (x2b);
\draw [thick,red,directed] (x2) -- (x5);
\draw [thick,red,directed] (x5) -- (x3);
\draw [thick,red,enddirected] (x3) -- (x3b);
\draw [thick,Green,enddirected] (x2) -- (x2c);
\draw [thick,Green,directed] (x3) -- (x4);
\draw [thick,Green,enddirected] (x4) -- (x4b);
\draw [thick,black] (x3m1d)--(x3m1u);
\draw [thick,black] (x3m2d)--(x3m2u);
\draw [thick,blue,densely dashed,enddirected] (x5) -- (x5b);
\draw [thick,Green,enddirected] (x5) -- (x5c);
\fill (x2) circle (1.5pt);
\fill (x3) circle (1.5pt);
\fill ([xshift=-2pt,yshift=-1.5pt]x2b) rectangle ++ (4pt,3pt);
\fill (x5) circle (1.5pt);
\end{tikzpicture}}}}}}
\newcommand{\diagOneCsimp}{\ifTikzExternal{\tikzsetnextfilename{diagOneCsimp}}{\vcenter{\hbox{{\begin{tikzpicture}[scale=1]
\coordinate (x1) at (-.25,-.25) ;
\coordinate (x2) at  (0,0) ;
\coordinate (x2b) at  (0,1) ;
\coordinate (x2c) at  (-.3,0) ;
\coordinate (x3) at  (1,1) ;
\coordinate (x3b) at (1.25,1.25) ;
\coordinate (x3c) at  (1.0,1.3) ;
\coordinate (x3m1u) at  (.85,1.1) ;
\coordinate (x3m1d) at  (.85,0.9) ;
\coordinate (x3m2u) at  (.8,1.1) ;
\coordinate (x3m2d) at  (.8,0.9) ;
\coordinate (x4) at  (0,1) ;
\coordinate (x4a) at  (0,1.3) ;
\coordinate (x4b) at  (-.3,1) ;
\coordinate (x5) at  (.5,.5) ;
\coordinate (x5b) at  (.5,.8) ;
\draw [thick,red] (x1) -- (x2) ;
\draw [thick,blue,directed]  (x2) -- (x2b);
\draw [thick,red,directed] (x2) -- (x5);
\draw [thick,red,directed] (x5) -- (x3);
\draw [thick,red,enddirected] (x3) -- (x3b);
\draw [thick,Green,directed] (x3) -- (x4);
\draw [thick,black] (x3m1d)--(x3m1u);
\draw [thick,black] (x3m2d)--(x3m2u);
\draw [thick,blue,densely dashed,enddirected] (x5) -- (x5b);
\draw [thick,Green,enddirected] (x4) -- (x4b);
\fill (x2) circle (1.5pt);
\fill (x3) circle (1.5pt);
\fill ([xshift=-2pt,yshift=-1.5pt]x2b) rectangle ++ (4pt,3pt);
\fill (x5) circle (1.5pt);
\end{tikzpicture}}}}}}
\newcommand{\diagOneCsimpSimp}{\ifTikzExternal{\tikzsetnextfilename{diagOneCsimpSimp}}{\vcenter{\hbox{{\begin{tikzpicture}[scale=1]
\coordinate (x1) at (-.25,-.25) ;
\coordinate (x2) at  (0,0) ;
\coordinate (x2b) at  (0,1) ;
\coordinate (x4) at  (-.3,1) ;
\coordinate (x2c) at  (-.3,0) ;
\coordinate (x3b) at (0.25,1.25) ;
\coordinate (x3c) at  (1.0,1.3) ;
\coordinate (x5) at  (.5,.5) ;
\coordinate (x5b) at  (.5,.8) ;
\draw [thick,red] (x1) -- (x2) ;
\draw [thick,blue,directed]  (x2) -- (x2b);
\draw [thick,red,directed] (x2) -- (x5);
\draw [thick,red,directed] (x5) -- (x2b);
\draw [thick,red,enddirected] (x2b) -- (x3b);
\draw [thick,blue,densely dashed,enddirected] (x5) -- (x5b);
\draw [thick,Green,enddirected] (x2b) -- (x4);
\fill (x2) circle (1.5pt);
\fill ([xshift=-2pt,yshift=-1.5pt]x2b) rectangle ++ (4pt,3pt);
\fill (x5) circle (1.5pt);
\end{tikzpicture}}}}}}

\newcommand{\diagOneB}{\ifTikzExternal{\tikzsetnextfilename{diagOneB}}{\vcenter{\hbox{{\begin{tikzpicture}[scale=1]
\coordinate (x1) at (-.25,-.25) ;
\coordinate (x2) at  (0,0) ;
\coordinate (x2b) at  (0,1) ;
\coordinate (x2c) at  (-.3,0) ;
\coordinate (x3) at  (1,1) ;
\coordinate (x3b) at (1.25,1.25) ;
\coordinate (x3m1u) at  (.84,.96) ;
\coordinate (x3m1d) at  (0.96,0.84) ;
\coordinate (x3m2u) at  (.81,.93) ;
\coordinate (x3m2d) at  (.93,0.81) ;
\coordinate (x4) at  (0,1) ;
\coordinate (x4b) at  (-.3,1) ;
\coordinate (x5) at  (.5,.5) ;
\coordinate (x5b) at  (.5,.8) ;
\coordinate (x5c) at  (.2,.5) ;
\draw [thick,red] (x1) -- (x2) ;
\draw [thick,blue,directed]  (x2) -- (x2b);
\draw [thick,red,directed] (x2) -- (x5);
\draw [thick,red,directed] (x5) -- (x3);
\draw [thick,red,enddirected] (x3) -- (x3b);
\draw [thick,Green,enddirected] (x2) -- (x2c);
\draw [thick,Green,directed] (x3) -- (x4);
\draw [thick,Green,enddirected] (x4) -- (x4b);
\draw [thick,black] (x3m1d)--(x3m1u);
\draw [thick,black] (x3m2d)--(x3m2u);
\draw [thick,blue,densely dashed,enddirected] (x5) -- (x5b);
\draw [thick,Green,enddirected] (x5) -- (x5c);
\fill (x2) circle (1.5pt);
\fill (x3) circle (1.5pt);
\fill ([xshift=-2pt,yshift=-1.5pt]x2b) rectangle ++ (4pt,3pt);
\fill (x5) circle (1.5pt);
\end{tikzpicture}}}}}}
\newcommand{\diagOneBsimp}{\ifTikzExternal{\tikzsetnextfilename{diagOneBsimp}}{\vcenter{\hbox{{\begin{tikzpicture}[scale=1]
\coordinate (x1) at (-.25,-.25) ;
\coordinate (x2) at  (0,0) ;
\coordinate (x2b) at  (0,1) ;
\coordinate (x2c) at  (-.3,0) ;
\coordinate (x3) at  (1,1) ;
\coordinate (x3b) at (1.25,1.25) ;
\coordinate (x3c) at  (1.0,1.3) ;
\coordinate (x3m1u) at  (.84,.96) ;
\coordinate (x3m1d) at  (0.96,0.84) ;
\coordinate (x3m2u) at  (.81,.93) ;
\coordinate (x3m2d) at  (.93,0.81) ;
\coordinate (x4) at  (0,1) ;
\coordinate (x4a) at  (0,1.3) ;
\coordinate (x4b) at  (-.3,1) ;
\coordinate (x5) at  (.5,.5) ;
\coordinate (x5b) at  (.5,.8) ;
\draw [thick,red] (x1) -- (x2) ;
\draw [thick,blue,directed]  (x2) -- (x2b);
\draw [thick,red,directed] (x2) -- (x5);
\draw [thick,red,directed] (x5) -- (x3);
\draw [thick,red,enddirected] (x3) -- (x3b);
\draw [thick,Green,directed] (x3) -- (x4);
\draw [thick,Green,enddirected] (x4) -- (x4b);
\draw [thick,black] (x3m1d)--(x3m1u);
\draw [thick,black] (x3m2d)--(x3m2u);
\draw [thick,blue,densely dashed,enddirected] (x5) -- (x5b);
\fill (x2) circle (1.5pt);
\fill (x3) circle (1.5pt);
\fill ([xshift=-2pt,yshift=-1.5pt]x2b) rectangle ++ (4pt,3pt);
\fill (x5) circle (1.5pt);
\end{tikzpicture}}}}}}

\newcommand{\propdiagExt}{\ifTikzExternal{\tikzsetnextfilename{propdiagExt}}{\vcenter{\hbox{{\begin{tikzpicture}[scale=1]
\coordinate (x1) at (0,0) ;
\coordinate (x2) at  (0.6,0) ;
\coordinate (x1p) at  (-.2,0) ;
\coordinate (x2p) at  (0.8,0) ;
\fill (x1) circle (1.5pt);
\fill (x2) circle (1.5pt);
\draw [thick] (.3,0) circle (3mm);
\draw [black] (x1) -- (x1p);
\draw [black] (x2) -- (x2p);
\end{tikzpicture}}}}}}

\newcommand{\LHnoExternal}{\ifTikzExternal{\tikzsetnextfilename{LHnoExternal}}{\vcenter{\hbox{{\begin{tikzpicture}[scale=1]
\coordinate (x1) at (0,0) ;
\coordinate (x2) at  (0.6,0) ;
\coordinate (x3) at  (0.3,0.5) ;
\coordinate (x2p) at  (0.8,0) ;
\fill (x1) circle (1.5pt);
\fill (x2) circle (1.5pt);
\fill (x3) circle (1.5pt);
\draw [thick] (x2) arc (60:120:0.6) ;
\draw [thick] (x1) arc (-120:-60:0.6) ;
\draw [black,thick] (x1)  -- (x3) -- (x2);
\end{tikzpicture}}}}}}
\newcommand{\propdiagTwo}{\ifTikzExternal{\tikzsetnextfilename{propdiagTwo}}{\vcenter{\hbox{{\begin{tikzpicture}[scale=1]
\coordinate (x1) at (0,0) ;
\coordinate (x2) at  (0.6,0) ;
\coordinate (x3) at  (1.2,0) ;
\fill (x1) circle (1.5pt);
\fill (x2) circle (1.5pt);
\fill (x3) circle (1.5pt);
\draw [thick] (.3,0) circle (3mm);
\draw [thick] (.9,0) circle (3mm);
\end{tikzpicture}}}}}}

\newcommand{\LERWdiagOneBlack}{\ifTikzExternal{\tikzsetnextfilename{LERWdiagOneBlack}}{\vcenter{\hbox{{\begin{tikzpicture}[scale=1]
\coordinate (x) at  (1.5,0);    
\coordinate (x2) at (2,0) ;
\coordinate (y) at  (2.5,0.5); 
\coordinate (y2) at  (2,1);
\coordinate (z) at  (1.5,1) ;
\draw [thick,red] (x) -- (x2);
\draw [thick,red,enddirected]  (y2) -- (z);
\draw [thick,red](2.5,0.5) arc (0:90:0.5);
\draw [thick,red](2,0) arc (-90:0:0.5);
\draw [densely dashed,thick,blue] (x2) -- (y2);
\fill (y) circle (1.5pt);
\fill (x2) circle (1.5pt);
\fill (y2) circle (1.5pt);
\end{tikzpicture}}}}}}

\usepackage{tikz}
\ifTikzExternal{\usetikzlibrary{external}\tikzexternalize[prefix=tikz/,up to date check=md5]}
\usetikzlibrary{arrows,shapes,positioning,decorations.markings}
\usetikzlibrary{arrows.meta}
\usetikzlibrary{decorations.markings}
\tikzstyle arrowstyle=[scale=1]
\tikzstyle directed=[postaction={decorate,decoration={markings,
    mark=at position .65 with {\arrow[arrowstyle]{stealth}}}}]
\tikzstyle endreversedirected=[postaction={decorate,decoration={markings,
    mark=at position 1.0 with {\arrow[arrowstyle]{stealth}}}}]
\tikzset{  enddirected/.style={-stealth}
} 
\tikzstyle reverse directed=[postaction={decorate,decoration={markings,
    mark=at position .65 with {\arrowreversed[arrowstyle]{stealth};}}}]   
\tikzset{enddirected double/.style={
shorten >=2.8pt,
postaction={line width=\pgflinewidth,enddirected},
postaction={double,draw}}}  
\tikzset{double enddirected/.style={
shorten >=2.8pt,
postaction={line width=\pgflinewidth,enddirected},
postaction={double,draw}}}  
    
\usepackage[dvipsnames]{xcolor}
\usetikzlibrary{calc}

\newcommand{\bubble}{\vcenter{\hbox{\ifTikzExternal{\tikzsetnextfilename{bubble}}
{\begin{tikzpicture}[scale=1]
                \coordinate (x1) at (0,0) ;
                \draw[] (x1) arc (0:360: 0.3);
            \end{tikzpicture}}}}%
}
\newcommand{\tadpole}{\vcenter{\hbox{\ifTikzExternal{\tikzsetnextfilename{tadpole}}\begin{tikzpicture}[scale=1]
                \coordinate (x1) at (0,0) ;
                \draw[] (x1) arc (-90:270: 0.3);
                 \fill (x1) circle (1.5pt);
            \end{tikzpicture}\,}}%
}
\newcommand{\banana}{\ifTikzExternal{\tikzsetnextfilename{banana}}\vcenter{\hbox{\begin{tikzpicture}[scale=1]
                \coordinate (origin) at (0,0) ;
                \coordinate (end) at (0.6,0) ;
                \draw[] (origin) arc (180:-180: 0.3);
                 \fill (origin) circle (1.5pt);
                 \fill (end) circle (1.5pt);
            \end{tikzpicture}\,}}%
}
\newcommand{\RedBlueInt}{\vcenter{\hbox{\ifTikzExternal{\tikzsetnextfilename{RedBlueInt}}
        \begin{tikzpicture}[scale=1,x={(0.8,0)},y={(0,1)}]
                \coordinate (origin) at (0,0) ;
                \coordinate (Red-In) at (-0.5,-0.5) ;
                \coordinate (Red-Out) at (0.5,0.5) ;
                \coordinate (Blue-In) at (0,-0.55) ;
                \coordinate (Blue-Out) at (0,0.55) ;
                \draw [thick, red] (Red-In) -- (origin) ;
                \draw [thick, red, enddirected] (origin) -- (Red-Out);
                \draw [thick, blue] (Blue-In) -- (origin) ;
                \draw [thick, blue, enddirected] (origin) -- (Blue-Out);
                \fill ([xshift=-2pt,yshift=-1.5pt]origin) rectangle ++ (4pt,3pt);
            \end{tikzpicture}}}}%
\newcommand{\RedBlueIntChainOne}{\vcenter{\hbox{\ifTikzExternal{\tikzsetnextfilename{RedBlueIntChainOne}}
            \begin{tikzpicture}[scale=1,x={(0.8,0)},y={(0,1)}]
                    \path 
                     (0,-0.35) coordinate (Blue-In)
                     ++(0,0.3) coordinate (Red-Blue-int1)  
                     ++(0,0.4) coordinate (Red-Blue-int2)
                     ++(0,0.35) coordinate (Blue-Out);
                    \coordinate (Red-In) at ($(Red-Blue-int1) + (-0.3,-0.3)$) ;
                    \coordinate (Red-Out) at ($(Red-Blue-int2) + (0.3,0.3)$) ;
                    \draw [thick, red] (Red-In) -- (Red-Blue-int1) ;
                    \draw [thick, red] (Red-Blue-int1) .. controls ($($(Red-Blue-int1)!0.33!(Red-Blue-int2)$) + (0.3,0)$) and ($($(Red-Blue-int1)!0.67!(Red-Blue-int2)$) + (0.3,0)$)  .. (Red-Blue-int2) ;
                    \draw [thick, red, enddirected] (Red-Blue-int2) -- (Red-Out);
                    \draw [thick, blue] (Blue-In) -- (Red-Blue-int1) ;
                    \draw [thick, blue] (Red-Blue-int1) -- (Red-Blue-int2) ;
                    \draw [thick, blue, enddirected] (Red-Blue-int2) -- (Blue-Out);
                    \fill ([xshift=-2pt,yshift=-1.5pt]Red-Blue-int1) rectangle ++ (4pt,3pt);
                    \fill ([xshift=-2pt,yshift=-1.5pt]Red-Blue-int2) rectangle ++ (4pt,3pt);
                \end{tikzpicture}}}%
}
\newcommand{\RedBlueIntChainTwo}{\vcenter{\hbox{\ifTikzExternal{\tikzsetnextfilename{RedBlueIntChainTwo}}
            \begin{tikzpicture}[scale=1,x={(0.8,0)},y={(0,1)}]
                    \path 
                     (0,-0.35) coordinate (Blue-In)
                     ++(0,0.3) coordinate (Red-Blue-int1)  
                     ++(0,0.4) coordinate (Red-Blue-int2)
                     ++(0,0.4) coordinate (Red-Blue-int3)
                     ++(0,0.35) coordinate (Blue-Out);
                    \coordinate (Red-In) at ($(Red-Blue-int1) + (-0.3,-0.3)$) ;
                    \coordinate (Red-Out) at ($(Red-Blue-int3) + (0.3,0.3)$) ;
                    \draw [thick, red] (Red-In) -- (Red-Blue-int1) ;
                    \draw [thick, red] (Red-Blue-int1) .. controls ($($(Red-Blue-int1)!0.33!(Red-Blue-int2)$) + (0.3,0)$) and ($($(Red-Blue-int1)!0.67!(Red-Blue-int2)$) + (0.3,0)$)  .. (Red-Blue-int2) ;
                    \draw [thick, red] (Red-Blue-int2) .. controls ($($(Red-Blue-int2)!0.33!(Red-Blue-int3)$) + (-0.3,0)$) and ($($(Red-Blue-int2)!0.67!(Red-Blue-int3)$) + (-0.3,0)$)  .. (Red-Blue-int3) ;
                    \draw [thick, red, enddirected] (Red-Blue-int3) -- (Red-Out);
                    \draw [thick, blue] (Blue-In) -- (Red-Blue-int1)
                                                  -- (Red-Blue-int2) 
                                                  -- (Red-Blue-int3) ;
                    \draw [thick, blue, enddirected] (Red-Blue-int3) -- (Blue-Out);
                    \fill ([xshift=-2pt,yshift=-1.5pt]Red-Blue-int1) rectangle ++ (4pt,3pt);
                    \fill ([xshift=-2pt,yshift=-1.5pt]Red-Blue-int2) rectangle ++ (4pt,3pt);
                    \fill ([xshift=-2pt,yshift=-1.5pt]Red-Blue-int3) rectangle ++ (4pt,3pt);
                \end{tikzpicture}}}%
}
\newcommand{\BlueBlueInt}{\vcenter{\hbox{\ifTikzExternal{\tikzsetnextfilename{blue-blue}}
            \begin{tikzpicture}[scale=1,x={(0.8,0)},y={(0,1)}]
                   \path 
                    (0,0) coordinate (Blue-In1)
                    ++(0,0.5) coordinate (Blue-Blue-int1)  
                    ++(0.5,0) coordinate (Blue-Blue-int2)
                    ++(0,0.5) coordinate (Blue-Out2);
                   \coordinate (Blue-In2) at ($(Blue-Blue-int2) + (0,-0.5)$) ;
                   \coordinate (Blue-Out1) at ($(Blue-Blue-int1) + (0,0.5)$) ;
                   \draw [thick, blue] (Blue-In1) -- (Blue-Blue-int1) ;
                   \draw [thick, blue, enddirected] (Blue-Blue-int1) -- (Blue-Out1);
                   \draw [thick, blue] (Blue-In2) -- (Blue-Blue-int2) ;
                   \draw [thick, blue, enddirected] (Blue-Blue-int2) -- (Blue-Out2);
                   \draw [thick, black,densely dotted] (Blue-Blue-int1) -- (Blue-Blue-int2) ;
                   \fill ([xshift=-2pt,yshift=-1.5pt]Blue-Blue-int1) rectangle ++ (4pt,3pt);
                   \fill ([xshift=-2pt,yshift=-1.5pt]Blue-Blue-int2) rectangle ++ (4pt,3pt);
                \end{tikzpicture}}}%
}
\newcommand{\BlueBlueIntChainOne}{\vcenter{\hbox{\ifTikzExternal{\tikzsetnextfilename{BlueBlueIntChainOne}}
            \begin{tikzpicture}[scale=1,x={(0.8,0)},y={(0,1)}]
                   \path 
                    (0,0) coordinate (Blue-In1)
                    ++(0,0.3) coordinate (Blue-Blue-int11)  
                    ++(0,0.35) coordinate (Blue-Blue-int12)
                    ++(0,0.35) coordinate (Blue-Out1);
                   \path 
                    ($(0,0)+(0.5,0)$) coordinate (Blue-In2)
                    ++(0,0.3) coordinate (Blue-Blue-int21)  
                    ++(0,0.35) coordinate (Blue-Blue-int22)
                    ++(0,0.35) coordinate (Blue-Out2);
                   \draw [thick, blue] (Blue-In1) -- (Blue-Blue-int11) 
                                                  -- (Blue-Blue-int12) ;
                   \draw [thick, blue, enddirected] (Blue-Blue-int12) -- (Blue-Out1);
                   \draw [thick, blue] (Blue-In2) -- (Blue-Blue-int21) 
                                                  -- (Blue-Blue-int22) ;
                   \draw [thick, blue, enddirected] (Blue-Blue-int22) -- (Blue-Out2);
                   \draw [thick, black,densely dotted] (Blue-Blue-int11) -- (Blue-Blue-int21) ;
                   \draw [thick, black,densely dotted] (Blue-Blue-int12) -- (Blue-Blue-int22) ;
                   \fill ([xshift=-2pt,yshift=-1.5pt]Blue-Blue-int11) rectangle ++ (4pt,3pt);
                   \fill ([xshift=-2pt,yshift=-1.5pt]Blue-Blue-int21) rectangle ++ (4pt,3pt);
                   \fill ([xshift=-2pt,yshift=-1.5pt]Blue-Blue-int12) rectangle ++ (4pt,3pt);
                   \fill ([xshift=-2pt,yshift=-1.5pt]Blue-Blue-int22) rectangle ++ (4pt,3pt);
                \end{tikzpicture}}}%
}
\newcommand{\BlueBlueIntChainTwo}{\vcenter{\hbox{\ifTikzExternal{\tikzsetnextfilename{BlueBlueIntChainTwo}}
            \begin{tikzpicture}[scale=1,x={(0.8,0)},y={(0,1)}]
                   \path 
                    (0,0) coordinate (Blue-In1)
                    ++(0,0.3) coordinate (Blue-Blue-int11)  
                    ++(0,0.35) coordinate (Blue-Blue-int12)
                    ++(0,0.35) coordinate (Blue-Blue-int13)
                    ++(0,0.35) coordinate (Blue-Out1);
                   \path 
                    ($(0,0)+(0.5,0)$) coordinate (Blue-In2)
                    ++(0,0.3) coordinate (Blue-Blue-int21)  
                    ++(0,0.35) coordinate (Blue-Blue-int22)
                    ++(0,0.35) coordinate (Blue-Blue-int23)
                    ++(0,0.35) coordinate (Blue-Out2);
                   \draw [thick, blue] (Blue-In1) -- (Blue-Blue-int11) 
                                                  -- (Blue-Blue-int12) 
                                                  -- (Blue-Blue-int13) ;
                   \draw [thick, blue, enddirected] (Blue-Blue-int13) -- (Blue-Out1);
                   \draw [thick, blue] (Blue-In2) -- (Blue-Blue-int21) 
                                                  -- (Blue-Blue-int22) 
                                                  -- (Blue-Blue-int23) ;
                   \draw [thick, blue, enddirected] (Blue-Blue-int23) -- (Blue-Out2);
                   \draw [thick, black,densely dotted] (Blue-Blue-int11) -- (Blue-Blue-int21) ;
                   \draw [thick, black,densely dotted] (Blue-Blue-int12) -- (Blue-Blue-int22) ;
                   \draw [thick, black,densely dotted] (Blue-Blue-int13) -- (Blue-Blue-int23) ;
                   \fill ([xshift=-2pt,yshift=-1.5pt]Blue-Blue-int11) rectangle ++ (4pt,3pt);
                   \fill ([xshift=-2pt,yshift=-1.5pt]Blue-Blue-int21) rectangle ++ (4pt,3pt);
                   \fill ([xshift=-2pt,yshift=-1.5pt]Blue-Blue-int12) rectangle ++ (4pt,3pt);
                   \fill ([xshift=-2pt,yshift=-1.5pt]Blue-Blue-int22) rectangle ++ (4pt,3pt);
                   \fill ([xshift=-2pt,yshift=-1.5pt]Blue-Blue-int13) rectangle ++ (4pt,3pt);
                   \fill ([xshift=-2pt,yshift=-1.5pt]Blue-Blue-int23) rectangle ++ (4pt,3pt);
                \end{tikzpicture}}}%
}

\newcommand{\interactingVertexExample}{\ifTikzExternal{\tikzsetnextfilename{interactingVertexExample}}\vcenter{\hbox{
\begin{tikzpicture}[scale=1,x={(-1,0)},y={(0,1)}]
		\path 
         (0,0) coordinate (origin)%
         +(-0.07,0) coordinate (blank)%
         +(0,-0.5) coordinate (In-Red)%
         +(0,0.5) coordinate (End-Red)%
         ++(0.4,0) coordinate (Int1)%
         +(0,-0.3) coordinate (Blue1)%
         +(0,-0.5) coordinate (Blue-Dots1)%
         ++(0.4,0) coordinate (Int2)%
         +(0,-0.3) coordinate (Blue2)%
         +(0,-0.5) coordinate (Blue-Dots2)%
         +(0.4,0) coordinate (End-Green)%
        ;
         \draw [thick ,red,directed] (In-Red) -- (origin) ;
         \draw [thick,red,enddirected]  (origin) -- (End-Red);
         \draw [thick,blue,directed] (Blue1) -- (Int1);
         \draw [thick,blue, densely dash dot dot] (Blue-Dots1) -- (Blue1);
         \draw [thick,blue,directed] (Blue2) -- (Int2);
         \draw [thick,blue, densely dash dot dot] (Blue-Dots2) -- (Blue2);
         \draw [thick,Green,directed] (origin) -- (Int1) node [pos=0.5, sloped, font=\footnotesize\bfseries, black ] {||};
         \draw [thick,Green,directed] (Int1) -- (Int2);
         \draw [thick,Green,enddirected] (Int2) -- (End-Green);
         \fill (origin) circle (1.5pt);
     	\fill ([xshift=-2pt,yshift=-1.5pt]Int1) rectangle ++ (4pt,3pt);
     	\fill ([xshift=-2pt,yshift=-1.5pt]Int2) rectangle ++ (4pt,3pt);
    \end{tikzpicture}}}%
}
\newcommand{\interactingVertexGAmmaPlusExample}{\ifTikzExternal{\tikzsetnextfilename{interactingVertexGAmmaPlusExample}}\vcenter{\hbox{
\begin{tikzpicture}[scale=1,x={(-1,0)},y={(0,1)}]
		\path 
         (0,0) coordinate (origin)%
         +(-0.07,0) coordinate (blank)%
         +(0,-0.5) coordinate (In-Red)%
         +(0,0.5) coordinate (End-Red)%
         ++(0.4,0) coordinate (Int1)%
         +(0,-0.3) coordinate (Blue1)%
         +(0,-0.5) coordinate (Blue-Dots1)%
         ++(0.4,0) coordinate (Int2)%
         +(0,-0.3) coordinate (Blue2)%
         +(0,-0.5) coordinate (Blue-Dots2)%
         +(0.4,0) coordinate (End-Green)%
        ;
         \draw [thick ,red,directed] (In-Red) -- (origin)  node [pos=0.5, sloped, font=\footnotesize\bfseries, black ] {||} ;
         \draw [thick,red,enddirected]  (origin) -- (End-Red);
         \draw [thick,blue,directed] (Blue1) -- (Int1);
         \draw [thick,blue, densely dash dot dot] (Blue-Dots1) -- (Blue1);
         \draw [thick,blue,directed] (Blue2) -- (Int2);
         \draw [thick,blue, densely dash dot dot] (Blue-Dots2) -- (Blue2);
         \draw [thick,Green,directed] (origin) -- (Int1);
         \draw [thick,Green,directed] (Int1) -- (Int2);
         \draw [thick,Green,enddirected] (Int2) -- (End-Green);
         \fill (origin) circle (1.5pt);
     	\fill ([xshift=-2pt,yshift=-1.5pt]Int1) rectangle ++ (4pt,3pt);
     	\fill ([xshift=-2pt,yshift=-1.5pt]Int2) rectangle ++ (4pt,3pt);
    \end{tikzpicture}}}%
}
\newcommand{\LplusExample}{\ifTikzExternal{\tikzsetnextfilename{LplusExample}}\vcenter{\hbox{
\begin{tikzpicture}[scale=1,x={(1,0)},y={(0,0.3)}]
		\path 
         (0,0) coordinate (end)%
          +(1,0) coordinate (fantom1)%
         ++(0,-1) coordinate (down)%
         ++(-2.7,0) coordinate (left)%
         ++(0,1) coordinate (beginning)%
          +(-0.1,0) coordinate (fantom2)%
            ;
         \draw[black,line width=0.6pt,->] (end) -- (down) -- (left) -- (beginning) ;
         \fill[white] (fantom1) circle (0.1pt);
         \fill[white] (fantom2) circle (0.1pt);
    \end{tikzpicture}}}%
}
\newcommand{\independetLoopEG}{\vcenter{\hbox{\ifTikzExternal{\tikzsetnextfilename{independetLoopEG}}
\begin{tikzpicture}[scale=1,x={(1,0)},y={(0,1)}]
		\path 
         (0,0) coordinate (origin)  
         ++(-0.4,-0.4) coordinate (In-Red)
         +(-0.2,-0.2) coordinate (In-Dots)
         (origin)++(0.5,0.5) coordinate (Emit-Green1)
         ++(0.6,0) coordinate (interaction1)
         +(0,-0.3) coordinate (In-Blue1)
         +(0,-0.5) coordinate (In-Blue-Dots1)
         (Emit-Green1)++(0.5,0.5) coordinate (Emit-Green2)
         ++(0.6,0) coordinate (interaction2)
         +(0,-0.3) coordinate (In-Blue2)
         +(0,-0.5) coordinate (In-Blue-Dots2)
         (Emit-Green2)++(0.5,0.5) coordinate (Emit-Green3)
         ++(-1.1,0) coordinate (interaction3)
         +(0,-0.3) coordinate (In-Blue3)
         +(0,-0.5) coordinate (In-Blue-Dots3)
         let \p1=(origin),\p2=(Emit-Green3) in (\x1,\y2) coordinate  (interaction4)
         +(-0.3,0) coordinate (Out-Green)
         (Emit-Green3)++(0.3,0.3) coordinate (Out-Red)
         +(0.25,0.25) coordinate (Out-Dots)
            ;
        \coordinate (Der1m) at ($ ($(Emit-Green2)!0.6!(Emit-Green3)$) -  (-0.1,0.1) $) ;
        \coordinate (Der1p) at ($ ($(Emit-Green2)!0.6!(Emit-Green3)$) +  (-0.1,0.1) $) ;
        \coordinate (Der2p) at ($ ($(Emit-Green2)!0.7!(Emit-Green3)$) + (-0.1,0.1) $) ;
        \coordinate (Der2m) at ($ ($(Emit-Green2)!0.7!(Emit-Green3)$) - (-0.1,0.1) $) ;
        \draw [thick,densely dash dot dot, red] (In-Dots) -- (In-Red) ; 
        \draw [thick, red, directed] (In-Red) -- (origin) ; 
        \draw [thick, red, directed] (origin) -- (Emit-Green1) ; 
        \draw [thick, red, directed] (Emit-Green1) -- (Emit-Green2) ; 
        \draw [thick, red, directed] (Emit-Green2) -- (Emit-Green3) ; 
        \draw [thick, red, directed] (Emit-Green3) -- (Out-Red) ; 
        \draw [thick,densely dash dot dot, red] (Out-Red) -- (Out-Dots) ; 
        \draw [thick, blue, directed] (origin) -- (interaction4)  ; 
        \draw [thick, blue, directed] (In-Blue1) -- (interaction1)  ; 
        \draw [thick, blue, densely dash dot dot] (In-Blue-Dots1) -- (In-Blue1)  ; 
        \draw [thick, blue, directed] (In-Blue2) -- (interaction2)  ; 
        \draw [thick, blue, densely dash dot dot] (In-Blue-Dots2) -- (In-Blue2)  ; 
        \draw [thick, blue, directed] (In-Blue3) -- (interaction3)  ; 
        \draw [thick, blue, densely dash dot dot] (In-Blue-Dots3) -- (In-Blue3)  ; 
        \draw [thick, Green, directed] (Emit-Green1) -- (interaction1)  node [pos=0.7, sloped, font=\footnotesize\bfseries, black ] {||} ; 
        \draw [thick, Green, directed] (Emit-Green2) -- (interaction2)  node [pos=0.7, sloped, font=\footnotesize\bfseries, black ] {||} ; 
        \draw [thick, Green, directed] (Emit-Green3) -- (interaction3)  ; 
        \draw [thick, Green, directed] (interaction3) -- (interaction4)  ; 
        \draw [thick, Green, enddirected] (interaction4) -- (Out-Green)  ; 
        \draw [thick,black] (Der1p) -- (Der1m) ;
        \draw [thick,black] (Der2p) -- (Der2m) ;
        \fill (origin) circle (1.2pt);
        \fill (Emit-Green1) circle (1.2pt);
        \fill (Emit-Green2) circle (1.2pt);
        \fill (Emit-Green3) circle (1.2pt);
     	\fill ([xshift=-2pt,yshift=-1.5pt]interaction1) rectangle ++ (4pt,3pt);
     	\fill ([xshift=-2pt,yshift=-1.5pt]interaction2) rectangle ++ (4pt,3pt);
     	\fill ([xshift=-2pt,yshift=-1.5pt]interaction3) rectangle ++ (4pt,3pt);
     	\fill ([xshift=-2pt,yshift=-1.5pt]interaction4) rectangle ++ (4pt,3pt);
     	\node[below right,inner sep=0pt] at (origin) {${}^{\gamma_1^1}$};
     	\node[below right=-0.01 and -0.09,inner sep=0pt] at (Emit-Green1) {${}_{\gamma_2^-}$};
     	\node[below right=-0.01 and -0.09,inner sep=0pt] at (Emit-Green2) {${}_{\gamma_2^-}$};
     	\node[below right=-0.1 and 0.03,inner sep=-1pt] at (Emit-Green3) {${}^{\gamma_2^{+(2)}}$};
    \end{tikzpicture}}}%
}
\newcommand{\independetLoopMinusEG}{\vcenter{\hbox{\ifTikzExternal{\tikzsetnextfilename{independetLoopMinusEG}}
\begin{tikzpicture}[scale=1,x={(1,0)},y={(0,1)}]
		\path 
         (0,0) coordinate (origin)  
         ++(-0.4,-0.4) coordinate (In-Red)
         +(-0.2,-0.2) coordinate (In-Dots)
         (origin)++(0.5,0.5) coordinate (Emit-Green1)
         ++(0.6,0) coordinate (interaction1)
         +(0,-0.3) coordinate (In-Blue1)
         +(0,-0.5) coordinate (In-Blue-Dots1)
         (Emit-Green1)++(0.5,0.5) coordinate (Emit-Green2)
         ++(0.6,0) coordinate (interaction2)
         +(0,-0.3) coordinate (In-Blue2)
         +(0,-0.5) coordinate (In-Blue-Dots2)
         (Emit-Green2)++(0.5,0.5) coordinate (Emit-Green3)
         ++(-1.1,0) coordinate (interaction3)
         +(0,-0.3) coordinate (In-Blue3)
         +(0,-0.5) coordinate (In-Blue-Dots3)
         let \p1=(origin),\p2=(Emit-Green3) in (\x1,\y2) coordinate  (interaction4)
         +(-0.3,0) coordinate (Out-Green)
         (Emit-Green3)++(0.3,0.3) coordinate (Out-Red)
         +(0.25,0.25) coordinate (Out-Dots)
            ;
        \draw [thick,densely dash dot dot, red] (In-Dots) -- (In-Red) ; 
        \draw [thick, red, directed] (In-Red) -- (origin) ; 
        \draw [thick, red, directed] (origin) -- (Emit-Green1) ; 
        \draw [thick, red, directed] (Emit-Green1) -- (Emit-Green2) ; 
        \draw [thick, red, directed] (Emit-Green2) -- (Emit-Green3) ; 
        \draw [thick, red, directed] (Emit-Green3) -- (Out-Red) ; 
        \draw [thick,densely dash dot dot, red] (Out-Red) -- (Out-Dots) ; 
        \draw [thick, blue, directed] (origin) -- (interaction4)  ; 
        \draw [thick, blue, directed] (In-Blue1) -- (interaction1)  ; 
        \draw [thick, blue, densely dash dot dot] (In-Blue-Dots1) -- (In-Blue1)  ; 
        \draw [thick, blue, directed] (In-Blue2) -- (interaction2)  ; 
        \draw [thick, blue, densely dash dot dot] (In-Blue-Dots2) -- (In-Blue2)  ; 
        \draw [thick, blue, directed] (In-Blue3) -- (interaction3)  ; 
        \draw [thick, blue, densely dash dot dot] (In-Blue-Dots3) -- (In-Blue3)  ; 
        \draw [thick, Green, directed] (Emit-Green1) -- (interaction1)  node [pos=0.7, sloped, font=\footnotesize\bfseries, black ] {||} ; 
        \draw [thick, Green, directed] (Emit-Green2) -- (interaction2)  node [pos=0.7, sloped, font=\footnotesize\bfseries, black ] {||} ; 
        \draw [thick, Green, directed] (Emit-Green3) -- (interaction3) node [pos=0.2, sloped, font=\footnotesize\bfseries, black ] {||} ; 
        \draw [thick, Green, directed] (interaction3) -- (interaction4)  ; 
        \draw [thick, Green, enddirected] (interaction4) -- (Out-Green)  ; 
        \fill (origin) circle (1.2pt);
        \fill (Emit-Green1) circle (1.2pt);
        \fill (Emit-Green2) circle (1.2pt);
        \fill (Emit-Green3) circle (1.2pt);
     	\fill ([xshift=-2pt,yshift=-1.5pt]interaction1) rectangle ++ (4pt,3pt);
     	\fill ([xshift=-2pt,yshift=-1.5pt]interaction2) rectangle ++ (4pt,3pt);
     	\fill ([xshift=-2pt,yshift=-1.5pt]interaction3) rectangle ++ (4pt,3pt);
     	\fill ([xshift=-2pt,yshift=-1.5pt]interaction4) rectangle ++ (4pt,3pt);
     	\node[below right,inner sep=0pt] at (origin) {${}^{\gamma_1^1}$};
     	\node[below right=-0.01 and -0.09,inner sep=0pt] at (Emit-Green1) {${}_{\gamma_2^-}$};
     	\node[below right=-0.01 and -0.09,inner sep=0pt] at (Emit-Green2) {${}_{\gamma_2^-}$};
     	\node[below right=-0.1 and 0.03,inner sep=-1pt] at (Emit-Green3) {${}^{\gamma_2^{-(2)}}$};
    \end{tikzpicture}}}%
}
\newcommand{\independetLoopEGlabeled}{\vcenter{\hbox{\ifTikzExternal{\tikzsetnextfilename{independetLoopEGlabeled}}
\begin{tikzpicture}[scale=1,x={(1.1,0)},y={(0,1.1)}]
		\path 
         (0,0) coordinate (origin)  
         ++(-0.3,-0.3) coordinate (In-Red)
         +(-0.2,-0.2) coordinate (In-Dots)
         (origin)++(0.65,0.65) coordinate (Emit-Green1)
         ++(0.4,0) coordinate (interaction1)
         +(0,-0.3) coordinate (In-Blue1)
         +(0,-0.5) coordinate (In-Blue-Dots1)
         (Emit-Green1)++(0.65,0.65) coordinate (Emit-Green2)
         ++(0.4,0) coordinate (interaction2)
         +(0,-0.3) coordinate (In-Blue2)
         +(0,-0.5) coordinate (In-Blue-Dots2)
         (Emit-Green2)++(0.7,0.7) coordinate (Emit-Green3)
         ++(-0.65,0) coordinate (interaction3)
         +(0,-0.3) coordinate (In-Blue3)
         +(0,-0.5) coordinate (In-Blue-Dots3)
         ++(-0.7,0) coordinate (interaction4)
         +(0,-0.3) coordinate (In-Blue4)
         +(0,-0.5) coordinate (In-Blue-Dots4)
          let \p1=(origin),\p2=(Emit-Green3) in (\x1,\y2) coordinate  (interaction5)
         +(-0.3,0) coordinate (Out-Green)
         (Emit-Green3)++(0.3,0.3) coordinate (Out-Red)
         +(0.25,0.25) coordinate (Out-Dots)
            ;
        \coordinate (Der5p) at ($ ($(interaction1)!0.4!(Emit-Green1)$) + (-0.0,0.12) $) ;
        \coordinate (Der5m) at ($ ($(interaction1)!0.4!(Emit-Green1)$) - (-0.0,0.12) $) ;
        \coordinate (Der6p) at ($ ($(interaction1)!0.5!(Emit-Green1)$) + (-0.0,0.12) $) ;
        \coordinate (Der6m) at ($ ($(interaction1)!0.5!(Emit-Green1)$) - (-0.0,0.12) $) ;
        \coordinate (Der3p) at ($ ($(interaction2)!0.4!(Emit-Green2)$) + (-0.0,0.12) $) ;
        \coordinate (Der3m) at ($ ($(interaction2)!0.4!(Emit-Green2)$) - (-0.0,0.12) $) ;
        \coordinate (Der4p) at ($ ($(interaction2)!0.5!(Emit-Green2)$) + (-0.0,0.12) $) ;
        \coordinate (Der4m) at ($ ($(interaction2)!0.5!(Emit-Green2)$) - (-0.0,0.12) $) ;
        \coordinate (Der1m) at ($ ($(Emit-Green2)!0.6!(Emit-Green3)$) -  (-0.1,0.1) $) ;
        \coordinate (Der1p) at ($ ($(Emit-Green2)!0.6!(Emit-Green3)$) +  (-0.1,0.1) $) ;
        \coordinate (Der2p) at ($ ($(Emit-Green2)!0.7!(Emit-Green3)$) + (-0.1,0.1) $) ;
        \coordinate (Der2m) at ($ ($(Emit-Green2)!0.7!(Emit-Green3)$) - (-0.1,0.1) $) ;
        \draw [thick,densely dash dot dot, red] (In-Dots) -- (In-Red) ; 
        \draw [thick, red, directed] (In-Red) -- (origin) ; 
        \draw [thick, red, directed] (origin) -- (Emit-Green1) ; 
        \draw [thick, red,  dash dot dot] (Emit-Green1) -- (Emit-Green2) ; 
        \draw [thick, red, directed] (Emit-Green2) -- (Emit-Green3) ; 
        \draw [thick, red, directed] (Emit-Green3) -- (Out-Red) ; 
        \draw [thick,densely dash dot dot, red] (Out-Red) -- (Out-Dots) ; 
        \draw [thick, blue, directed] (origin) -- (interaction5)  ; 
        \draw [thick, blue, directed] (In-Blue1) -- (interaction1)  ; 
        \draw [thick, blue, densely dash dot dot] (In-Blue-Dots1) -- (In-Blue1)  ; 
        \draw [thick, blue, directed] (In-Blue2) -- (interaction2)  ; 
        \draw [thick, blue, densely dash dot dot] (In-Blue-Dots2) -- (In-Blue2)  ; 
        \draw [thick, blue, directed] (In-Blue3) -- (interaction3)  ; 
        \draw [thick, blue, densely dash dot dot] (In-Blue-Dots3) -- (In-Blue3)  ; 
        \draw [thick, blue, directed] (In-Blue4) -- (interaction4)  ; 
        \draw [thick, blue, densely dash dot dot] (In-Blue-Dots4) -- (In-Blue4)  ; 
        \draw [thick, Green, directed] (Emit-Green1) -- (interaction1)  ; 
        \draw [thick, Green, directed] (Emit-Green2) -- (interaction2)  ; 
        \draw [thick, Green, directed] (Emit-Green3) -- (interaction3)  ; 
        \draw [thick, Green, densely dash dot dot] (interaction3) -- (interaction4)  ; 
        \draw [thick, Green, directed] (interaction4) -- (interaction5)  ; 
        \draw [thick, Green, enddirected] (interaction5) -- (Out-Green)  ; 
        \draw [thick,black] (Der1p) -- (Der1m) ;
        \draw [thick,black] (Der2p) -- (Der2m) ;
        \draw [thick,black] (Der3p) -- (Der3m) ;
        \draw [thick,black] (Der4p) -- (Der4m) ;
        \draw [thick,black] (Der5p) -- (Der5m) ;
        \draw [thick,black] (Der6p) -- (Der6m) ;
        \fill (origin) circle (1.2pt);
        \fill (Emit-Green1) circle (1.2pt);
        \fill (Emit-Green2) circle (1.2pt);
        \fill (Emit-Green3) circle (1.2pt);
     	\fill ([xshift=-2pt,yshift=-1.5pt]interaction1) rectangle ++ (4pt,3pt);
     	\fill ([xshift=-2pt,yshift=-1.5pt]interaction2) rectangle ++ (4pt,3pt);
     	\fill ([xshift=-2pt,yshift=-1.5pt]interaction3) rectangle ++ (4pt,3pt);
     	\fill ([xshift=-2pt,yshift=-1.5pt]interaction4) rectangle ++ (4pt,3pt);
     	\fill ([xshift=-2pt,yshift=-1.5pt]interaction5) rectangle ++ (4pt,3pt);
     	\node[left,inner sep=0.2pt] at (In-Blue-Dots1) {${}^1$};
     	\node[left,inner sep=0.2pt] at (In-Blue-Dots2) {${}^k$};
     	\node[left,inner sep=0.2pt] at (In-Blue-Dots3) {${}^{k{+}1}$};
     	\node[left,inner sep=0.2pt] at (In-Blue-Dots4) {${}^{k{+}n}$};
    \end{tikzpicture}}}%
}
\newcommand{\independetLoopEGsimp}{\vcenter{\hbox{\ifTikzExternal{\tikzsetnextfilename{independetLoopEGsimp}}
\begin{tikzpicture}[scale=1,x={(1,0)},y={(0,1)}]
		\path 
         (0,0) coordinate (origin)  
         ++(-0.3,-0.3) coordinate (In-Red)
         +(-0.2,-0.2) coordinate (In-Dots)
         (origin)++(0.9,0.9) coordinate (interaction1)
         +(0,-0.3) coordinate (In-Blue1)
         +(0,-0.5) coordinate (In-Blue-Dots1)
         (interaction1)++(1,1) coordinate (interaction2)
         +(0,-0.3) coordinate (In-Blue2)
         +(0,-0.5) coordinate (In-Blue-Dots2)
         (interaction2)++(-0.6,0) coordinate (interaction3)
         +(0,-0.3) coordinate (In-Blue3)
         +(0,-0.5) coordinate (In-Blue-Dots3)
         ++(-0.65,0) coordinate (interaction4)
         +(0,-0.3) coordinate (In-Blue4)
         +(0,-0.5) coordinate (In-Blue-Dots4)
         ++(-0.65,0) coordinate (interaction5)
         +(-0.3,0) coordinate (Out-Green)
         (interaction2)++(0.3,0.3) coordinate (Out-Red)
         +(0.25,0.25) coordinate (Out-Dots)
            ;
        \draw [thick,densely dash dot dot, red] (In-Dots) -- (In-Red) ; 
        \draw [thick, red, directed] (In-Red) -- (origin)  ;
        \draw [thick, red, directed] (origin) -- (interaction1); 
        \draw [thick, red, dash dot dot] (interaction1) -- (interaction2) ; 
        \draw [thick, red, directed] (interaction2) -- (Out-Red) ; 
        \draw [thick,densely dash dot dot, red] (Out-Red) -- (Out-Dots) ; 
        \draw [thick, blue, directed] (origin) -- (interaction5)  ; 
        \draw [thick, blue, directed] (In-Blue1) -- (interaction1)  ; 
        \draw [thick, blue, densely dash dot dot] (In-Blue-Dots1) -- (In-Blue1)  ; 
        \draw [thick, blue, directed] (In-Blue2) -- (interaction2)  ; 
        \draw [thick, blue, densely dash dot dot] (In-Blue-Dots2) -- (In-Blue2)  ; 
        \draw [thick, blue, directed] (In-Blue3) -- (interaction3)  ; 
        \draw [thick, blue, densely dash dot dot] (In-Blue-Dots3) -- (In-Blue3)  ; 
        \draw [thick, blue, directed] (In-Blue4) -- (interaction4)  ; 
        \draw [thick, blue, densely dash dot dot] (In-Blue-Dots4) -- (In-Blue4)  ; 
        \draw [thick, Green, directed] (interaction2) -- (interaction3)  ; 
        \draw [thick, Green, densely dash dot dot] (interaction3) -- (interaction4)  ; 
        \draw [thick, Green, directed] (interaction4) -- (interaction5)  ; 
        \draw [thick, Green, enddirected] (interaction5) -- (Out-Green)  ; 
        \fill (origin) circle (1.2pt);
     	\fill ([xshift=-2pt,yshift=-1.5pt]interaction1) rectangle ++ (4pt,3pt);
     	\fill ([xshift=-2pt,yshift=-1.5pt]interaction2) rectangle ++ (4pt,3pt);
     	\fill ([xshift=-2pt,yshift=-1.5pt]interaction3) rectangle ++ (4pt,3pt);
     	\fill ([xshift=-2pt,yshift=-1.5pt]interaction4) rectangle ++ (4pt,3pt);
     	\fill ([xshift=-2pt,yshift=-1.5pt]interaction5) rectangle ++ (4pt,3pt);
     	\node[left,inner sep=0.2pt] at (In-Blue-Dots1) {${}^{1}$};
     	\node[left,inner sep=0.2pt] at (In-Blue-Dots2) {${}^{k}$};
     	\node[left,inner sep=0.2pt] at (In-Blue-Dots3) {${}^{k{+}1}$};
     	\node[left,inner sep=0.2pt] at (In-Blue-Dots4) {${}^{k{+}n}$};
    \end{tikzpicture}}}%
}
\newcommand{\independetLoopEGsimpIntegrated}{\vcenter{\hbox{\ifTikzExternal{\tikzsetnextfilename{independetLoopEGsimpIntegrated}}
\begin{tikzpicture}[scale=1,x={(1,0)},y={(0,1)}]
		\path 
        (0,0) coordinate (origin)  
        ++(-0.3,-0.3) coordinate (In-Red)
         +(-0.2,-0.2) coordinate (In-Dots)
        (origin)++(45:0.6) coordinate (center)
        ++ (-45:0.6) coordinate (interaction1)
         +(0,-0.3) coordinate (In-Blue1)
         +(0,-0.5) coordinate (In-Blue-Dots1)
        (center) ++ (45:0.6) coordinate (interaction2)
         +(0,-0.3) coordinate (In-Blue2)
         +(0,-0.5) coordinate (In-Blue-Dots2)
        (center) ++ (85:0.6)  coordinate (interaction3)
         +(0,-0.3) coordinate (In-Blue3)
         +(0,-0.5) coordinate (In-Blue-Dots3)
        (center) ++ (140:0.6) coordinate (interaction4)
        ++ (-0.3,0) coordinate (interaction4bis)
         +(0,-0.3) coordinate (In-Blue4)
         +(0,-0.5) coordinate (In-Blue-Dots4)
        (origin) + (-0.3,0) coordinate (Out-Green)
        (interaction2)++(0.3,0.3) coordinate (Out-Red)
         +(0.25,0.25) coordinate (Out-Dots)
            ;
        \draw [thick,densely dash dot dot, red] (In-Dots) -- (In-Red) ; 
        \draw [thick, red, directed] (In-Red) -- (origin)  ;
        \drawarc[thick, red, directed]{center}{origin}{interaction1} ;
        \drawarc[thick, red, dash dot dot]{center}{interaction1}{interaction2} ;
        \draw [thick, red, directed] (interaction2) -- (Out-Red) ; 
        \draw [thick,densely dash dot dot, red] (Out-Red) -- (Out-Dots) ; 
        \draw [thick, blue, directed] (In-Blue1) -- (interaction1)  ; 
        \draw [thick, blue, densely dash dot dot] (In-Blue-Dots1) -- (In-Blue1)  ; 
        \draw [thick, blue, directed] (In-Blue2) -- (interaction2)  ; 
        \draw [thick, blue, densely dash dot dot] (In-Blue-Dots2) -- (In-Blue2)  ; 
        \draw [thick, blue, directed] (In-Blue3) -- (interaction3)  ; 
        \draw [thick, blue, densely dash dot dot] (In-Blue-Dots3) -- (In-Blue3)  ; 
        \draw [thick, blue, directed] (In-Blue4) -- (interaction4bis)  ; 
        \draw [thick, blue, densely dash dot dot] (In-Blue-Dots4) -- (In-Blue4)  ; 
        \drawarc [thick, Green, directed] {center}{interaction2}{interaction3} ;  
        \drawarc [thick, Green, dash dot dot] {center}{interaction3}{interaction4}  ; 
        \drawarc [thick, Green, directed] {center}{interaction4}{origin} ;  
        \draw [thick, Green, enddirected] (origin) -- (Out-Green)  ; 
        \draw [thick,black,densely dotted] (interaction4) -- (interaction4bis) ;
     	\fill ([xshift=-2pt,yshift=-1.5pt]interaction1) rectangle ++ (4pt,3pt);
     	\fill ([xshift=-2pt,yshift=-1.5pt]interaction2) rectangle ++ (4pt,3pt);
     	\fill ([xshift=-2pt,yshift=-1.5pt]interaction3) rectangle ++ (4pt,3pt);
     	\fill ([xshift=-2pt,yshift=-1.5pt]interaction4) rectangle ++ (4pt,3pt);
     	\fill ([xshift=-2pt,yshift=-1.5pt]interaction4bis) rectangle ++ (4pt,3pt);
     	\fill ([xshift=-2pt,yshift=-1.5pt]origin) rectangle ++ (4pt,3pt);
     	\node[left,inner sep=0.2pt] at (In-Blue-Dots1) {${}^{1}$};
     	\node[left,inner sep=0.2pt] at (In-Blue-Dots2) {${}^{k}$};
     	\node[left,inner sep=0.2pt] at (In-Blue-Dots3) {${}^{k{+}1}$};
     	\node[left,inner sep=0.2pt] at (In-Blue-Dots4) {${}^{k{+}n}$}; 
    \end{tikzpicture}}}%
}
\newcommand{\independetLoopMinusEGlabeled}{\vcenter{\hbox{\ifTikzExternal{\tikzsetnextfilename{independetLoopMinusEGlabeled}}
\begin{tikzpicture}[scale=1,x={(1,0)},y={(0,1)}]
		\path 
         (0,0) coordinate (origin)  
         ++(-0.3,-0.3) coordinate (In-Red)
         +(-0.2,-0.2) coordinate (In-Dots)
         (origin)++(0.6,0.6) coordinate (Emit-Green1)
         ++(0.4,0) coordinate (interaction1)
         +(0,-0.3) coordinate (In-Blue1)
         +(0,-0.5) coordinate (In-Blue-Dots1)
         (Emit-Green1)++(0.7,0.7) coordinate (Emit-Green2)
         ++(0.4,0) coordinate (interaction2)
         +(0,-0.3) coordinate (In-Blue2)
         +(0,-0.5) coordinate (In-Blue-Dots2)
         (Emit-Green2)++(0.7,0.7) coordinate (Emit-Green3)
         ++(-0.75,0) coordinate (interaction3)
         +(0,-0.3) coordinate (In-Blue3)
         +(0,-0.5) coordinate (In-Blue-Dots3)
         ++(-0.55,0) coordinate (interaction4)
         +(0,-0.3) coordinate (In-Blue4)
         +(0,-0.5) coordinate (In-Blue-Dots4)
         ++(-0.7,0) coordinate (interaction5)
         +(-0.3,0) coordinate (Out-Green)
         (Emit-Green3)++(0.3,0.3) coordinate (Out-Red)
         +(0.25,0.25) coordinate (Out-Dots)
            ;
        \coordinate (Der1p) at ($ ($(interaction1)!0.4!(Emit-Green1)$) + (-0.0,0.12) $) ;
        \coordinate (Der1m) at ($ ($(interaction1)!0.4!(Emit-Green1)$) - (-0.0,0.12) $) ;
        \coordinate (Der2p) at ($ ($(interaction1)!0.5!(Emit-Green1)$) + (-0.0,0.12) $) ;
        \coordinate (Der2m) at ($ ($(interaction1)!0.5!(Emit-Green1)$) - (-0.0,0.12) $) ;
        \coordinate (Der3p) at ($ ($(interaction2)!0.4!(Emit-Green2)$) + (-0.0,0.12) $) ;
        \coordinate (Der3m) at ($ ($(interaction2)!0.4!(Emit-Green2)$) - (-0.0,0.12) $) ;
        \coordinate (Der4p) at ($ ($(interaction2)!0.5!(Emit-Green2)$) + (-0.0,0.12) $) ;
        \coordinate (Der4m) at ($ ($(interaction2)!0.5!(Emit-Green2)$) - (-0.0,0.12) $) ;
        \coordinate (Der5p) at ($ ($(interaction3)!0.6!(Emit-Green3)$) + (-0.0,0.12) $) ;
        \coordinate (Der5m) at ($ ($(interaction3)!0.6!(Emit-Green3)$) - (-0.0,0.12) $) ;
        \coordinate (Der6p) at ($ ($(interaction3)!0.5!(Emit-Green3)$) + (-0.0,0.12) $) ;
        \coordinate (Der6m) at ($ ($(interaction3)!0.5!(Emit-Green3)$) - (-0.0,0.12) $) ;
        \draw [thick,densely dash dot dot, red] (In-Dots) -- (In-Red) ; 
        \draw [thick, red, directed] (In-Red) -- (origin) ; 
        \draw [thick, red, directed] (origin) -- (Emit-Green1) ; 
        \draw [thick, red, dash dot dot] (Emit-Green1) -- (Emit-Green2) ; 
        \draw [thick, red, directed] (Emit-Green2) -- (Emit-Green3) ; 
        \draw [thick, red, directed] (Emit-Green3) -- (Out-Red) ; 
        \draw [thick,densely dash dot dot, red] (Out-Red) -- (Out-Dots) ; 
        \draw [thick, blue, directed] (origin) -- (interaction5)  ; 
        \draw [thick, blue, directed] (In-Blue1) -- (interaction1)  ; 
        \draw [thick, blue, densely dash dot dot] (In-Blue-Dots1) -- (In-Blue1)  ; 
        \draw [thick, blue, directed] (In-Blue2) -- (interaction2)  ; 
        \draw [thick, blue, densely dash dot dot] (In-Blue-Dots2) -- (In-Blue2)  ; 
        \draw [thick, blue, directed] (In-Blue3) -- (interaction3)  ; 
        \draw [thick, blue, densely dash dot dot] (In-Blue-Dots3) -- (In-Blue3)  ; 
        \draw [thick, blue, directed] (In-Blue4) -- (interaction4)  ; 
        \draw [thick, blue, densely dash dot dot] (In-Blue-Dots4) -- (In-Blue4)  ; 
        \draw [thick, Green, directed] (Emit-Green1) -- (interaction1)  ; 
        \draw [thick, Green, directed] (Emit-Green2) -- (interaction2)  ; 
        \draw [thick, Green, directed] (Emit-Green3) -- (interaction3)  ; 
        \draw [thick, Green, dash dot dot] (interaction3) -- (interaction4)  ; 
        \draw [thick, Green, directed] (interaction4) -- (interaction5)  ; 
        \draw [thick, Green, enddirected] (interaction5) -- (Out-Green)  ; 
        \draw [thick,black] (Der1p) -- (Der1m) ;
        \draw [thick,black] (Der2p) -- (Der2m) ;
        \draw [thick,black] (Der3p) -- (Der3m) ;
        \draw [thick,black] (Der4p) -- (Der4m) ;
        \draw [thick,black,double] (Der5p) -- (Der5m) ;
        \fill (origin) circle (1.2pt);
        \fill (Emit-Green1) circle (1.2pt);
        \fill (Emit-Green2) circle (1.2pt);
        \fill (Emit-Green3) circle (1.2pt);
     	\fill ([xshift=-2pt,yshift=-1.5pt]interaction1) rectangle ++ (4pt,3pt);
     	\fill ([xshift=-2pt,yshift=-1.5pt]interaction2) rectangle ++ (4pt,3pt);
     	\fill ([xshift=-2pt,yshift=-1.5pt]interaction3) rectangle ++ (4pt,3pt);
     	\fill ([xshift=-2pt,yshift=-1.5pt]interaction4) rectangle ++ (4pt,3pt);
     	\fill ([xshift=-2pt,yshift=-1.5pt]interaction5) rectangle ++ (4pt,3pt);
     	\node[left,inner sep=0.2pt] at (In-Blue-Dots1) {${}^{1}$};
     	\node[left,inner sep=0.2pt] at (In-Blue-Dots2) {${}^{k{-}1}$};
     	\node[left,inner sep=0.2pt] at (In-Blue-Dots3) {${}^{k}$};
     	\node[left,inner sep=0.2pt] at (In-Blue-Dots4) {${}^{k{+}n}$};
    \end{tikzpicture}}}%
}
\newcommand{\independetLoopMinusEGsimp}{\vcenter{\hbox{\ifTikzExternal{\tikzsetnextfilename{independetLoopMinusEGsimp}}
\begin{tikzpicture}[scale=1,x={(1,0)},y={(0,1)}]
		\path 
        (0,0) coordinate (origin)  %
        ++(-0.3,-0.3) coordinate (In-Red)%
         +(-0.2,-0.2) coordinate (In-Dots)%
        (origin)++(0.5,0.5) coordinate (interaction1)%
         +(0,-0.3) coordinate (In-Blue1)%
         +(0,-0.5) coordinate (In-Blue-Dots1)%
        (interaction1)++(0.4,0.4) coordinate (interaction2)%
         +(0,-0.3) coordinate (In-Blue2)%
         +(0,-0.5) coordinate (In-Blue-Dots2)%
        (interaction2)++(0.5,0.5) coordinate (interaction3)%
         +(0,-0.3) coordinate (In-Blue3)%
         +(0,-0.5) coordinate (In-Blue-Dots3)%
        ++(-0.7,0) coordinate (interaction4)%
         +(0,-0.3) coordinate (In-Blue4)%
         +(0,-0.5) coordinate (In-Blue-Dots4)%
        ++(-0.7,0) coordinate (interaction5)%
         +(-0.3,0) coordinate (Out-Green)%
        (interaction3)++(0.3,0.3) coordinate (Out-Red)%
         +(0.25,0.25) coordinate (Out-Dots)%
        ;
        \draw [thick,densely dash dot dot, red] (In-Dots) -- (In-Red) ; 
        \draw [thick, red, directed] (In-Red) -- (origin) ;
        \draw [thick, red, directed] (origin) -- (interaction1); 
        \draw [thick, red, dash dot dot] (interaction1) -- (interaction2) ; 
        \draw [thick, red, directed] (interaction2) -- (interaction3) ; 
        \draw [thick, red, directed] (interaction3) -- (Out-Red) ; 
        \draw [thick,densely dash dot dot, red] (Out-Red) -- (Out-Dots) ; 
        \draw [thick, blue, directed] (origin) -- (interaction5)  ; 
        \draw [thick, blue, directed] (In-Blue1) -- (interaction1)  ; 
        \draw [thick, blue, densely dash dot dot] (In-Blue-Dots1) -- (In-Blue1)  ; 
        \draw [thick, blue, directed] (In-Blue2) -- (interaction2)  ; 
        \draw [thick, blue, densely dash dot dot] (In-Blue-Dots2) -- (In-Blue2)  ; 
        \draw [thick, blue, directed] (In-Blue3) -- (interaction3)  ; 
        \draw [thick, blue, densely dash dot dot] (In-Blue-Dots3) -- (In-Blue3)  ; 
        \draw [thick, blue, directed] (In-Blue4) -- (interaction4)  ; 
        \draw [thick, blue, densely dash dot dot] (In-Blue-Dots4) -- (In-Blue4)  ; 
        \draw [thick, Green, dash dot dot] (interaction3) -- (interaction4)  ; 
        \draw [thick, Green, directed] (interaction4) -- (interaction5)  ; 
        \draw [thick, Green, enddirected] (interaction5) -- (Out-Green)  ; 
        \fill (origin) circle (1.2pt);
     	\fill ([xshift=-2pt,yshift=-1.5pt]interaction1) rectangle ++ (4pt,3pt);
     	\fill ([xshift=-2pt,yshift=-1.5pt]interaction2) rectangle ++ (4pt,3pt);
     	\fill ([xshift=-2pt,yshift=-1.5pt]interaction3) rectangle ++ (4pt,3pt);
     	\fill ([xshift=-2pt,yshift=-1.5pt]interaction4) rectangle ++ (4pt,3pt);
     	\fill ([xshift=-2pt,yshift=-1.5pt]interaction5) rectangle ++ (4pt,3pt);
     	\node[left,inner sep=0.2pt] at (In-Blue-Dots1) {${}^{1}$};
     	\node[right,inner sep=0pt] at (In-Blue-Dots2) {$\,{}^{k{-}1}$};
     	\node[left,inner sep=0.2pt] at (In-Blue-Dots3) {${}^{k}$};
     	\node[left,inner sep=0.2pt] at (In-Blue-Dots4) {${}^{k{+}n}$};
    \end{tikzpicture}}}%
}
\newcommand{\independetLoopMinusEGsimpIntegrated}{\vcenter{\hbox{\ifTikzExternal{\tikzsetnextfilename{independetLoopMinusEGsimpIntegrated}}
\begin{tikzpicture}[scale=1,x={(1,0)},y={(0,1)}]
		\path 
        (0,0) coordinate (origin)  
        ++(-0.3,-0.3) coordinate (In-Red)
         +(-0.2,-0.2) coordinate (In-Dots)
        (origin)++(45:0.6) coordinate (center)
        ++ (-90:0.6) coordinate (interaction1)
         +(0,-0.3) coordinate (In-Blue1)
         +(0,-0.5) coordinate (In-Blue-Dots1)
        (center) ++ (-45:0.6) coordinate (interaction2)
         +(0,-0.3) coordinate (In-Blue2)
         +(0,-0.5) coordinate (In-Blue-Dots2)
        (center) ++ (45:0.6)  coordinate (interaction3)
         +(0,-0.3) coordinate (In-Blue3)
         +(0,-0.5) coordinate (In-Blue-Dots3)
        (center) ++ (135:0.6) coordinate (interaction4)
        ++ (-0.3,0) coordinate (interaction4bis)
         +(0,-0.3) coordinate (In-Blue4)
         +(0,-0.5) coordinate (In-Blue-Dots4)
        (origin) + (-0.3,0) coordinate (Out-Green)
        (interaction3)++(0.3,0.3) coordinate (Out-Red)
         +(0.25,0.25) coordinate (Out-Dots)
            ;
        \draw [thick,densely dash dot dot, red] (In-Dots) -- (In-Red) ; 
        \draw [thick, red, directed] (In-Red) -- (origin)  ;
        \drawarc[thick, red, directed]{center}{origin}{interaction1} ;
        \drawarc[thick, red, dash dot dot]{center}{interaction1}{interaction2} ;
        \drawarc[thick, red, directed]{center}{interaction2}{interaction3} ;
        \draw [thick, red, directed] (interaction3) -- (Out-Red) ; 
        \draw [thick,densely dash dot dot, red] (Out-Red) -- (Out-Dots) ; 
        \draw [thick, blue, directed] (In-Blue1) -- (interaction1)  ; 
        \draw [thick, blue, densely dash dot dot] (In-Blue-Dots1) -- (In-Blue1)  ; 
        \draw [thick, blue, directed] (In-Blue2) -- (interaction2)  ; 
        \draw [thick, blue, densely dash dot dot] (In-Blue-Dots2) -- (In-Blue2)  ; 
        \draw [thick, blue, directed] (In-Blue3) -- (interaction3)  ; 
        \draw [thick, blue, densely dash dot dot] (In-Blue-Dots3) -- (In-Blue3)  ; 
        \draw [thick, blue, directed] (In-Blue4) -- (interaction4bis)  ; 
        \draw [thick, blue, densely dash dot dot] (In-Blue-Dots4) -- (In-Blue4)  ; 
        \drawarc [thick, Green, dash dot dot] {center}{interaction3}{interaction4}  ; 
        \drawarc [thick, Green, directed] {center}{interaction4}{origin} ;  
        \draw [thick, Green, enddirected] (origin) -- (Out-Green)  ; 
        \draw [thick,black,densely dotted] (interaction4) -- (interaction4bis) ;
     	\fill ([xshift=-2pt,yshift=-1.5pt]interaction1) rectangle ++ (4pt,3pt);
     	\fill ([xshift=-2pt,yshift=-1.5pt]interaction2) rectangle ++ (4pt,3pt);
     	\fill ([xshift=-2pt,yshift=-1.5pt]interaction3) rectangle ++ (4pt,3pt);
     	\fill ([xshift=-2pt,yshift=-1.5pt]interaction4) rectangle ++ (4pt,3pt);
     	\fill ([xshift=-2pt,yshift=-1.5pt]interaction4bis) rectangle ++ (4pt,3pt);
     	\fill ([xshift=-2pt,yshift=-1.5pt]origin) rectangle ++ (4pt,3pt);
     	\node[left,inner sep=0.2pt] at (In-Blue-Dots1) {${}^{1}$};
     	\node[right,inner sep=0pt] at (In-Blue-Dots2) {$\,{}^{k{-}1}$};
     	\node[left,inner sep=0.2pt] at (In-Blue-Dots3) {${}^{k}$};
     	\node[left,inner sep=0.2pt] at (In-Blue-Dots4) {${}^{k{+}n}$};
    \end{tikzpicture}}}%
}
\newcommand{\GammaTwointeractingPrecedingGammaOnebis}{\vcenter{\hbox{\ifTikzExternal{\tikzsetnextfilename{Gamma2interactingPrecedingGamma1bis}}
\begin{tikzpicture}[scale=1,x={(1,0)},y={(0,1)}]
		\path 
         (0,0) coordinate (origin)  
         ++(-0.4,-0.4) coordinate (In-Red)
         +(-0.2,-0.2) coordinate (In-Dots)
         (origin)++(1.5,1.5) coordinate (Emit-Green)
         ++(-0.6,0) coordinate (interaction1)
         +(0,-0.3) coordinate (In-Blue1)
         +(0,-0.5) coordinate (In-Blue-Dots1)
         ++(-0.45,0) coordinate (interaction2)
         +(0,-0.3) coordinate (In-Blue2)
         +(0,-0.5) coordinate (In-Blue-Dots2)
         ++(-0.45,0) coordinate (interaction3)
         +(-0.3,0) coordinate (Out-Green)
         (Emit-Green)++(0.3,0.3) coordinate (Out-Red)
         +(0.25,0.25) coordinate (Out-Dots)
            ;
        \coordinate (Der1p) at ($ ($(origin)!0.85!(Emit-Green)$) + (-0.07,0.07) $) ;
        \coordinate (Der1m) at ($ ($(origin)!0.85!(Emit-Green)$) - (-0.07,0.07) $) ;
        \coordinate (Der2p) at ($ ($(origin)!0.9!(Emit-Green)$) + (-0.07,0.07) $) ;
        \coordinate (Der2m) at ($ ($(origin)!0.9!(Emit-Green)$) - (-0.07,0.07) $) ;
        \draw [thick,densely dash dot dot, red] (In-Dots) -- (In-Red) ; 
        \draw [thick, red, directed] (In-Red) -- (origin) ; 
        \draw [thick, red, directed] (origin) -- (Emit-Green) ; 
        \draw [thick, red, directed] (Emit-Green) -- (Out-Red) ; 
        \draw [thick,densely dash dot dot, red] (Out-Red) -- (Out-Dots) ; 
        \draw [thick, blue, directed] (origin) -- (interaction3)  ; 
        \draw [thick, blue, directed] (In-Blue1) -- (interaction1)  ; 
        \draw [thick, blue, densely dash dot dot] (In-Blue-Dots1) -- (In-Blue1)  ; 
        \draw [thick, blue, directed] (In-Blue2) -- (interaction2)  ; 
        \draw [thick, blue, densely dash dot dot] (In-Blue-Dots2) -- (In-Blue2)  ; 
        \draw [thick, Green, directed] (Emit-Green) -- (interaction1)  ; 
        \draw [thick, dash dot dot, Green] (interaction1) -- (interaction2)  ; 
        \draw [thick, Green, directed] (interaction2) -- (interaction3)  ; 
        \draw [thick, Green, enddirected] (interaction3) -- (Out-Green)  ; 
        \draw [thick,black] (Der1p) -- (Der1m) ;
        \draw [thick,black] (Der2p) -- (Der2m) ;
        \fill (origin) circle (1.2pt);
        \fill (Emit-Green) circle (1.2pt);
     	\fill ([xshift=-2pt,yshift=-1.5pt]interaction1) rectangle ++ (4pt,3pt);
     	\fill ([xshift=-2pt,yshift=-1.5pt]interaction2) rectangle ++ (4pt,3pt);
     	\fill ([xshift=-2pt,yshift=-1.5pt]interaction3) rectangle ++ (4pt,3pt);
     	\node[left,inner sep=0.2pt] at (In-Blue-Dots1) {${}^{1}$};
     	\node[left,inner sep=0.2pt] at (In-Blue-Dots2) {${}^{n}$};
    \end{tikzpicture}}}%
}
\newcommand{\GammaTwointeractingPrecedingGammaOneIntegratedbis}{\vcenter{\hbox{\ifTikzExternal{\tikzsetnextfilename{Gamma2interactingPrecedingGamma1integratedbis}}
\begin{tikzpicture}[scale=1,x={(1,0)},y={(0,1)}]
		\path 
         (0,0) coordinate (origin)%
         ++(0.4,-0.4) coordinate (In-Red)%
         +(0.2,-0.2) coordinate (In-Dots)%
         ($(origin)+(-0.5,0)$) ++(120:0.5) coordinate (interaction1)%
          +(0,-0.3) coordinate (In-Blue1)%
          +(0,-0.5) coordinate (In-Blue-Dots1)%
          +(0,0.5) coordinate (blank)%
         ($(origin)+(-0.5,0)$) ++(240:0.5) coordinate (interaction2)%
         +(0,-0.3) coordinate (In-Blue2)%
         +(0,-0.5) coordinate (In-Blue-Dots2)%
         (origin) +(-0.3,0) coordinate (Out-Green)%
         (origin) ++(0.4,0.4) coordinate (Out-Red)%
         +(0.25,0.25) coordinate (Out-Dots)%
        ;
        \draw [thick,densely dash dot dot, red] (In-Dots) -- (In-Red) ; 
        \draw [thick, red, directed] (In-Red) -- (origin) ; 
        \draw [thick, red, directed] (origin) -- (Out-Red) ; 
        \draw [thick,densely dash dot dot, red] (Out-Red) -- (Out-Dots) ; 
        \draw [thick, blue, directed] (In-Blue1) -- (interaction1)  ; 
        \draw [thick, blue, densely dash dot dot] (In-Blue-Dots1) -- (In-Blue1)  ; 
        \draw [thick, blue, directed] (In-Blue2) -- (interaction2)  ; 
        \draw [thick, blue, densely dash dot dot] (In-Blue-Dots2) -- (In-Blue2)  ; 
        \draw [thick, Green, directed] (origin) arc (0:120:0.5) ; 
        \draw [thick, Green, dash dot dot] (interaction1) arc (120:240:0.5) ; 
        \draw [thick, Green, directed] (interaction2) arc (240:360:0.5) ; 
        \draw [thick, Green, enddirected] (origin) -- (Out-Green)  ; 
        \fill (origin) circle (1.2pt);
     	\fill ([xshift=-2pt,yshift=-1.5pt]origin) rectangle ++ (4pt,3pt);
     	\fill ([xshift=-2pt,yshift=-1.5pt]interaction1) rectangle ++ (4pt,3pt);
     	\fill ([xshift=-2pt,yshift=-1.5pt]interaction2) rectangle ++ (4pt,3pt);
     	\node[right,inner sep=1pt] at (In-Blue-Dots1) {${}^{1}$};
     	\node[right,inner sep=0pt] at (In-Blue-Dots2) {$\,{}^{n}$};
    \end{tikzpicture}}}%
}
\newcommand{\GammaTwointeractingPrecedingGammaOneSubbedbis}{\vcenter{\hbox{\ifTikzExternal{\tikzsetnextfilename{Gamma2interactingPrecedingGamma1Subbedbis}}
\begin{tikzpicture}[scale=1,x={(1,0)},y={(0,1)}]
		\path 
         (0,0) coordinate (origin)  
         ++(-0.4,-0.4) coordinate (In-Red)
         +(-0.2,-0.2) coordinate (In-Dots)
         (origin)++(0.5,0.5) coordinate (Emit-Green1)
         ++(0.4,0) coordinate (interaction1)
         +(0,-0.3) coordinate (In-Blue1)
         +(0,-0.5) coordinate (In-Blue-Dots1)
         (Emit-Green1)++(0.5,0.5) coordinate (Emit-Green2)
         ++(0.4,0) coordinate (interaction2)
         +(0,-0.3) coordinate (In-Blue2)
         +(0,-0.5) coordinate (In-Blue-Dots2)
         (Emit-Green2)++(0.5,0.5) coordinate (Emit-Green3)
         ++(-1.5,0) coordinate (interaction3)
         +(-0.3,0) coordinate (Out-Green)
         (Emit-Green3)++(0.3,0.3) coordinate (Out-Red)
         +(0.25,0.25) coordinate (Out-Dots)
            ;
        \coordinate (Der5p) at ($ ($(interaction1)!0.4!(Emit-Green1)$) + (-0.0,0.12) $) ;
        \coordinate (Der5m) at ($ ($(interaction1)!0.4!(Emit-Green1)$) - (-0.0,0.12) $) ;
        \coordinate (Der6p) at ($ ($(interaction1)!0.5!(Emit-Green1)$) + (-0.0,0.12) $) ;
        \coordinate (Der6m) at ($ ($(interaction1)!0.5!(Emit-Green1)$) - (-0.0,0.12) $) ;
        \coordinate (Der3p) at ($ ($(interaction2)!0.4!(Emit-Green2)$) + (-0.0,0.12) $) ;
        \coordinate (Der3m) at ($ ($(interaction2)!0.4!(Emit-Green2)$) - (-0.0,0.12) $) ;
        \coordinate (Der4p) at ($ ($(interaction2)!0.5!(Emit-Green2)$) + (-0.0,0.12) $) ;
        \coordinate (Der4m) at ($ ($(interaction2)!0.5!(Emit-Green2)$) - (-0.0,0.12) $) ;
        \coordinate (Der1p) at ($ ($(interaction3)!0.9!(Emit-Green3)$) + (-0.0,0.12) $) ;
        \coordinate (Der1m) at ($ ($(interaction3)!0.9!(Emit-Green3)$) - (-0.0,0.12) $) ;
        \coordinate (Der2p) at ($ ($(interaction3)!0.85!(Emit-Green3)$) + (-0.0,0.12) $) ;
        \coordinate (Der2m) at ($ ($(interaction3)!0.85!(Emit-Green3)$) - (-0.0,0.12) $) ;
        \draw [thick,densely dash dot dot, red] (In-Dots) -- (In-Red) ; 
        \draw [thick, red, directed] (In-Red) -- (origin) ; 
        \draw [thick, red, directed] (origin) -- (Emit-Green1) ; 
        \draw [thick, red, dash dot dot] (Emit-Green1) -- (Emit-Green2) ; 
        \draw [thick, red, directed] (Emit-Green2) -- (Emit-Green3) ; 
        \draw [thick, red, directed] (Emit-Green3) -- (Out-Red) ; 
        \draw [thick,densely dash dot dot, red] (Out-Red) -- (Out-Dots) ; 
        \draw [thick, blue, directed] (origin) -- (interaction3)  ; 
        \draw [thick, blue, directed] (In-Blue1) -- (interaction1)  ; 
        \draw [thick, blue, densely dash dot dot] (In-Blue-Dots1) -- (In-Blue1)  ; 
        \draw [thick, blue, directed] (In-Blue2) -- (interaction2)  ; 
        \draw [thick, blue, densely dash dot dot] (In-Blue-Dots2) -- (In-Blue2)  ; 
        \draw [thick, Green, directed] (Emit-Green1) -- (interaction1)  ; 
        \draw [thick, Green, directed] (Emit-Green2) -- (interaction2)  ; 
        \draw [thick, Green, directed] (Emit-Green3) -- (interaction3)  ; 
        \draw [thick, Green, enddirected] (interaction3) -- (Out-Green)  ; 
        \draw [thick,black] (Der1p) -- (Der1m) ;
        \draw [thick,black] (Der2p) -- (Der2m) ;
        \draw [thick,black] (Der3p) -- (Der3m) ;
        \draw [thick,black] (Der4p) -- (Der4m) ;
        \draw [thick,black] (Der5p) -- (Der5m) ;
        \draw [thick,black] (Der6p) -- (Der6m) ;
        \fill (origin) circle (1.2pt);
        \fill (Emit-Green1) circle (1.2pt);
        \fill (Emit-Green2) circle (1.2pt);
        \fill (Emit-Green3) circle (1.2pt);
     	\fill ([xshift=-2pt,yshift=-1.5pt]interaction1) rectangle ++ (4pt,3pt);
     	\fill ([xshift=-2pt,yshift=-1.5pt]interaction2) rectangle ++ (4pt,3pt);
     	\fill ([xshift=-2pt,yshift=-1.5pt]interaction3) rectangle ++ (4pt,3pt);
     	\node[left,inner sep=0.2pt] at (In-Blue-Dots1) {${}^{1}$};
     	\node[left,inner sep=0.2pt] at (In-Blue-Dots2) {${}^{n}$};
    \end{tikzpicture}}}%
}
\newcommand{\GammaTwointeractingPrecedingGammaOneSubbedIntegratedbis}{\vcenter{\hbox{\ifTikzExternal{\tikzsetnextfilename{Gamma2interactingPrecedingGamma1SubbedIntegratedbis}}
\begin{tikzpicture}[scale=1,x={(1,0)},y={(0,1)}]
		\path 
         (0,0) coordinate (origin)%
         ++(0.4,-0.4) coordinate (In-Red)%
         +(0.2,-0.2) coordinate (In-Dots)%
         ($(origin)+(-0.5,0)$) ++(120:0.5) coordinate (interaction1)%
          +(0,-0.3) coordinate (In-Blue1)%
          +(0,-0.5) coordinate (In-Blue-Dots1)%
          +(0,0.5) coordinate (blank)%
         ($(origin)+(-0.5,0)$) ++(240:0.5) coordinate (interaction2)%
         +(0,-0.3) coordinate (In-Blue2)%
         +(0,-0.5) coordinate (In-Blue-Dots2)%
         (origin) +(-0.3,0) coordinate (Out-Green)%
         (origin) ++(0.4,0.4) coordinate (Out-Red)%
         +(0.25,0.25) coordinate (Out-Dots)%
        ;
        \draw [thick,densely dash dot dot, red] (In-Dots) -- (In-Red) ; 
        \draw [thick, red, directed] (In-Red) -- (origin) ; 
        \draw [thick, red, directed] (origin) arc (0:120:0.5) ; 
        \draw [thick, red, dash dot dot] (interaction1) arc (120:240:0.5) ; 
        \draw [thick, red, directed] (interaction2) arc (240:360:0.5) ; 
        \draw [thick, red, directed] (origin) -- (Out-Red) ; 
        \draw [thick,densely dash dot dot, red] (Out-Red) -- (Out-Dots) ; 
        \draw [thick, blue, directed] (In-Blue1) -- (interaction1)  ; 
        \draw [thick, blue, densely dash dot dot] (In-Blue-Dots1) -- (In-Blue1)  ; 
        \draw [thick, blue, directed] (In-Blue2) -- (interaction2)  ; 
        \draw [thick, blue, densely dash dot dot] (In-Blue-Dots2) -- (In-Blue2)  ; 
        \draw [thick, Green, enddirected] (origin) -- (Out-Green)  ; 
        \fill (origin) circle (1.2pt);
     	\fill ([xshift=-2pt,yshift=-1.5pt]origin) rectangle ++ (4pt,3pt);
     	\fill ([xshift=-2pt,yshift=-1.5pt]interaction1) rectangle ++ (4pt,3pt);
     	\fill ([xshift=-2pt,yshift=-1.5pt]interaction2) rectangle ++ (4pt,3pt);
     	\node[right,inner sep=1pt] at (In-Blue-Dots1) {${}^{1}$};
     	\node[right,inner sep=0pt] at (In-Blue-Dots2) {$\,{}^{n}$};
    \end{tikzpicture}}}%
}
\newcommand{\interactingGammaOne}{\vcenter{\hbox{\ifTikzExternal{\tikzsetnextfilename{interactingGamma1}}%
\begin{tikzpicture}[scale=1,x={(1,0)},y={(0,1)}]
		\path 
         (0,0) coordinate (origin)  
         +(-0.65,-0.65) coordinate (In-Dots)
         +(-0.4,-0.4) coordinate (In-Red)
         +(0.4,0.4) coordinate (Out-Red)
         +(0.65,0.65) coordinate (Out-Dots)
         +(0,0.5) coordinate (Out-Blue)
         ++(-0.7,0) coordinate (interaction)
         +(-0.3,0) coordinate (Out-Green)
         ++(0,-0.3) coordinate (In-Blue)
         ++(0,-0.2) coordinate (In-Blue-Dots)
            ;
        \draw [thick, densely dash dot dot, red] (In-Red) -- (In-Dots) ; 
        \draw [thick, red, directed] (In-Red) -- (origin) ; 
        \draw [thick, red, directed] (origin) -- (Out-Red) ; 
        \draw [thick,densely dash dot dot, red] (Out-Red) -- (Out-Dots) ; 
        \draw [thick, densely dashed, blue, enddirected] (origin) -- (Out-Blue)  ; 
        \draw [thick, blue, directed] (In-Blue) -- (interaction)  ; 
        \draw [thick, blue, densely dash dot dot] (In-Blue-Dots) -- (In-Blue)  ; 
        \draw [thick, Green, directed] (origin) -- (interaction)  ; 
        \draw [thick, Green, enddirected] (interaction) -- (Out-Green)  ; 
        \fill (origin) circle (1.2pt);
     	\fill ([xshift=-2pt,yshift=-1.5pt]interaction) rectangle ++ (4pt,3pt);
    \end{tikzpicture}}}%
}
\newcommand{\interactingGammaOneBis}{\vcenter{\hbox{\ifTikzExternal{\tikzsetnextfilename{interactingGamma1bis}}
\begin{tikzpicture}[scale=1,x={(1,0)},y={(0,1)}]
		\path 
         (0,0) coordinate (origin)  
         +(-0.65,-0.65) coordinate (In-Dots)
         +(-0.4,-0.4) coordinate (In-Red)
         +(0.4,0.4) coordinate (Out-Red)
         +(0.65,0.65) coordinate (Out-Dots)
         +(0,0.5) coordinate (Out-Blue)
         ++(-0.7,0) coordinate (interaction)
         +(-0.3,0) coordinate (Out-Green)
         ++(0,-0.3) coordinate (In-Blue)
         ++(0,-0.2) coordinate (In-Blue-Dots)
            ;
        \draw [thick,densely dash dot dot, red] (In-Red) -- (In-Dots) ; 
        \draw [thick, red, directed] (In-Red) -- (origin) ; 
        \draw [thick, red, directed] (origin) -- (Out-Red) ; 
        \draw [thick,densely dash dot dot, red] (Out-Red) -- (Out-Dots) ; 
        \draw [thick, densely dashed, blue, enddirected] (origin) -- (Out-Blue)  ; 
        \draw [thick, blue, directed] (In-Blue) -- (interaction)  ; 
        \draw [thick, blue, densely dash dot dot] (In-Blue-Dots) -- (In-Blue)  ; 
        \draw [thick, Green, directed] (origin) -- (interaction)  ; 
        \draw [thick, Green, enddirected] (interaction) -- (Out-Green)  ; 
        \fill (origin) circle (1.2pt);
     	\fill ([xshift=-2pt,yshift=-1.5pt]interaction) rectangle ++ (4pt,3pt);
     	\node[below right,inner sep=-1pt] (origin) {\scalebox{0.8}{$\gamma_1^{\delta \chi}$}};
    \end{tikzpicture}}}%
}
\newcommand{\GammaOneSubbedByGammaTwopbis}{\vcenter{\hbox{\ifTikzExternal{\tikzsetnextfilename{Gamma1SubbedByGamma2pbis}}%
\begin{tikzpicture}[scale=1,x={(1,0)},y={(0,1)}]
		\path 
         (0,0) coordinate (origin)  
         +(-0.65,-0.65) coordinate (In-Dots)
         +(-0.4,-0.4) coordinate (In-Red)
         +(0,0.3) coordinate (Out-Blue)
         ++(0.4,0.4) coordinate (Emit-Green)
         +(0.3,0.3) coordinate (Out-Red)
         +(0.55,0.55) coordinate (Out-Dots)
         ++(-0.7,0) coordinate (interaction)
         +(-0.3,0) coordinate (Out-Green)
         ++(0,-0.3) coordinate (In-Blue)
         ++(0,-0.2) coordinate (In-Blue-Dots)
            ;
        \coordinate (Der1p) at ($ ($(origin)!0.7!(Emit-Green)$) + (-0.07,0.07) $) ;
        \coordinate (Der1m) at ($ ($(origin)!0.7!(Emit-Green)$) - (-0.07,0.07) $) ;
        \coordinate (Der2p) at ($ ($(origin)!0.6!(Emit-Green)$) + (-0.07,0.07) $) ;
        \coordinate (Der2m) at ($ ($(origin)!0.6!(Emit-Green)$) - (-0.07,0.07) $) ;
        \draw [thick,densely dash dot dot, red] (In-Dots) -- (In-Red) ; 
        \draw [thick, red, directed] (In-Red) -- (origin) ; 
        \draw [thick, red, directed] (origin) -- (Emit-Green) ; 
        \draw [thick, red, directed] (Emit-Green) -- (Out-Red) ; 
        \draw [thick,densely dash dot dot, red] (Out-Red) -- (Out-Dots) ; 
        \draw [thick, densely dashed, blue, enddirected] (origin) -- (Out-Blue)  ; 
        \draw [thick, blue, directed] (In-Blue) -- (interaction)  ; 
        \draw [thick, blue, densely dash dot dot] (In-Blue-Dots) -- (In-Blue)  ; 
        \draw [thick, Green, directed] (Emit-Green) -- (interaction)  ; 
        \draw [thick, Green, enddirected] (interaction) -- (Out-Green)  ; 
        \draw [thick,black] (Der1p) -- (Der1m) ;
        \draw [thick,black] (Der2p) -- (Der2m) ;
        \fill (origin) circle (1.2pt);
        \fill (Emit-Green) circle (1.2pt);
     	\fill ([xshift=-2pt,yshift=-1.5pt]interaction) rectangle ++ (4pt,3pt);
     	\node[below right,inner sep=-0.5pt] at (origin) {\scalebox{0.8}{$\gamma_1^{1}$}};
     	\node[below right,inner sep=-0.5pt] at (Emit-Green) {\scalebox{0.8}{$\gamma_2^{+}$}};
    \end{tikzpicture}}}%
}
\newcommand{\GammaOneSubbedByGammaTwop}{\vcenter{\hbox{\ifTikzExternal{\tikzsetnextfilename{Gamma1SubbedByGamma2p}}
\begin{tikzpicture}[scale=1,x={(1,0)},y={(0,1)}]
		\path 
         (0,0) coordinate (origin)  
         +(-0.65,-0.65) coordinate (In-Dots)
         +(-0.4,-0.4) coordinate (In-Red)
         +(0,0.3) coordinate (Out-Blue)
         ++(0.4,0.4) coordinate (Emit-Green)
         +(0.3,0.3) coordinate (Out-Red)
         +(0.55,0.55) coordinate (Out-Dots)
         ++(-0.7,0) coordinate (interaction)
         +(-0.3,0) coordinate (Out-Green)
         ++(0,-0.3) coordinate (In-Blue)
         ++(0,-0.2) coordinate (In-Blue-Dots)
            ;
        \coordinate (Der1p) at ($ ($(origin)!0.7!(Emit-Green)$) + (-0.07,0.07) $) ;
        \coordinate (Der1m) at ($ ($(origin)!0.7!(Emit-Green)$) - (-0.07,0.07) $) ;
        \coordinate (Der2p) at ($ ($(origin)!0.6!(Emit-Green)$) + (-0.07,0.07) $) ;
        \coordinate (Der2m) at ($ ($(origin)!0.6!(Emit-Green)$) - (-0.07,0.07) $) ;
        \draw [thick,densely dash dot dot, red] (In-Dots) -- (In-Red) ; 
        \draw [thick, red, directed] (In-Red) -- (origin) ; 
        \draw [thick, red, directed] (origin) -- (Emit-Green) ; 
        \draw [thick, red, directed] (Emit-Green) -- (Out-Red) ; 
        \draw [thick,densely dash dot dot, red] (Out-Red) -- (Out-Dots) ; 
        \draw [thick, densely dashed, blue, enddirected] (origin) -- (Out-Blue)  ; 
        \draw [thick, blue, directed] (In-Blue) -- (interaction)  ; 
        \draw [thick, blue, densely dash dot dot] (In-Blue-Dots) -- (In-Blue)  ; 
        \draw [thick, Green, directed] (Emit-Green) -- (interaction)  ; 
        \draw [thick, Green, enddirected] (interaction) -- (Out-Green)  ; 
        \draw [thick,black] (Der1p) -- (Der1m) ;
        \draw [thick,black] (Der2p) -- (Der2m) ;
        \fill (origin) circle (1.2pt);
        \fill (Emit-Green) circle (1.2pt);
     	\fill ([xshift=-2pt,yshift=-1.5pt]interaction) rectangle ++ (4pt,3pt);
     	\node[below right,inner sep=1pt] at (origin) {${}^{\gamma_1^1}$};
     	\node[below right,inner sep=0.2pt] at (Emit-Green) {${}^{\gamma_2^+}$};
    \end{tikzpicture}}}%
}
\newcommand{\interactingGvertex}{\vcenter{\hbox{\ifTikzExternal{\tikzsetnextfilename{interactingGvertex}}%
\begin{tikzpicture}[scale=1,x={(1,0)},y={(0,1)}]
		\path 
         (0,0) coordinate (origin)  
         +(-0.65,-0.65) coordinate (In-Dots)
         +(-0.4,-0.4) coordinate (In-Red)
         +(0.4,0.4) coordinate (Out-Red)
         +(0.65,0.65) coordinate (Out-Dots)
         +(0,0.5) coordinate (Out-Blue)
         ++(-0.7,0) coordinate (interaction1)
         +(0,-0.3) coordinate (In-Blue1)
         +(0,-0.5) coordinate (In-Blue-Dots1)
         ++(-0.4,0) coordinate (interaction2)
         +(-0.3,0) coordinate (Out-Green)
         ++(0,-0.3) coordinate (In-Blue2)
         ++(0,-0.2) coordinate (In-Blue-Dots2)
            ;
        \coordinate (Der1p) at ($ ($(origin)!0.3!(interaction1)$) + (0.0,0.12) $) ;
        \coordinate (Der1m) at ($ ($(origin)!0.3!(interaction1)$) - (0.0,0.12) $) ;
        \coordinate (Der2p) at ($ ($(origin)!0.4!(interaction1)$) + (-0.0,0.12) $) ;
        \coordinate (Der2m) at ($ ($(origin)!0.4!(interaction1)$) - (-0.0,0.12) $) ;    
        \draw [thick,densely dash dot dot, red] (In-Red) -- (In-Dots) ; 
        \draw [thick, red, directed] (In-Red) -- (origin) ; 
        \draw [thick, red, directed] (origin) -- (Out-Red) ; 
        \draw [thick,densely dash dot dot, red] (Out-Red) -- (Out-Dots) ; 
        \draw [thick, blue, directed] (In-Blue1) -- (interaction1)  ; 
        \draw [thick, blue, densely dash dot dot] (In-Blue-Dots1) -- (In-Blue1)  ; 
        \draw [thick, blue, directed] (In-Blue2) -- (interaction2)  ; 
        \draw [thick, blue, densely dash dot dot] (In-Blue-Dots2) -- (In-Blue2)  ; 
        \draw [thick, Green, directed] (origin) -- (interaction1)  ; 
        \draw [thick, Green, directed] (interaction1) -- (interaction2)  ; 
        \draw [thick, Green, enddirected] (interaction2) -- (Out-Green)  ; 
        \draw [thick,black] (Der1p) -- (Der1m) ;
        \draw [thick,black] (Der2p) -- (Der2m) ;
        \fill (origin) circle (1.2pt);
     	\fill ([xshift=-2pt,yshift=-1.5pt]interaction1) rectangle ++ (4pt,3pt);
     	\fill ([xshift=-2pt,yshift=-1.5pt]interaction2) rectangle ++ (4pt,3pt);
     	\node[left,inner sep=0.2pt] at (In-Blue-Dots1) {${}^{1}$};
     	\node[left,inner sep=0.2pt] at (In-Blue-Dots2) {${}^{2}$};
     	\node[below right,inner sep=.5pt] at (origin) {${}^{{\gamma_2^-}}$};
    \end{tikzpicture}}}%
}
\newcommand{\interactingGvertexbis}{\vcenter{\hbox{\ifTikzExternal{\tikzsetnextfilename{interactingGvertexbis}}
\begin{tikzpicture}[scale=1,x={(1,0)},y={(0,1)}]
\path 
(0,0) coordinate (origin)  
+(-0.65,-0.65) coordinate (In-Dots)
+(-0.4,-0.4) coordinate (In-Red)
+(0.4,0.4) coordinate (Out-Red)
+(0.65,0.65) coordinate (Out-Dots)
+(0,0.5) coordinate (Out-Blue)
++(-0.7,0) coordinate (interaction1)
+(0,-0.3) coordinate (In-Blue1)
+(0,-0.5) coordinate (In-Blue-Dots1)
++(-0.4,0) coordinate (interaction2)
+(-0.3,0) coordinate (Out-Green)
++(0,-0.3) coordinate (In-Blue2)
++(0,-0.2) coordinate (In-Blue-Dots2)
;
\coordinate (Der1p) at ($ ($(origin)!0.3!(interaction1)$) + (0.0,0.12) $) ;
\coordinate (Der1m) at ($ ($(origin)!0.3!(interaction1)$) - (0.0,0.12) $) ;
\coordinate (Der2p) at ($ ($(origin)!0.4!(interaction1)$) + (-0.0,0.12) $) ;
\coordinate (Der2m) at ($ ($(origin)!0.4!(interaction1)$) - (-0.0,0.12) $) ;    
\draw [thick,densely dash dot dot, red] (In-Red) -- (In-Dots) ; 
\draw [thick, red, directed] (In-Red) -- (origin) ; 
\draw [thick, red, directed] (origin) -- (Out-Red) ; 
\draw [thick,densely dash dot dot, red] (Out-Red) -- (Out-Dots) ; 
\draw [thick, blue, directed] (In-Blue1) -- (interaction1)  ; 
\draw [thick, blue, densely dash dot dot] (In-Blue-Dots1) -- (In-Blue1)  ; 
\draw [thick, blue, directed] (In-Blue2) -- (interaction2)  ; 
\draw [thick, blue, densely dash dot dot] (In-Blue-Dots2) -- (In-Blue2)  ; 
\draw [thick, Green, directed] (origin) -- (interaction1)  ; 
\draw [thick, Green, directed] (interaction1) -- (interaction2)  ; 
\draw [thick, Green, enddirected] (interaction2) -- (Out-Green)  ; 
\draw [thick,black] (Der1p) -- (Der1m) ;
\draw [thick,black] (Der2p) -- (Der2m) ;
\fill (origin) circle (1.2pt);
\fill ([xshift=-2pt,yshift=-1.5pt]interaction1) rectangle ++ (4pt,3pt);
\fill ([xshift=-2pt,yshift=-1.5pt]interaction2) rectangle ++ (4pt,3pt);
\node[below right,inner sep=-0.5pt] at (origin) {\scalebox{0.8}{$\gamma_2^{-}$}};
\node[left,inner sep=0.2pt] at (In-Blue-Dots1) {${}^{1}$};
\node[left,inner sep=0.2pt] at (In-Blue-Dots2) {${}^{2}$};
\end{tikzpicture}}}%
}
\newcommand{\interactingGvertexIntegrated}{\vcenter{\hbox{\ifTikzExternal{\tikzsetnextfilename{interactingGvertexIntegrated}}
\begin{tikzpicture}[scale=1,x={(1,0)},y={(0,1)}]
\path 
(0,0) coordinate (origin)  
++(-0.7,0) coordinate (interaction1)        
+(0.65,-0.65) coordinate (In-Dots)
+(0.4,-0.4) coordinate (In-Red)
+(0.4,0.4) coordinate (Out-Red)
+(0.65,0.65) coordinate (Out-Dots)
+(0,-0.3) coordinate (In-Blue1)
+(0,-0.5) coordinate (In-Blue-Dots1)
++(-0.4,0) coordinate (interaction2)
+(-0.3,0) coordinate (Out-Green)
++(0,-0.3) coordinate (In-Blue2)
++(0,-0.2) coordinate (In-Blue-Dots2)
;
\coordinate (Der1p) at ($ ($(origin)!0.3!(interaction1)$) + (0.0,0.12) $) ;
\coordinate (Der1m) at ($ ($(origin)!0.3!(interaction1)$) - (0.0,0.12) $) ;
\coordinate (Der2p) at ($ ($(origin)!0.4!(interaction1)$) + (-0.0,0.12) $) ;
\coordinate (Der2m) at ($ ($(origin)!0.4!(interaction1)$) - (-0.0,0.12) $) ;    
\draw [thick,densely dash dot dot, red] (In-Red) -- (In-Dots) ; 
\draw [thick, red, directed] (In-Red) -- (interaction1) ; 
\draw [thick, red, directed] (interaction1) -- (Out-Red) ; 
\draw [thick,densely dash dot dot, red] (Out-Red) -- (Out-Dots) ; 
\draw [thick, blue, directed] (In-Blue1) -- (interaction1)  ; 
\draw [thick, blue, densely dash dot dot] (In-Blue-Dots1) -- (In-Blue1)  ; 
\draw [thick, blue, directed] (In-Blue2) -- (interaction2)  ; 
\draw [thick, blue, densely dash dot dot] (In-Blue-Dots2) -- (In-Blue2)  ; 
\draw [thick, Green, directed] (interaction1) -- (interaction2)  ; 
\draw [thick, Green, enddirected] (interaction2) -- (Out-Green)  ; 
\fill ([xshift=-2pt,yshift=-1.5pt]interaction1) rectangle ++ (4pt,3pt);
\fill ([xshift=-2pt,yshift=-1.5pt]interaction2) rectangle ++ (4pt,3pt);
\node[left,inner sep=0.2pt] at (In-Blue-Dots1) {${}^{1}$};
\node[left,inner sep=0.2pt] at (In-Blue-Dots2) {${}^{2}$};
\end{tikzpicture}}}%
}
\newcommand{\GvertexSubbedbyGammaTwo}{\vcenter{\hbox{\ifTikzExternal{\tikzsetnextfilename{GvertexSubbedbyGamma2}}
\begin{tikzpicture}[scale=1,x={(1,0)},y={(0,1)}]
        \path 
         (0,0) coordinate (origin)  
         +(-0.65,-0.65) coordinate (In-Dots)
         +(-0.4,-0.4) coordinate (In-Red)
         ++(-0.7,0) coordinate (interaction1)
            +(0,-0.3) coordinate (In-Blue1)
            +(0,-0.5) coordinate (In-Blue-Dots1)
         (origin)++(0.4,0.4) coordinate (Emit-Green)
         +(0.4,0.4) coordinate (Out-Red)
         +(0.65,0.65) coordinate (Out-Dots)
         ++(-1.4,0) coordinate (interaction2)
         +(-0.3,0) coordinate (Out-Green)
         ++(0,-0.3) coordinate (In-Blue2)
         ++(0,-0.2) coordinate (In-Blue-Dots2)
            ;
        \coordinate (Der1p) at ($ ($(origin)!0.3!(interaction1)$) + (0.0,0.12) $) ;
        \coordinate (Der1m) at ($ ($(origin)!0.3!(interaction1)$) - (0.0,0.12) $) ;
        \coordinate (Der2p) at ($ ($(origin)!0.4!(interaction1)$) + (-0.0,0.12) $) ;
        \coordinate (Der2m) at ($ ($(origin)!0.4!(interaction1)$) - (-0.0,0.12) $) ;    
        \coordinate (Der3p) at ($ ($(origin)!0.7!(Emit-Green)$) + (0.07,-0.07) $) ;
        \coordinate (Der3m) at ($ ($(origin)!0.7!(Emit-Green)$) - (0.07,-0.07) $) ;
        \coordinate (Der4p) at ($ ($(origin)!0.6!(Emit-Green)$) + (0.07,-0.07) $) ;
        \coordinate (Der4m) at ($ ($(origin)!0.6!(Emit-Green)$) - (0.07,-0.07) $) ;    
        \draw [thick,densely dash dot dot, red] (In-Red) -- (In-Dots) ; 
        \draw [thick, red, directed] (In-Red) -- (origin) ; 
        \draw [thick, red, directed] (origin) -- (Emit-Green) ; 
        \draw [thick, red, directed] (Emit-Green) -- (Out-Red) ; 
        \draw [thick,densely dash dot dot, red] (Out-Red) -- (Out-Dots) ; 
        \draw [thick, blue, directed] (In-Blue1) -- (interaction1)  ; 
        \draw [thick, blue, densely dash dot dot] (In-Blue-Dots1) -- (In-Blue1)  ; 
        \draw [thick, blue, directed] (In-Blue2) -- (interaction2)  ; 
        \draw [thick, blue, densely dash dot dot] (In-Blue-Dots2) -- (In-Blue2)  ; 
        \draw [thick, Green, directed] (origin) -- (interaction1)  ; 
        \draw [thick, Green, directed] (Emit-Green) -- (interaction2)  ; 
        \draw [thick, Green, enddirected] (interaction2) -- (Out-Green)  ; 
        \draw [thick,black] (Der1p) -- (Der1m) ;
        \draw [thick,black] (Der2p) -- (Der2m) ;
        \draw [thick,black] (Der3p) -- (Der3m) ;
        \draw [thick,black] (Der4p) -- (Der4m) ;
        \fill (origin) circle (1.2pt);
        \fill (Emit-Green) circle (1.2pt);
     	\fill ([xshift=-2pt,yshift=-1.5pt]interaction1) rectangle ++ (4pt,3pt);
     	\fill ([xshift=-2pt,yshift=-1.5pt]interaction2) rectangle ++ (4pt,3pt);
     	\node[left,inner sep=0.2pt] at (In-Blue-Dots1) {${}^{1}$};
     	\node[left,inner sep=0.2pt] at (In-Blue-Dots2) {${}^{2}$};
     	\node[below right,inner sep=.5pt] at (origin) {${}^{{\gamma_2^-}}$};
     	\node[below right,inner sep=.5pt] at (Emit-Green) {${}^{{\gamma_2^+}}$};
    \end{tikzpicture}}}%
}
\newcommand{\GvertexSubbedbyGammaTwoBis}{\vcenter{\hbox{\ifTikzExternal{\tikzsetnextfilename{GvertexSubbedbyGamma2bis}}
\begin{tikzpicture}[scale=1,x={(1,0)},y={(0,1)}]
\path 
(0,0) coordinate (origin)  
+(-0.65,-0.65) coordinate (In-Dots)
+(-0.4,-0.4) coordinate (In-Red)
++(-0.7,0) coordinate (interaction1)
+(0,-0.3) coordinate (In-Blue1)
+(0,-0.5) coordinate (In-Blue-Dots1)
(origin)++(0.4,0.4) coordinate (Emit-Green)
+(0.4,0.4) coordinate (Out-Red)
+(0.65,0.65) coordinate (Out-Dots)
++(-1.4,0) coordinate (interaction2)
+(-0.3,0) coordinate (Out-Green)
++(0,-0.3) coordinate (In-Blue2)
++(0,-0.2) coordinate (In-Blue-Dots2)
;
\coordinate (Der1p) at ($ ($(origin)!0.3!(interaction1)$) + (0.0,0.12) $) ;
\coordinate (Der1m) at ($ ($(origin)!0.3!(interaction1)$) - (0.0,0.12) $) ;
\coordinate (Der2p) at ($ ($(origin)!0.4!(interaction1)$) + (-0.0,0.12) $) ;
\coordinate (Der2m) at ($ ($(origin)!0.4!(interaction1)$) - (-0.0,0.12) $) ;    
\coordinate (Der3p) at ($ ($(origin)!0.7!(Emit-Green)$) + (0.07,-0.07) $) ;
\coordinate (Der3m) at ($ ($(origin)!0.7!(Emit-Green)$) - (0.07,-0.07) $) ;
\coordinate (Der4p) at ($ ($(origin)!0.6!(Emit-Green)$) + (0.07,-0.07) $) ;
\coordinate (Der4m) at ($ ($(origin)!0.6!(Emit-Green)$) - (0.07,-0.07) $) ;    
\draw [thick,densely dash dot dot, red] (In-Red) -- (In-Dots) ; 
\draw [thick, red, directed] (In-Red) -- (origin) ; 
\draw [thick, red, directed] (origin) -- (Emit-Green) ; 
\draw [thick, red, directed] (Emit-Green) -- (Out-Red) ; 
\draw [thick,densely dash dot dot, red] (Out-Red) -- (Out-Dots) ; 
\draw [thick, blue, directed] (In-Blue1) -- (interaction1)  ; 
\draw [thick, blue, densely dash dot dot] (In-Blue-Dots1) -- (In-Blue1)  ; 
\draw [thick, blue, directed] (In-Blue2) -- (interaction2)  ; 
\draw [thick, blue, densely dash dot dot] (In-Blue-Dots2) -- (In-Blue2)  ; 
\draw [thick, Green, directed] (origin) -- (interaction1)  ; 
\draw [thick, Green, directed] (Emit-Green) -- (interaction2)  ; 
\draw [thick, Green, enddirected] (interaction2) -- (Out-Green)  ; 
\draw [thick,black] (Der1p) -- (Der1m) ;
\draw [thick,black] (Der2p) -- (Der2m) ;
\draw [thick,black] (Der3p) -- (Der3m) ;
\draw [thick,black] (Der4p) -- (Der4m) ;
\fill (origin) circle (1.2pt);
\fill (Emit-Green) circle (1.2pt);
\fill ([xshift=-2pt,yshift=-1.5pt]interaction1) rectangle ++ (4pt,3pt);
\fill ([xshift=-2pt,yshift=-1.5pt]interaction2) rectangle ++ (4pt,3pt);
\node[left,inner sep=0.2pt] at (In-Blue-Dots1) {${}^{1}$};
\node[left,inner sep=0.2pt] at (In-Blue-Dots2) {${}^{2}$};
\node[below right,inner sep=-0.5pt] at (origin) {\scalebox{0.8}{$\gamma_2^{-}$}};
\node[below right,inner sep=-0.5pt] at (Emit-Green) {\scalebox{0.8}{$\gamma_2^{+}$}};
\end{tikzpicture}}}%
}
\newcommand{\GammaTwominusMakingLoop}{\vcenter{\hbox{\ifTikzExternal{\tikzsetnextfilename{Gamma2minusMakingLoop}}
\begin{tikzpicture}[scale=1,x={(1,0)},y={(0,1)}]
\path 
(0,0) coordinate (origin)  
++(-0.4,-0.4) coordinate (In-Red)
+(-0.2,-0.2) coordinate (In-Dots)
(origin)++(0.5,0.5) coordinate (Emit-Green1)
++(0.4,0) coordinate (interaction1)
+(0,-0.3) coordinate (In-Blue1)
+(0,-0.5) coordinate (In-Blue-Dots1)
(Emit-Green1)++(0.5,0.5) coordinate (Emit-Green2)
++(0.4,0) coordinate (interaction2)
+(0,-0.3) coordinate (In-Blue2)
+(0,-0.5) coordinate (In-Blue-Dots2)
(Emit-Green2)++(0.5,0.5) coordinate (Emit-Green3)
++(-1.5,0) coordinate (interaction3)
+(-0.3,0) coordinate (Out-Green)
(Emit-Green3)++(0.3,0.3) coordinate (Out-Red)
+(0.25,0.25) coordinate (Out-Dots)
;
\coordinate (Der5p) at ($ ($(interaction1)!0.4!(Emit-Green1)$) + (-0.0,0.12) $) ;
\coordinate (Der5m) at ($ ($(interaction1)!0.4!(Emit-Green1)$) - (-0.0,0.12) $) ;
\coordinate (Der6p) at ($ ($(interaction1)!0.5!(Emit-Green1)$) + (-0.0,0.12) $) ;
\coordinate (Der6m) at ($ ($(interaction1)!0.5!(Emit-Green1)$) - (-0.0,0.12) $) ;
\coordinate (Der3p) at ($ ($(interaction2)!0.4!(Emit-Green2)$) + (-0.0,0.12) $) ;
\coordinate (Der3m) at ($ ($(interaction2)!0.4!(Emit-Green2)$) - (-0.0,0.12) $) ;
\coordinate (Der4p) at ($ ($(interaction2)!0.5!(Emit-Green2)$) + (-0.0,0.12) $) ;
\coordinate (Der4m) at ($ ($(interaction2)!0.5!(Emit-Green2)$) - (-0.0,0.12) $) ;
\coordinate (Der1p) at ($ ($(interaction3)!0.9!(Emit-Green3)$) + (-0.0,0.12) $) ;
\coordinate (Der1m) at ($ ($(interaction3)!0.9!(Emit-Green3)$) - (-0.0,0.12) $) ;
\coordinate (Der2p) at ($ ($(interaction3)!0.85!(Emit-Green3)$) + (-0.0,0.12) $) ;
\coordinate (Der2m) at ($ ($(interaction3)!0.85!(Emit-Green3)$) - (-0.0,0.12) $) ;
\draw [thick,densely dash dot dot, red] (In-Dots) -- (In-Red) ; 
\draw [thick, red, directed] (In-Red) -- (origin) ; 
\draw [thick, red, directed] (origin) -- (Emit-Green1) ; 
\draw [thick, red, dash dot dot] (Emit-Green1) -- (Emit-Green2) ; 
\draw [thick, red, directed] (Emit-Green2) -- (Emit-Green3) ; 
\draw [thick, red, directed] (Emit-Green3) -- (Out-Red) ; 
\draw [thick,densely dash dot dot, red] (Out-Red) -- (Out-Dots) ; 
\draw [thick, blue, directed] (origin) -- (interaction3)  ; 
\draw [thick, blue, directed] (In-Blue1) -- (interaction1)  ; 
\draw [thick, blue, densely dash dot dot] (In-Blue-Dots1) -- (In-Blue1)  ; 
\draw [thick, blue, directed] (In-Blue2) -- (interaction2)  ; 
\draw [thick, blue, densely dash dot dot] (In-Blue-Dots2) -- (In-Blue2)  ; 
\draw [thick, Green, directed] (Emit-Green1) -- (interaction1)  ; 
\draw [thick, Green, directed] (Emit-Green2) -- (interaction2)  ; 
\draw [thick, Green, directed] (Emit-Green3) -- (interaction3)  ; 
\draw [thick,black] (Der1p) -- (Der1m) ;
\draw [thick,black] (Der2p) -- (Der2m) ;
\draw [thick,black] (Der3p) -- (Der3m) ;
\draw [thick,black] (Der4p) -- (Der4m) ;
\draw [thick,black] (Der5p) -- (Der5m) ;
\draw [thick,black] (Der6p) -- (Der6m) ;
\fill (origin) circle (1.2pt);
\fill (Emit-Green1) circle (1.2pt);
\fill (Emit-Green2) circle (1.2pt);
\fill (Emit-Green3) circle (1.2pt);
	\fill ([xshift=-2pt,yshift=-1.5pt]interaction1) rectangle ++ (4pt,3pt);
	\fill ([xshift=-2pt,yshift=-1.5pt]interaction2) rectangle ++ (4pt,3pt);
	\fill ([xshift=-2pt,yshift=-1.5pt]interaction3) rectangle ++ (4pt,3pt);
	\node[below right,inner sep=1pt] at (origin) {\scalebox{0.8}{$\gamma_1^{1}$}};
	\node[above left,inner sep=0.1pt] at (Emit-Green1) {\scalebox{0.8}{$\gamma_2^{-}$}};
	\node[above left,inner sep=0.1pt] at (Emit-Green2) {\scalebox{0.8}{$\gamma_2^{-}$}};
	\node[below right,inner sep=0.2pt] at (Emit-Green3) {\scalebox{0.8}{$\gamma_2^{-}$}};
	\node[above left,inner sep=0.1pt] at (interaction3) {\scalebox{0.8}{$g^{1}$}};
\node[left,inner sep=0.2pt] at (In-Blue-Dots1) {${}^{1}$};
\node[left,inner sep=0.2pt] at (In-Blue-Dots2) {${}^{2}$};
\end{tikzpicture}}}%
}
\newcommand{\GammaTwoplusMakingLoop}{\vcenter{\hbox{\ifTikzExternal{\tikzsetnextfilename{Gamma2plusMakingLoop}}
\begin{tikzpicture}[scale=1,x={(1,0)},y={(0,1)}]
\path 
(0,0) coordinate (origin)  
++(-0.4,-0.4) coordinate (In-Red)
+(-0.2,-0.2) coordinate (In-Dots)
(origin)++(1.5,1.5) coordinate (Emit-Green)
++(-0.6,0) coordinate (interaction1)
+(0,-0.3) coordinate (In-Blue1)
+(0,-0.5) coordinate (In-Blue-Dots1)
++(-0.45,0) coordinate (interaction2)
+(0,-0.3) coordinate (In-Blue2)
+(0,-0.5) coordinate (In-Blue-Dots2)
++(-0.45,0) coordinate (interaction3)
+(-0.3,0) coordinate (Out-Green)
(Emit-Green)++(0.3,0.3) coordinate (Out-Red)
+(0.25,0.25) coordinate (Out-Dots)
;
\coordinate (Der1p) at ($ ($(origin)!0.85!(Emit-Green)$) + (-0.07,0.07) $) ;
\coordinate (Der1m) at ($ ($(origin)!0.85!(Emit-Green)$) - (-0.07,0.07) $) ;
\coordinate (Der2p) at ($ ($(origin)!0.9!(Emit-Green)$) + (-0.07,0.07) $) ;
\coordinate (Der2m) at ($ ($(origin)!0.9!(Emit-Green)$) - (-0.07,0.07) $) ;
\draw [thick,densely dash dot dot, red] (In-Dots) -- (In-Red) ; 
\draw [thick, red, directed] (In-Red) -- (origin) ; 
\draw [thick, red, directed] (origin) -- (Emit-Green) ; 
\draw [thick, red, directed] (Emit-Green) -- (Out-Red) ; 
\draw [thick,densely dash dot dot, red] (Out-Red) -- (Out-Dots) ; 
\draw [thick, blue, directed] (origin) -- (interaction3)  ; 
\draw [thick, blue, directed] (In-Blue1) -- (interaction1)  ; 
\draw [thick, blue, densely dash dot dot] (In-Blue-Dots1) -- (In-Blue1)  ; 
\draw [thick, blue, directed] (In-Blue2) -- (interaction2)  ; 
\draw [thick, blue, densely dash dot dot] (In-Blue-Dots2) -- (In-Blue2)  ; 
\draw [thick, Green, directed] (Emit-Green) -- (interaction1)  ; 
\draw [thick, densely dash dot dot, Green] (interaction1) -- (interaction2)  ; 
\draw [thick, Green, directed] (interaction2) -- (interaction3)  ; 
\draw [thick,black] (Der1p) -- (Der1m) ;
\draw [thick,black] (Der2p) -- (Der2m) ;
\fill (origin) circle (1.2pt);
\fill (Emit-Green) circle (1.2pt);
 	\fill ([xshift=-2pt,yshift=-1.5pt]interaction1) rectangle ++ (4pt,3pt);
 	\fill ([xshift=-2pt,yshift=-1.5pt]interaction2) rectangle ++ (4pt,3pt);
 	\fill ([xshift=-2pt,yshift=-1.5pt]interaction3) rectangle ++ (4pt,3pt);
 	\node[below right,inner sep=0pt] at (origin) {\scalebox{0.8}{$\gamma_1^{1}$}};
 	\node[below right,inner sep=0pt] at (Emit-Green) {\scalebox{0.8}{$\gamma_2^{+}$}};
 	\node[above, inner sep=3pt] at (interaction1) {\scalebox{0.8}{$~g^{\delta \chi}$}};
 	\node[above ,inner sep=3pt] at (interaction2) {\scalebox{0.8}{$g^{\delta \chi}$}};
 	\node[above left,inner sep=0pt] at (interaction3) {\scalebox{0.8}{$g^{1}$}};
\node[left,inner sep=0.2pt] at (In-Blue-Dots1) {${}^{1}$};
\node[left,inner sep=0.2pt] at (In-Blue-Dots2) {${}^{2}$};
\end{tikzpicture}}}%
}
\newcommand{\NOTallowedTwoLoop}{\vcenter{\hbox{\ifTikzExternal{\tikzsetnextfilename{NOTallowedTwoLoop}}
\begin{tikzpicture}[scale=1,x={(1,0)},y={(0,1)}]
\path 
(0,0) coordinate (origin)  
+(-0.3,-0.3) coordinate (In-Red)
++(.3,.3) coordinate (Emit-Blue1)
+(0,.3) coordinate (Out-Blue)
++(.3,.3) coordinate (Emit-Blue2)
++(.4,.4) coordinate (Emit-Green1)
++(.3,.3) coordinate (Emit-Green2)
++(0.3,0.3) coordinate (Out-Red)
;
\path
let \p1 = (origin), \p2 = (Emit-Green1) in  (\x1,\y2) coordinate (interaction1)
;
\path
let \p1 = (Emit-Blue2), \p2 = (Emit-Green2) in  (\x1,\y2) coordinate (interaction2)
;
\path 
($ (Emit-Green1) + (-0.2,0.1) $) coordinate (Der1p)
++(0,-0.2) coordinate (Der1m)
++(0.05,0) coordinate (Der2m)
++(0,0.2) coordinate (Der2p)
;
\path 
($ (Emit-Green2) + (-0.2,0.1) $) coordinate (Der3p)
++(0,-0.2) coordinate (Der3m)
++(0.05,0) coordinate (Der4m)
++(0,0.2) coordinate (Der4p)
;
\fill[yellow!30,fill opacity=1,shift={(Emit-Blue2)}] (-0.1,-0.07) rectangle (0.75,0.85);
\draw [thick, Green, directed] (Emit-Green1) -- (interaction1)  ; 
\draw [thick, Green, directed] (Emit-Green2) -- (interaction2)  ; 
\draw [thick, red, directed] (In-Red) -- (origin) ; 
\draw [thick, red, directed] (origin) -- (Emit-Blue1) ; 
\draw [thick, red, directed] (Emit-Blue1) -- (Emit-Blue2) ; 
\draw [thick, red, directed] (Emit-Blue2) -- (Emit-Green1) ; 
\draw [thick, red, directed] (Emit-Green1) -- (Emit-Green2) ; 
\draw [thick, red, enddirected] (Emit-Green2) -- (Out-Red) ; 
\draw [thick, blue, directed] (origin) -- (interaction1)  ; 
\draw [thick, blue, enddirected, densely dashed] (Emit-Blue1) -- (Out-Blue)  ; 
\draw [thick, blue, directed] (Emit-Blue2) -- (interaction2)  ; 
\draw [thick,black] (Der1p) -- (Der1m) ;
\draw [thick,black] (Der2p) -- (Der2m) ;
\draw [thick,black] (Der3p) -- (Der3m) ;
\draw [thick,black] (Der4p) -- (Der4m) ;
\fill (origin) circle (1.2pt);
\fill (Emit-Blue1) circle (1.2pt);
\fill (Emit-Blue2) circle (1.2pt);
\fill (Emit-Green1) circle (1.2pt);
\fill (Emit-Green2) circle (1.2pt);
\fill ([xshift=-2pt,yshift=-1.5pt]interaction1) rectangle ++ (4pt,3pt);
\fill ([xshift=-2pt,yshift=-1.5pt]interaction2) rectangle ++ (4pt,3pt);
\end{tikzpicture}}}%
}
\newcommand{\allowedTwoLoop}{\vcenter{\hbox{\ifTikzExternal{\tikzsetnextfilename{allowedTwoLoop}}
\begin{tikzpicture}[scale=1,x={(1,0)},y={(0,1)}]
		\path 
         (0,0) coordinate (origin)  
          +(-0.3,-0.3) coordinate (In-Red)
         ++(.3,.3) coordinate (Emit-Blue1)
         ++(.3,.3) coordinate (Emit-Blue2)
          +(0,.3) coordinate (Out-Blue)
         ++(.4,.4) coordinate (Emit-Green1)
         ++(.3,.3) coordinate (Emit-Green2)
         ++(0.3,0.3) coordinate (Out-Red)
        ;
        \path
         let \p1 = (origin), \p2 = (Emit-Green1) in  (\x1,\y2) coordinate (interaction1)
        ;
        \path
         let \p1 = (Emit-Blue1), \p2 = (Emit-Green2) in  (\x1,\y2) coordinate (interaction2)
        ;
        \path 
          ($ (Emit-Green1) + (-0.2,0.1) $) coordinate (Der1p)
          ++(0,-0.2) coordinate (Der1m)
          ++(0.05,0) coordinate (Der2m)
          ++(0,0.2) coordinate (Der2p)
        ;
        \path 
          ($ (Emit-Green2) + (-0.2,0.1) $) coordinate (Der3p)
          ++(0,-0.2) coordinate (Der3m)
          ++(0.05,0) coordinate (Der4m)
          ++(0,0.2) coordinate (Der4p)
        ;
        \draw [thick, Green, directed] (Emit-Green1) -- (interaction1)  ; 
        \draw [thick, Green, directed] (Emit-Green2) -- (interaction2)  ; 
        \draw [line width=3pt, white] (Emit-Blue1) -- (interaction2) ; 
        \draw [thick, red, directed] (In-Red) -- (origin) ; 
        \draw [thick, red, directed] (origin) -- (Emit-Blue1) ; 
        \draw [thick, red, directed] (Emit-Blue1) -- (Emit-Blue2) ; 
        \draw [thick, red, directed] (Emit-Blue2) -- (Emit-Green1) ; 
        \draw [thick, red, directed] (Emit-Green1) -- (Emit-Green2) ; 
        \draw [thick, red, enddirected] (Emit-Green2) -- (Out-Red) ; 
        \draw [thick, blue, directed] (origin) -- (interaction1)  ; 
        \draw [thick, blue, directed ] (Emit-Blue1) -- (interaction2)  ; 
        \draw [thick, blue,enddirected, densely dashed] (Emit-Blue2) -- (Out-Blue)  ; 
        \draw [thick,black] (Der1p) -- (Der1m) ;
        \draw [thick,black] (Der2p) -- (Der2m) ;
        \draw [thick,black] (Der3p) -- (Der3m) ;
        \draw [thick,black] (Der4p) -- (Der4m) ;
        \fill (origin) circle (1.2pt);
        \fill (Emit-Blue1) circle (1.2pt);
        \fill (Emit-Blue2) circle (1.2pt);
        \fill (Emit-Green1) circle (1.2pt);
        \fill (Emit-Green2) circle (1.2pt);
     	\fill ([xshift=-2pt,yshift=-1.5pt]interaction1) rectangle ++ (4pt,3pt);
     	\fill ([xshift=-2pt,yshift=-1.5pt]interaction2) rectangle ++ (4pt,3pt);
    \end{tikzpicture}}}%
}
\newcommand{\vertexOneDashedGreenLess}{\ifTikzExternal{\tikzsetnextfilename{vertexOneDashedGreenLess}}\vcenter{\hbox{\begin{tikzpicture}[scale=1,x={(1,0)},y={(0,1)}]
        \coordinate (x1) at (0,0) ;
        \coordinate (x1b) at (-0.05,0) ;
        \coordinate (x2) at  (0,0.5) ;
        \coordinate (x2a) at  (0.4,0.9) ;
        \coordinate (x2b) at  (0,1) ;
        \coordinate (x2c) at  (0.5,.5) ;
        \draw [thick ,red,directed] (x1) -- (x2) ;
        \draw [thick,red,enddirected]  (x2) -- (x2a);
        \draw [thick,blue,densely dashed,enddirected] (x2) -- (x2b);
        \fill (x2) circle (1.5pt);
        \fill[white] (x1b) circle (1pt);
	\end{tikzpicture}}}
}
\newcommand{\vertexGGammaChi}{\ifTikzExternal{\tikzsetnextfilename{vertexGGammaChi}}\vcenter{\hbox{\begin{tikzpicture}[scale=1,x={(1,0)},y={(0,1)}]
        \path 
         (0,0) coordinate (origin)
          +(0,-0.5) coordinate (In-Red) 
          +(0.45,0) coordinate (Out-Green-Dashed) 
          +(0,0.5) coordinate (Out-Blue) 
          +(0.4,0.4) coordinate (Out-Red) 
         ++(-0.3,0) coordinate (interaction) 
          +(-0.4,-0.2) coordinate (In-Green) 
          +(-0.4,0.2) coordinate (Out-Green) 
        ;
        \draw [thick ,red,directed] (In-Red) -- (origin) ;
        \draw [thick,red,enddirected]  (origin) -- (Out-Red);
        \draw [thick,blue,densely dashed,enddirected] (origin) -- (Out-Blue);
        \draw [thick,Green, enddirected, densely dashed ] (origin) -- (Out-Green-Dashed) ;
        \draw [thick,Green,directed] (In-Green) -- (interaction);
        \draw [thick,Green,enddirected] (interaction) -- (Out-Green);
        \draw [thick,black, densely dotted, dash expand off] ($(interaction)+(-0.05,0)$) -- (origin);
        \fill ([xshift=-1.75pt,yshift=-1.25pt]origin) rectangle ++ (3.5pt,2.5pt);
        \fill ([xshift=-1.75pt,yshift=-1.25pt]interaction) rectangle ++ (3.5pt,2.5pt);
     	\node[left,inner sep=0pt,Green] at (Out-Green) {${}^{_\alpha}$};
     	\node[below left,inner sep=0pt,Green] at (In-Green) {${}^{_\alpha}$};
	\end{tikzpicture}
    }}
}
\newcommand{\vertexGGammaPsi}{\ifTikzExternal{\tikzsetnextfilename{vertexGGammaPsi}}\vcenter{\hbox{\begin{tikzpicture}[scale=1,x={(1,0)},y={(0,1)}]
        \path 
         (0,0) coordinate (origin)
          +(0,-0.5) coordinate (In-Red) 
          +(0.45,0) coordinate (Out-Green-Dashed) 
          +(0,0.5) coordinate (Out-Blue) 
          +(0.4,0.4) coordinate (Out-Red) 
         ++(-0.4,0) coordinate (interaction) 
          +(0,-0.5) coordinate (In-Blue1) 
          +(0,0.5) coordinate (Out-Blue1) 
        ;
        \draw [thick ,red,directed] (In-Red) -- (origin) ;
        \draw [thick,red,enddirected]  (origin) -- (Out-Red);
        \draw [thick,blue,densely dashed,enddirected] (origin) -- (Out-Blue);
        \draw [thick,blue,directed] (In-Blue1) -- (interaction);
        \draw [thick,blue,enddirected] (interaction) -- (Out-Blue1);
        \draw [thick,Green, enddirected, densely dashed ] (origin) -- (Out-Green-Dashed) ;
        \draw [thick,black, densely dotted, dash expand off] ($(interaction)+(-0.05,0)$) -- (origin);
        \fill ([xshift=-1.75pt,yshift=-1.25pt]origin) rectangle ++ (3.5pt,2.5pt);
        \fill ([xshift=-1.75pt,yshift=-1.25pt]interaction) rectangle ++ (3.5pt,2.5pt);
	\end{tikzpicture}
    }}
}
\newcommand{\diagGGammaChi}{\ifTikzExternal{\tikzsetnextfilename{diagGGammaChi}}\vcenter{\hbox{\begin{tikzpicture}[scale=1,x={(1,0)},y={(0,1)}]
        \path 
         (0,0) coordinate (origin)
          +(0,-0.5) coordinate (In-Red) 
          +(0,0.5) coordinate (Out-Blue) 
          +(0.4,0.4) coordinate (Out-Red) 
         ++(-0.2,0) coordinate (interaction) 
          +(-0.1,0) coordinate (Out-Green) 
        ;
        \draw [thick ,red,directed] (In-Red) -- (origin) ;
        \draw [thick,red,enddirected]  (origin) -- (Out-Red);
        \draw [thick,blue,densely dashed,enddirected] (origin) -- (Out-Blue);
        \draw [thin,draw=white,double=Green, double distance=0.8pt] ($(origin)+(0.05,0)$) arc[x radius=2mm, y radius=1mm, start angle=45, end angle=-225 ];
        \draw[thick, Green, enddirected]  ($(origin)+(-0.1,-1.75mm)$) -- ($(origin)+(-1.2mm,-1.75mm)$);
        \draw [thick,black, densely dotted, dash expand off] ($(interaction)+(-0.05,0)$) -- (origin);
        \fill ([xshift=-1.75pt,yshift=-1.25pt]origin) rectangle ++ (3.5pt,2.5pt);
        \fill ([xshift=-1.75pt,yshift=-1.25pt]interaction) rectangle ++ (3.5pt,2.5pt);
	\end{tikzpicture}
    }}
}
\newcommand{\diagGammaOneGGammaPsi}{\ifTikzExternal{\tikzsetnextfilename{diagGammaOneGGammaPsi}}\vcenter{\hbox{\begin{tikzpicture}[scale=1,x={(1,0)},y={(0,1)}]
		\path 
         (0,0) coordinate (origin)
          +(-0.2,-0.2) coordinate (In-Red)  
         ++(0.5,0.5) coordinate (Emit-Blue)
          +(0,0.3) coordinate (Out-Blue)
         let \p1=(origin),\p2=(Emit-Blue) in (\x1,\y2) coordinate (interaction)
         (Emit-Blue)++(0.2,0.2) coordinate (Out-Red)
            ;
        \draw [thick, red, directed] (In-Red) -- (origin) ; 
        \draw [thick, red, directed] (origin) -- (Emit-Blue) ; 
        \draw [thick, red, enddirected] (Emit-Blue) -- (Out-Red) ; 
        \draw [thick, blue, directed] (origin) -- (interaction)  ; 
        \draw [thick, blue, enddirected, densely dashed] (Emit-Blue) -- (Out-Blue)  ; 
        \draw [thick, black, densely dotted, dash expand off] (Emit-Blue) -- (interaction); 
        \fill (origin) circle (1.4pt);
        \fill (Emit-Blue) circle (1.4pt);
     	\fill ([xshift=-2pt,yshift=-1.5pt]Emit-Blue) rectangle ++ (4pt,3pt);
     	\fill ([xshift=-2pt,yshift=-1.5pt]interaction) rectangle ++ (4pt,3pt);
    \end{tikzpicture}%
    }}%
}
\newcommand{\diagGammaTwoPsiGGammaPsi}{\ifTikzExternal{\tikzsetnextfilename{diagGammaTwoPsiGGammaPsi}}\vcenter{\hbox{\begin{tikzpicture}[scale=1,x={(1,0)},y={(0,1)}]
		\path 
         (0,0) coordinate (origin)
          +(-0.2,-0.2) coordinate (In-Red)  
         ++(0.25,0) coordinate (dashed)
         ++(0.5,0.5) coordinate (Emit-Blue)
          +(0,0.3) coordinate (Out-Blue)
         let \p1=(origin),\p2=(Emit-Blue) in (\x1,\y2) coordinate (interaction)
         (Emit-Blue)++(0.2,0.2) coordinate (Out-Red)
            ;
        \draw [thick, red, directed] (In-Red) -- (origin) ; 
        \draw [thick, black, densely dotted, dash expand off]  (origin) -- (dashed) node [pos=0.5, sloped, font=\scriptsize\bfseries, black ] {||} ; 
        \draw [thick, red, directed] (dashed) -- (Emit-Blue) ; 
        \draw [thick, red, enddirected] (Emit-Blue) -- (Out-Red) ; 
        \draw [thick, blue, directed] (origin) -- (interaction)  ; 
        \draw [thick, blue, enddirected, densely dashed] (Emit-Blue) -- (Out-Blue)  ; 
        \draw [thick, black, densely dotted, dash expand off] (Emit-Blue) -- (interaction); 
        \fill (origin) circle (1.4pt);
        \fill (dashed) circle (1.4pt);
        \fill (Emit-Blue) circle (1.4pt);
     	\fill ([xshift=-2pt,yshift=-1.5pt]Emit-Blue) rectangle ++ (4pt,3pt);
     	\fill ([xshift=-2pt,yshift=-1.5pt]interaction) rectangle ++ (4pt,3pt);
    \end{tikzpicture}%
    }}%
}
\newcommand{\vertexTwoBlueDashed}{\ifTikzExternal{\tikzsetnextfilename{vertexTwoBlueDashed}}{{\parbox{1.05cm}{{\begin{tikzpicture}[scale=1]
\coordinate (x1) at (0,0) ;
\coordinate (x2) at  (0,0.5) ;
\coordinate (x2a) at  (0,1) ;
\coordinate (x3) at  (0.5,0.5) ;
\coordinate (x4) at  (1,.5) ;
\coordinate (x3a) at  (0.5,1) ;
\coordinate (blank1) at (-0.07,0) ;
\coordinate (blank2) at  (0.57,0) ;
\draw[white] (blank1)--(blank2);
\draw [thick,red] (x1) -- (x2);
\draw [thick,red,enddirected]  (x3) -- (x3a);
\draw [thick,blue,densely dashed, enddirected]  (x2) -- (x2a);
\draw [thick,dotted] (x2) -- (x3) node [pos=0.5, sloped, font=\scriptsize\bfseries, black ] {||} ;
\draw [thick,Green,enddirected] (x3) -- (x4);
\fill (x2) circle (1.5pt);
\fill (x3) circle (1.5pt);
\end{tikzpicture}}}}}}
\newcommand{\diagGammaTwoPsiGammaTwoPsi}{\vcenter{\hbox{\ifTikzExternal{\tikzsetnextfilename{diagGammaTwoPsiGammaTwoPsi}}\begin{tikzpicture}[scale=1,x={(1,0)},y={(0,1)}]
		\path 
         (0,0) coordinate (origin)
          +(-0.2,-0.2) coordinate (In-Red)  
         ++(0.25,0) coordinate (dashed)
         ++(0.5,0.5) coordinate (Emit-Blue)
          +(0,0.3) coordinate (Out-Blue)
         ++(-0.3,0) coordinate (dashed2)
         ++(0,0) coordinate (Emit-Green)
         let \p1=(origin),\p2=(Emit-Green) in (\x1,\y2) coordinate (interaction)
         (Emit-Green)++(-0.2,0.2) coordinate (Out-Red)
            ;
        \draw [thick, red, directed] (In-Red) -- (origin) ; 
        \draw [thick, black, densely dotted, dash expand off]  (origin) -- (dashed) node [pos=0.5, sloped, font=\scriptsize\bfseries, black ] {||} ; 
        \draw [thick, red, directed] (dashed) -- (Emit-Blue) ; 
        \draw [thick, black, densely dotted, dash expand off]  (Emit-Blue) -- (dashed2) node [pos=0.5, sloped, font=\scriptsize\bfseries, black ] {||} ; 
        \draw [thick, red, enddirected] (Emit-Green) -- (Out-Red) ; 
        \draw [thick, blue, directed] (origin) -- (interaction)  ; 
        \draw [thick, blue, enddirected, densely dashed] (Emit-Blue) -- (Out-Blue)  ; 
        \draw [thick, Green, directed] (Emit-Green) -- (interaction); 
        \fill (origin) circle (1.4pt);
        \fill (dashed) circle (1.4pt);
        \fill (Emit-Blue) circle (1.4pt);
        \fill (dashed2) circle (1.4pt);
        \fill (Emit-Green) circle (1.4pt);
     	\fill ([xshift=-2pt,yshift=-1.5pt]interaction) rectangle ++ (4pt,3pt);
    \end{tikzpicture}%
    }}%
}
\newcommand{\diagSimplifyingTadpoleGammaPlus}{\vcenter{\hbox{\ifTikzExternal{\tikzsetnextfilename{diagSimplifyingTadpoleGammaPlus}}\begin{tikzpicture}[scale=1,x={(1,0)},y={(0,1)}]
		\path 
         (0,0) coordinate (In-Red)
         ++(0.2,0.2) coordinate (origin)  
         ++(0.25,0) coordinate (dashed)
         ++(0.3,0.3) coordinate (Emit-Blue)
          +(0,0.3) coordinate (Out-Blue)
         ++(0.25,0) coordinate (dashed2)
         ++(0.4,0.4) coordinate (Emit-Green)
         let \p1=(origin),\p2=(Emit-Green) in (\x1,\y2) coordinate (interaction)
         (Emit-Green)++(0.2,0.2) coordinate (Out-Red)
            ;
        \draw [thick, red, directed] (In-Red) -- (origin) ; 
        \draw [thick, black, densely dotted, dash expand off]  (origin) -- (dashed) node [pos=0.5, sloped, font=\scriptsize\bfseries, black ] {||} ; 
        \draw [thick, red, directed] (dashed) -- (Emit-Blue) ; 
        \draw [thick, black, densely dotted, dash expand off]  (Emit-Blue) -- (dashed2) node [pos=0.5, sloped, font=\scriptsize\bfseries, black ] {||} ; 
        \draw [thick, red, directed] (dashed2) -- (Emit-Green)  node [pos=0.7, sloped, font=\scriptsize\bfseries, black ] {||} ; 
        \draw [thick, red, enddirected] (Emit-Green) -- (Out-Red) ; 
        \draw [thick, blue, directed] (origin) -- (interaction)  ; 
        \draw [thick, blue, enddirected, densely dashed] (Emit-Blue) -- (Out-Blue)  ; 
        \draw [thick, Green, directed] (Emit-Green) -- (interaction); 
        \fill (origin) circle (1.4pt);
        \fill (dashed) circle (1.4pt);
        \fill (Emit-Blue) circle (1.4pt);
        \fill (dashed2) circle (1.4pt);
        \fill (Emit-Green) circle (1.4pt);
     	\fill ([xshift=-2pt,yshift=-1.5pt]interaction) rectangle ++ (4pt,3pt);
    \end{tikzpicture}%
    }}%
}
\newcommand{\diagSimplifyingTadpoleGammaMinus}{\vcenter{\hbox{\ifTikzExternal{\tikzsetnextfilename{diagSimplifyingTadpoleGammaMinus}}\begin{tikzpicture}[scale=1,x={(1,0)},y={(0,1)}]
		\path 
         (0,0) coordinate (In-Red)
         ++(0.2,0.2) coordinate (origin)  
         ++(0.25,0) coordinate (dashed)
         ++(0.3,0.3) coordinate (Emit-Blue)
          +(0,0.3) coordinate (Out-Blue)
         ++(0.25,0) coordinate (dashed2)
         ++(0.4,0.4) coordinate (Emit-Green)
         let \p1=(origin),\p2=(Emit-Green) in (\x1,\y2) coordinate (interaction)
         (Emit-Green)++(0.2,0.2) coordinate (Out-Red)
            ;
        \draw [thick, red, directed] (In-Red) -- (origin) ; 
        \draw [thick, black, densely dotted, dash expand off]  (origin) -- (dashed) node [pos=0.5, sloped, font=\scriptsize\bfseries, black ] {||} ; 
        \draw [thick, red, directed] (dashed) -- (Emit-Blue) ; 
        \draw [thick, black, densely dotted, dash expand off]  (Emit-Blue) -- (dashed2) node [pos=0.5, sloped, font=\scriptsize\bfseries, black ] {||} ; 
        \draw [thick, red, directed] (dashed2) -- (Emit-Green) ; 
        \draw [thick, red, enddirected] (Emit-Green) -- (Out-Red) ; 
        \draw [thick, blue, directed] (origin) -- (interaction)  ; 
        \draw [thick, blue, enddirected, densely dashed] (Emit-Blue) -- (Out-Blue)  ; 
        \draw [thick, Green, directed] (Emit-Green) -- (interaction) node [pos=0.3, sloped, font=\scriptsize\bfseries, black ] {||} ; 
        \fill (origin) circle (1.4pt);
        \fill (dashed) circle (1.4pt);
        \fill (Emit-Blue) circle (1.4pt);
        \fill (dashed2) circle (1.4pt);
        \fill (Emit-Green) circle (1.4pt);
     	\fill ([xshift=-2pt,yshift=-1.5pt]interaction) rectangle ++ (4pt,3pt);
    \end{tikzpicture}%
    }}%
}
\newcommand{\diagGammaOneGammaTwoPsi}{\vcenter{\hbox{\ifTikzExternal{\tikzsetnextfilename{diagGammaOneGammaTwoPsi}}\begin{tikzpicture}[scale=1,x={(1,0)},y={(0,1)}]
		\path 
         (0,0) coordinate (origin)
          +(-0.2,-0.2) coordinate (In-Red)  
         ++(0.7,0.7) coordinate (Emit-Blue)
          +(0,0.3) coordinate (Out-Blue)
         ++(-0.3,0) coordinate (dashed2)
         ++(0,0) coordinate (Emit-Green)
         let \p1=(origin),\p2=(Emit-Green) in (\x1,\y2) coordinate (interaction)
         (Emit-Green)++(-0.2,0.2) coordinate (Out-Red)
            ;
        \draw [thick, red, directed] (In-Red) -- (origin) ; 
        \draw [thick, red, directed] (origin) -- (Emit-Blue) ; 
        \draw [thick, black, densely dotted, dash expand off]  (Emit-Blue) -- (dashed2) node [pos=0.5, sloped, font=\scriptsize\bfseries, black ] {||} ; 
        \draw [thick, red, enddirected] (Emit-Green) -- (Out-Red) ; 
        \draw [thick, blue, directed] (origin) -- (interaction)  ; 
        \draw [thick, blue, enddirected, densely dashed] (Emit-Blue) -- (Out-Blue)  ; 
        \draw [thick, Green, directed] (Emit-Green) -- (interaction); 
        \fill (origin) circle (1.4pt);
        \fill (Emit-Blue) circle (1.4pt);
        \fill (dashed2) circle (1.4pt);
        \fill (Emit-Green) circle (1.4pt);
     	\fill ([xshift=-2pt,yshift=-1.5pt]interaction) rectangle ++ (4pt,3pt);
    \end{tikzpicture}%
    }}%
}
\newcommand{\diagSimplifyingTadpoleGammaOneGammaPlus}{\vcenter{\hbox{\ifTikzExternal{\tikzsetnextfilename{diagSimplifyingTadpoleGammaOneGammaPlus}}\begin{tikzpicture}[scale=1,x={(1,0)},y={(0,1)}]
		\path 
         (0,0) coordinate (In-Red)
         ++(0.2,0.2) coordinate (origin)  
         ++(0.3,0.3) coordinate (Emit-Blue)
          +(0,0.3) coordinate (Out-Blue)
         ++(0.25,0) coordinate (dashed2)
         ++(0.4,0.4) coordinate (Emit-Green)
         let \p1=(origin),\p2=(Emit-Green) in (\x1,\y2) coordinate (interaction)
         (Emit-Green)++(0.2,0.2) coordinate (Out-Red)
            ;
        \draw [thick, red, directed] (In-Red) -- (origin) ; 
        \draw [thick, red, directed] (origin) -- (Emit-Blue) ; 
        \draw [thick, black, densely dotted, dash expand off]  (Emit-Blue) -- (dashed2) node [pos=0.5, sloped, font=\scriptsize\bfseries, black ] {||} ; 
        \draw [thick, red, directed] (dashed2) -- (Emit-Green)  node [pos=0.7, sloped, font=\scriptsize\bfseries, black ] {||} ; 
        \draw [thick, red, enddirected] (Emit-Green) -- (Out-Red) ; 
        \draw [thick, blue, directed] (origin) -- (interaction)  ; 
        \draw [thick, blue, enddirected, densely dashed] (Emit-Blue) -- (Out-Blue)  ; 
        \draw [thick, Green, directed] (Emit-Green) -- (interaction); 
        \fill (origin) circle (1.4pt);
        \fill (Emit-Blue) circle (1.4pt);
        \fill (dashed2) circle (1.4pt);
        \fill (Emit-Green) circle (1.4pt);
     	\fill ([xshift=-2pt,yshift=-1.5pt]interaction) rectangle ++ (4pt,3pt);
    \end{tikzpicture}%
    }}%
}
\newcommand{\diagSimplifyingTadpoleGammaOneGammaMinus}{\vcenter{\hbox{\ifTikzExternal{\tikzsetnextfilename{diagSimplifyingTadpoleGammaOneGammaMinus}}\begin{tikzpicture}[scale=1,x={(1,0)},y={(0,1)}]
		\path 
         (0,0) coordinate (In-Red)
         ++(0.2,0.2) coordinate (origin)  
         ++(0.3,0.3) coordinate (Emit-Blue)
          +(0,0.3) coordinate (Out-Blue)
         ++(0.25,0) coordinate (dashed2)
         ++(0.4,0.4) coordinate (Emit-Green)
         let \p1=(origin),\p2=(Emit-Green) in (\x1,\y2) coordinate (interaction)
         (Emit-Green)++(0.3,0.3) coordinate (Out-Red)
            ;
        \draw [thick, red, directed] (In-Red) -- (origin) ; 
        \draw [thick, red, directed] (origin) -- (Emit-Blue) ; 
        \draw [thick, black, densely dotted, dash expand off]  (Emit-Blue) -- (dashed2) node [pos=0.5, sloped, font=\scriptsize\bfseries, black ] {||} ; 
        \draw [thick, red, directed] (dashed2) -- (Emit-Green) ; 
        \draw [thick, red, enddirected] (Emit-Green) -- (Out-Red) ; 
        \draw [thick, blue, directed] (origin) -- (interaction)  ; 
        \draw [thick, blue, enddirected, densely dashed] (Emit-Blue) -- (Out-Blue)  ; 
        \draw [thick, Green, directed] (Emit-Green) -- (interaction) node [pos=0.3, sloped, font=\scriptsize\bfseries, black ] {||} ; 
        \fill (origin) circle (1.4pt);
        \fill (Emit-Blue) circle (1.4pt);
        \fill (dashed2) circle (1.4pt);
        \fill (Emit-Green) circle (1.4pt);
     	\fill ([xshift=-2pt,yshift=-1.5pt]interaction) rectangle ++ (4pt,3pt);
    \end{tikzpicture}%
    }}%
}
\newcommand{\diagSimplifyingTadpoleGammaTwoPsi}{\vcenter{\hbox{\ifTikzExternal{\tikzsetnextfilename{diagSimplifyingTadpoleGammaTwoPsi}}\begin{tikzpicture}[scale=1,x={(1,0)},y={(0,1)}]
		\path 
         (0,0) coordinate (origin)
          +(-0.2,-0.2) coordinate (In-Red)  
         ++(0.3,0) coordinate (dashed)
         ++(0.7,0.7) coordinate (Emit-Blue)
          +(0,0.3) coordinate (Out-Blue)
         ++(0,0) coordinate (Emit-Green)
         let \p1=(origin),\p2=(Emit-Green) in (\x1,\y2) coordinate (interaction)
         (Emit-Green)++(0.2,0.2) coordinate (Out-Red)
            ;
        \draw [thick, red, directed] (In-Red) -- (origin) ; 
        \draw [thick, black, densely dotted, dash expand off]  (origin) -- (dashed) node [pos=0.5, sloped, font=\scriptsize\bfseries, black ] {||} ; 
        \draw [thick, red, directed] (dashed) -- (Emit-Blue) ; 
        \draw [thick, red, enddirected] (Emit-Green) -- (Out-Red) ; 
        \draw [thick, blue, directed] (origin) -- (interaction)  ; 
        \draw [thick, blue, enddirected, densely dashed] (Emit-Blue) -- (Out-Blue)  ; 
        \draw [thick, Green, directed] (Emit-Green) -- (interaction); 
        \fill (origin) circle (1.4pt);
        \fill (Emit-Blue) circle (1.4pt);
        \fill (dashed) circle (1.4pt);
        \fill (Emit-Green) circle (1.4pt);
     	\fill ([xshift=-2pt,yshift=-1.5pt]interaction) rectangle ++ (4pt,3pt);
    \end{tikzpicture}%
    }}%
}
\newcommand{\diagSimplifyingTadpoleGammaTwoPsiGammaPlus}{\vcenter{\hbox{\ifTikzExternal{\tikzsetnextfilename{diagSimplifyingTadpoleGammaTwoPsiGammaPlus}}\begin{tikzpicture}[scale=1,x={(1,0)},y={(0,1)}]
		\path 
         (0,0) coordinate (In-Red)
         ++(0.2,0.2) coordinate (origin)  
         ++(0.25,0) coordinate (dashed)
         ++(0.3,0.3) coordinate (Emit-Blue)
          +(0,0.3) coordinate (Out-Blue)
         ++(0.4,0.4) coordinate (Emit-Green)
         let \p1=(origin),\p2=(Emit-Green) in (\x1,\y2) coordinate (interaction)
         (Emit-Green)++(0.2,0.2) coordinate (Out-Red)
            ;
        \draw [thick, red, directed] (In-Red) -- (origin) ; 
        \draw [thick, black, densely dotted, dash expand off]  (origin) -- (dashed) node [pos=0.5, sloped, font=\scriptsize\bfseries, black ] {||} ; 
        \draw [thick, red, directed] (dashed) -- (Emit-Blue) ; 
        \draw [thick, red, directed] (Emit-Blue) -- (Emit-Green)  node [pos=0.7, sloped, font=\scriptsize\bfseries, black ] {||} ; 
        \draw [thick, red, enddirected] (Emit-Green) -- (Out-Red) ; 
        \draw [thick, blue, directed] (origin) -- (interaction)  ; 
        \draw [thick, blue, enddirected, densely dashed] (Emit-Blue) -- (Out-Blue)  ; 
        \draw [thick, Green, directed] (Emit-Green) -- (interaction); 
        \fill (origin) circle (1.4pt);
        \fill (Emit-Blue) circle (1.4pt);
        \fill (dashed) circle (1.4pt);
        \fill (Emit-Green) circle (1.4pt);
     	\fill ([xshift=-2pt,yshift=-1.5pt]interaction) rectangle ++ (4pt,3pt);
    \end{tikzpicture}%
    }}%
}
\newcommand{\diagSimplifyingTadpoleGammaTwoPsiGammaMinus}{\vcenter{\hbox{\ifTikzExternal{\tikzsetnextfilename{diagSimplifyingTadpoleGammaTwoPsiGammaMinus}}\begin{tikzpicture}[scale=1,x={(1,0)},y={(0,1)}]
		\path 
         (0,0) coordinate (In-Red)
         ++(0.2,0.2) coordinate (origin)  
         ++(0.25,0) coordinate (dashed)
         ++(0.3,0.3) coordinate (Emit-Blue)
          +(0,0.3) coordinate (Out-Blue)
         ++(0.4,0.4) coordinate (Emit-Green)
         ++(-0.95,0) coordinate (interaction)
         (Emit-Green)++(0.2,0.2) coordinate (Out-Red)
            ;
        \draw [thick, red, directed] (In-Red) -- (origin) ; 
        \draw [thick, black, densely dotted, dash expand off]  (origin) -- (dashed) node [pos=0.5, sloped, font=\scriptsize\bfseries, black ] {||} ; 
        \draw [thick, red, directed] (dashed) -- (Emit-Blue) ; 
        \draw [thick, red, directed] (Emit-Blue) -- (Emit-Green) ; 
        \draw [thick, red, enddirected] (Emit-Green) -- (Out-Red) ; 
        \draw [thick, blue, directed] (origin) -- (interaction)  ; 
        \draw [thick, blue, enddirected, densely dashed, dash expand off] (Emit-Blue) -- (Out-Blue)  ; 
        \draw [thick, Green, directed] (Emit-Green) -- (interaction) node [pos=0.3, sloped, font=\scriptsize\bfseries, black ] {||} ; 
        \fill (origin) circle (1.4pt);
        \fill (dashed) circle (1.4pt);
        \fill (Emit-Blue) circle (1.4pt);
        \fill (Emit-Green) circle (1.4pt);
     	\fill ([xshift=-2pt,yshift=-1.5pt]interaction) rectangle ++ (4pt,3pt);
    \end{tikzpicture}%
    }}%
}
\newcommand{\vertexIntPsiChiRelMean}{\ifTikzExternal{\tikzsetnextfilename{vertexIntPsiChiRelMean}}{\vcenter{\hbox{{\begin{tikzpicture}[scale=1]
\coordinate (x1a) at (0,0) ;
\coordinate (x1b) at  (0,0.5) ;
\coordinate (x1c) at  (0,1) ;
\coordinate (x2a) at  (0.5,0.5) ;
\coordinate (x2c) at  (-0.5,.5) ;
\coordinate (blank1) at (-0.07,0) ;
\coordinate (blank2) at  (0.57,0) ;
\draw[white] (blank1)--(blank2);
\draw [thick,white] (x1b)  --(x1c);
\draw [thick,blue] (x1a)  --(x1b);
\draw [thick,Green] (x2a) -- (x1b);
\fill ([xshift=-2pt,yshift=-1.5pt]x1b) rectangle ++ (4pt,3pt);
\end{tikzpicture}}}}}}
\newcommand{\vertexIntPsiChiRelFluct}{\ifTikzExternal{\tikzsetnextfilename{vertexIntPsiChiRelFluct}}{\vcenter{\hbox{{\begin{tikzpicture}[scale=1]
\coordinate (x1a) at (0,0) ;
\coordinate (x1b) at  (0,0.5) ;
\coordinate (x1c) at  (0,1) ;
\coordinate (x2a) at  (0.5,0.5) ;
\coordinate (x2c) at  (-0.5,.5) ;
\coordinate (blank1) at (-0.07,0) ;
\coordinate (blank2) at  (0.57,0) ;
\draw[white] (blank1)--(blank2);
\draw [thick,white] (x1b)  --(x1c);
\draw [thick,blue] (x1a)  --(x1b);
\draw [thick,Green,enddirected,densely dashed] (x1b) -- (x2c);
\draw [thick,Green] (x2a) -- (x1b);
\fill ([xshift=-2pt,yshift=-1.5pt]x1b) rectangle ++ (4pt,3pt);
\end{tikzpicture}}}}}}

\newlength{\bilderlength}

\begin{document}

\title{\bf\Large Field theories for Laplacian Growth}
\author{\bf\normalsize Paolo Pisapia{$^{1}$}, Assaf Shapira{$^{2}$}, Kay J\"org Wiese{$^{1}$}}
\date{\small {$^{1}$}CNRS-Laboratoire de Physique de l'Ecole Normale Sup\'erieure, PSL Research University, Sorbonne Universit\'e, Universit\'e Paris Cit\'e, 24 rue Lhomond, 75005 Paris, France,\\
\small {$^{2}$}Universit\'e Paris Cit\'e, CNRS, MAP5, 75006 Paris, France.
}

\maketitle


\begin{abstract}

\smallskip
Loop-erased random walks (LERW), the $O(n)$-model at ${n=-2}$ and Laplacian random walks (LRW) are three realizations of the same random process. While this equivalence holds on any graph, renormalization is possible only via the $O(-2)$-model. 
To generalize LRWs to $b$-LRWs or to Diffusion Limited Aggregation (DLA), a field theory directly on the Laplacian growth process is necessary. 
Here we construct an exact lattice action for LRWs and show that its perturbative expansion equals that of LERWs. We then generalize this approach to $b$-LRWs and DLA. 
\end{abstract}

\bigskip

\begin{multicols}{2}
\small 
\renewcommand{\contentsname}{\vspace*{-3.2ex}}

\tableofcontents
\end{multicols}

\begin{figure}[t]
\centerline{\fig{0.5}{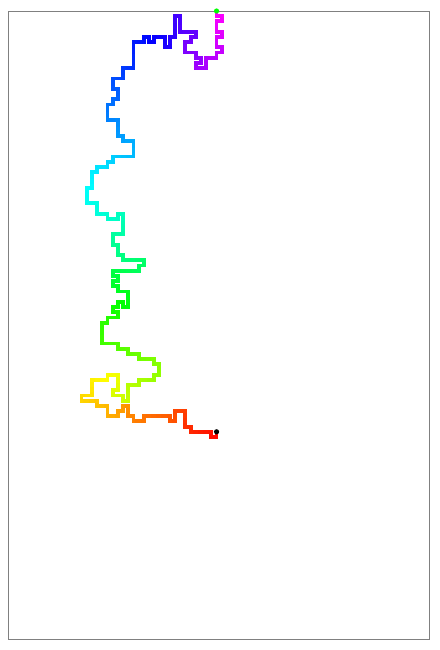}\hfill\fig{0.5}{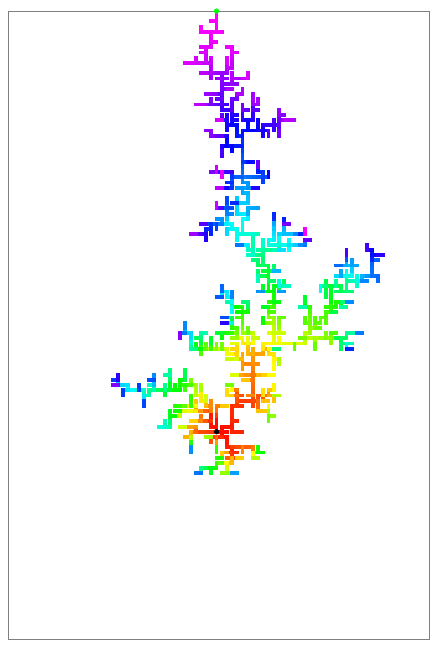}}
\caption{An example of a LRW (left) and a DLA cluster (right) grown in a box of size $80\times 120$. Particles are injected at the target (green dot at top center), and bounce from the walls. The seed is in the center at height $40$ (black dot). The process stops when the target is reached. The color shows time.}
\label{f:DLA}
\end{figure}

\section{Introduction}
\subsection{State of the Art}
Loop-erased random walks (LERWs) were introduced in 1980 by G.~Lawler \cite{Lawler1980} as a mathematically more tractable
alternative to self-avoiding polymers: one follows the trace of a random walker and erases from its trace 
a loop once it is formed. This leads to a non-selfintersecting path in all dimensions. 
In \cite{Lawler1980} Lawler further showed that loop-erased random walks can be reformulated as a growth process, 
termed Laplacian random walk (LRW): 
one solves the Laplace equation $\nabla^2 \Phi(x)=0$, with boundary condition $\Phi(x)=0$ on the 
already constructed curve, and $\Phi(x)=1$ at a prescribed target. Starting at the tip $x$ of the walk, the probability to advance to $y$ is proportional to $\Phi(y)$. The process stops when it has reached the target; see the left of Fig.~\ref{f:DLA}.

This process produces curves with a fractal dimension of $d_{\rm f}^{d=2}=5/4$ in dimension $d=2$, and was identified as a stochastic-L\"owner evolution \cite{Schramm2000,LawlerSchrammWerner2004,Cardy2005} with $\kappa=2$. 
In dimension $d=3$, no exact results are available. The fractal dimension of LERWs was measured numerically as
 $d_{\rm f}^{d=3} = 1.62400\pm 0.00005$ \cite{Wilson2010}.
It was later shown that the $O(n)$-model at ${n=-2}$, LERWs and LRWs are three manifestations of the same random process in all dimensions. 
This was first achieved via a heuristic mapping \cite{WieseFedorenko2018,WieseFedorenko2019}. 
A mathematically rigorous proof was presented in \cite{HelmuthShapira2020}, and later simplified \cite{ShapiraWiese2020}.
The latter formulation allowed the authors of \cite{KompanietsWiese2019} to evaluate 
the fractal dimension to $6$-loop order, $d_{\rm f}^{d=2} = 1.6243(10)$.
Finally, LERWs appear in so varied models as uniform spanning trees \cite{LawlerSchrammWerner2004}, Abelian sandpiles \cite{Dhar1990,MajumdarDhar1992} and the depinning of charge-density waves \cite{NarayanMiddleton1994,WieseFedorenko2018,WieseFedorenko2019}, and more, see Fig.~\ref{CDW-relations}.

While the equivalence between the three models above holds on any finite graphs and in the continuum limit, a renormalization-group treatment is yet available only via the formulation in terms of the $O(-2)$-model \cite{WieseFedorenko2018,WieseFedorenko2019,KompanietsWiese2019}. 
This is unsatisfactory, since there are several generalizations of Laplacian growth of a path, 
for which one needs a field theory directly on the growth process:
the first is the $b$-Laplacian random walk ($b$-LRW) where one takes the solution of the Laplace equation to the power of $b$ \cite{LyklemaEvertszPietronero1986}. In dimension $d=2$, conformal arguments give 
$d_{\rm f} = 1+\frac{3}{4(2b+1)}$ \cite{Hastings2002,Lawler2006}. A field theory for Laplacian growth would allow us to generalize this to dimension $d=3$.

An even more important generalization is the growth of clusters, known as diffusion-limited aggregation (DLA) \cite{WittenSander1981,WittenSander1983}. 
Here the idea is to place a seed for the cluster, start particles far away, allow them to diffuse until they hit the cluster 
upon which they attach. This process is repeated leading to clusters such as the one shown on the right of Fig.~\ref{f:DLA}. While standard numerical simulations stop after a fixed number $N$ of points was attached to the cluster, our stopping condition is that the cluster has reached the target. This kind of stopping condition is much more natural in the combinatorial approach we follow here.

As the diffusing particle solves the Laplace equation, the cluster can be obtained by the algorithm outlined above for Laplacian random walks, except that now the growth is not restricted to the tip of a path.
In dimension ${d=2}$, DLA has a fractal dimension of $d_{\rm f}^{\textsc{DLA}}|_{d=2}\approx 1.71$.
Key observations are that there are different fractal dimensions, depending on the observable and where in the cluster it is measured. This is known as multifractality. As we will not give results for DLA here, 
we refer the reader to the literature \cite{Meakin1983,Meakin1983b,NiemeyerPietroneroWiesmann1984,HalseyMeakinProcaccia1986,HalseyJensenKadanoffProcacciaShraiman1986,MeakinTolman1989,SanchezGuineaSanderHakimLouis1993,Halsey1994,MandelbrotVespignaniKaufman1995,Halsey2000,LopezPimentel2017,BergerProcacciaTurner2020,ProcacciaProcaccia2021,TentiHernandez-GuianceIrurzun2021}.
Apart from the beautiful growth process of Hastings-Levitov \cite{HastingsLevitov1998}, no predictive theory exists today, especially there is no analytic prediction for the fractal dimension of DLA beyond simple mean-field arguments \cite{WittenSander1981,WittenSander1983,BogoyavlenskiyChernova2000}. 
Furthermore, dielectric breakdown models are governed by the same equation, replacing 
the solution of the Laplace equation by its $b$-th power \cite{NiemeyerPietroneroWiesmann1984}, as for the $b$-Laplacian random walk.

In order to treat these more general models via field theory, one needs the latter defined directly as a growth process. 
The aim of this work is to achieve this. We first construct an exact lattice action for Laplacian random walks. It is a dynamical field theory, which evolves an ensemble of configurations. Starting with a seed, the process adds vertices according to the rules of Laplacian growth. It stops when the target is reached. 
When all configurations have reached the target, the result is the sum over all LRWs with their respective probabilities.
In a second step, we take the continuum limit. We check to 3-loop order that for LRWs we get the same diagrams as for LERWs. 

We then generalize our dynamic field theory to $b$-Laplacian random walks and DLA. The treatment of these two theories is relegated to future work. 

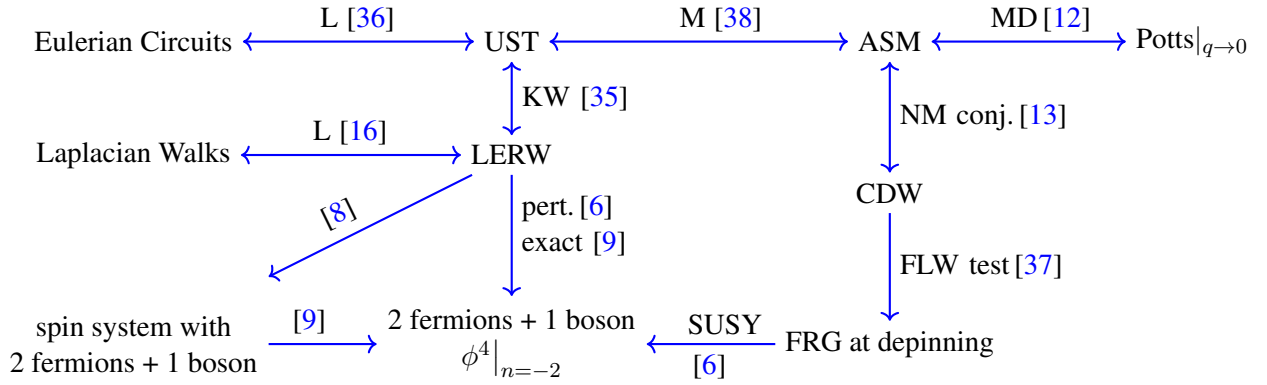
\begin{figure}[t]
\begin{center}\small
\ifTikzExternal{\tikzexternaldisable}
{\parbox{16.5cm}{{\begin{tikzpicture}
\node (t0) at (-1,4) {Eulerian Circuits} ;
\node (t1) at (4,2.5) {LERW} ;
\node (t2) at (-1,2.5) {Laplacian Walks} ;
\node (t3) at (4,4) {UST} ;
\node (t4) at (9,4) {ASM} ;
\node (t5) at (13,4) {Potts$|_{q\to 0}$} ;
\node (t6) at (9,2) {CDW} ;
\node (t7) at (9,0) {FRG at depinning} ;
\coordinate (t8bis) at (4,0) ;
\node (t8) at (4,0) {{\parbox[c][1.cm][c]{3.3cm}{\begin{center}2 fermions + 1 boson\\$\phi^{4}\big|_{n=-2}$\end{center}}}};
\node (t9) at (-1,0) {{\parbox{3.3cm}{\begin{center}spin system with \\ 2 fermions + 1 boson\end{center}}}};
\draw [thick,blue,<->] (t0) -- (t3);
\draw [thick,blue,<->] (t1) -- (t2);
\draw [thick,blue,<->] (t3) -- (t4);
\draw [thick,blue,<->] (t5) -- (t4);
\draw [thick,blue,<->] (t6) -- (t4);
\draw [thick,blue,<->] (t3) -- (t1);
\draw [thick,blue,->] (t6) -- (t7) ;
\draw [thick,blue,->] (t7) -- (t8);
\draw [thick,blue,->] (t1) -- (t8);
\draw [thick,blue,->] (t1) -- (t9);
\draw [thick,blue,->] (t9) -- (t8);
\path (t7) -- node [text width=1cm,midway,above ] {~\,SUSY} (t8);
\path (t7) -- node [text width=1cm,midway,below ] {~~\cite{WieseFedorenko2018}} (t8);
\path (t8) -- node [text width=0.8cm,midway,above ] {\cite{ShapiraWiese2020}} (t9);
\path (t1) -- node [text width=0.5cm,near end,above ] {\rotatebox{26}{~~\cite{HelmuthShapira2020}}} (t9);
\path (t1) -- node [text width=2cm,midway,right ] {KW \cite{KenyonWilson2015}} (t3);
\path (t1) -- node [text width=1cm,midway,above ] {L\;\cite{Lawler2006}} (t2);
\path (t0) -- node [text width=1cm,midway,above ] {L\;\cite{Kasteleyn1967}} (t3);
\path (t1) -- node [text width=2cm,pos=0.3,right ] {pert.\,\cite{WieseFedorenko2018}
\\exact~\cite{ShapiraWiese2020}} (t8bis);
\path (t4) -- node [text width=1cm,midway,above ] {MD\,\cite{MajumdarDhar1992}} (t5);
\path (t4) -- node [text width=3cm,midway,right ] {NM conj.\,\cite{NarayanMiddleton1994}} (t6);
\path (t6) -- node [text width=3cm,midway,right ] {FLW test\,\cite{FedorenkoLeDoussalWiese2008a}} (t7);
\path (t3) -- node [text width=.5cm,midway,above ] {M\;\cite{Majumdar1992}} (t4);
\end{tikzpicture}}}}
\end{center}
\vspace*{-1ex}
\caption{Relations between Laplacian walks, loop-erased random walks (LERW), uniform spanning trees (UST), Eulerian Circuits, the Abelian Sandpile Model (ASM), the Potts model in the limit of $q\to 0$, charge-density waves at depinning (CDW), mapping onto the FRG field theory at depinning, reducing to $\phi^4$-theory at $n=-2$, and equivalent to an interacting theory of 2 complex fermions and one complex boson.}
\label{CDW-relations}
\end{figure}
\ifTikzExternal{\tikzexternalenable}

\subsection{Outline}
This article is organized as follows. First, we pose some definitions in section \ref{Graph-theoretical definitions}.
We then review the mapping of loop-erased random walks onto the $O(n)$-model at ${n=-2}$
(Viennot's theorem, section \ref{s:Helmut-Shapira-Wiese theorem}), and its lattice field theory (section \ref{Action for LERW field theory}).
We show how this lattice field theory solves the Laplace equation with excluded vertices (section \ref{Solving the Laplace equation with excluded vertices}). 
This formulation allows us to give a new proof for the equivalence of LERWs and Laplacian Growth (Lawler's theorem \cite{Lawler1980}) in section \ref{Proof of Lawler's theorem from the lattice field theory}.

In section \ref{Laplacian Random Walks} we construct a field theory for Laplacian random walks. The first step is to define appropriate growth rates (section \ref{LRW as a growth process}). 
We then construct a dynamic lattice action, which propagates an ensemble of paths until they hit a prescribed target (section \ref{A first lattice field theory}). This is tested in section \ref{Test of lattice field theory}. 
Its continuum limit is constructed in section \ref{Continuum limit for L-LERW-1}. Propagators and vertices are discussed in section \ref{Propagators and vertices}, followed by power counting (section \ref{Power counting}).
The theory is simplified in two steps (sections \ref{A simpler lattice action} and \ref{Continuum limit for L-LERW-2}). 
The branching vertex is interpreted as the diffusion and drift of the tip of the curve in section \ref{Diffusion and drift}. The probability to pass through a given point is introduced in section \ref{Observable: passing through a point}.
This observable is perturbatively corrected (section \ref{1-loop corrections to the observer}).
In section \ref{Simplifying the chi-theory} we give a final simplification of our perturbative action. 
The latter is generalized to the $b$-Laplacian random walk in section \ref{b-Laplacian Random Walks}. 
 DLA is discussed in section \ref{DLA}. 
We start with its definition (section \ref{Definition}), followed by an example (section \ref{DLA-example}), and a static conjecture (section \ref{DLA-conjecture}) which we show to be incorrect. 
A proper dynamic theory is presented in section \ref{Lattice action for DLA as a growth process}, followed by a check on a small graph (section \ref{Numerical check}). This allows us to write an improved action in section \ref{DLA:Improved lattice action}. The analysis of its field theory will be discussed elsewhere \cite{PisapiaWiese2027}. 

In appendix \ref{a:Perturbative calculations} we study perturbation theory, showing which vertices are irrelevant, and how tadpoles cancel. Appendix \ref{app:pOperation Proof} shows how our field theory simplifies, which allows us to assert in appendix \ref{Observable including sub-loops up to 3-loop order} that the observer of LERWs is reproduced at least to 3-loop order.

\section{Models on a graph}
\label{Models on a graph}

\subsection{Graph-theoretical definitions}
\label{Graph-theoretical definitions}

\noindent{\underline{\bf Directed Graph}:} A {\em directed graph} $\ca G = \{\ca V, \ca E\}$ is a collection of {\em vertices} or {\em points in space} $\ca V= \{ v_1,..., v_n\}$ and {\em edges} $\ca E$, s.t.\ each edge $\vec e\in \ca E$ is an ordered pair of vertices $\vec e = \{v_i,v_j \}$, $1\le i,j\le n$. We note $\mathbb V(\ca G):= \ca V$, and $\mathbb E(\ca G):= \ca E$. We demand both $\ca V$ and $\ca E$ to contain only distinct elements. 

\smallskip
\noindent{\underline{\bf Metric of a directed Graph}:} A directed graph $\ca G$ can be endowed with a rate (metric) $\beta$, by defining for each edge $\vec e=(x,y)$, $\beta_{\vec e}\equiv \beta_{xy} > 0$. We assign a rate $\beta_{xy}=0$ to all pairs $( x, y ) \notin \ca E$. 

\smallskip
\noindent{\underline{\bf Undirected Graph}:} An undirected graph $\ca G$ is a collection of vertices $\ca V= \{ v_1,..., v_n\}$ and edges $\ca E$, s.t.\ each edge $ e\in \ca E$ is an unordered pair of vertices $ e = (v_i,v_j )$. Its metric is $\beta_{e}= \beta_{xy} = \beta_{yx}$. Pairs $( x,y) \in \ca E$ are {\em neighbors}. 

\smallskip
\noindent{\underline{\bf Uniform Graph}:} A uniform graph is an undirected graph for which all $\beta_{e}$ are equal, 
 $\beta_{e} = \beta ~~ {\forall e\in \ca E}$.

\smallskip
\noindent{\underline{\bf Subgraph}:} A {\em subgraph} $\ca G' = \{\ca V',\ca E' \} \subset \ca G= \{\ca V, \ca E\}$ is a graph with $\ca V' \subset \ca V$, and $\ca E' \subset \ca E$. 

\smallskip
\noindent{\underline{\bf Connected Graph}:} A graph $\ca G = \{\ca V, \ca E\}$ is {\em connected} if for each pair of vertices $\{x,y\}\subset \ca V$ there exists a {\em path} $( v_1= x, v_2, ..., v_n = y )$, s.t.\ $(v_i, v_{i+1}) \in \ca E$ for all $1\le i < n$. We will only consider connected graphs. 

\smallskip
\noindent{\underline{\bf Tree}:} A {\em tree} is a connected graph for which the sequence $( v_1= x, v_2, ..., v_n = y )$ is unique. 

\smallskip
\noindent{\underline{\bf Sub-Tree}:} A {\em sub-tree} of $\ca G$ is a tree $\ca T \subset \ca G$. 

\smallskip
\noindent{\underline{\bf Directed Tree}:} A {\em directed tree} is a tree for which each edge is directed. 

\smallskip
\noindent{\underline{\bf Rooted Tree}:} A {\em rooted tree} is a tree $\ca T =\{\ca V',\ca E'\}$, with root $\xr$, s.t.\ for each $x\in \ca V'$ there is a path $ \{ v_1= \xr, v_2, ..., v_n = x \}$ s.t.\ all edges $\{v_i, v_{i+1}\} \in \ca E'$ for all $1\le i < n$.
Choosing a tree $\ca T = \{ \ca V', \ca E' \} \subset \ca G$ as part of an undirected graph $\ca G$ endows it with a direction. 

\smallskip
\noindent{\underline{\bf Vertices of a tree}:} Define 
$\mathbb V(\ca T)$ to be the set of all vertices in $\ca T$. 

\smallskip
\noindent{\underline{\bf Random walk}:} A random walk on a graph $\ca G$ with start point (root) $\xr \in \ca V$ is a sequence 
$( v_1 = \xr, v_2, ... )$, with transition rates $\beta_{xy}$ to jump from $x$ to $y$, and $\lambda_x=m_x^2$ to be absorbed at site $x$. 
Denote by $r_x = \lambda_x+\sum_y \beta_{xy} $ the total rate at which the walk exits from vertex $x$, and 
$\hat \beta_{xy}$ the probability to move from $x$ to $y$, and $\hat \lambda_x$ the probability to be absorbed at $x$,
\be
\hat \beta_{xy}:= \frac{\beta_{xy}}{r_x}, \quad \hat \lambda_x := \frac{\lambda_x}{r_x}.
\ee 

\smallskip
\noindent{\underline{\bf Path}:} 
We define a {\em path} $\omega$ to be a sequence of vertices, denoted $\omega=(\omega_1=\xr,\dots,\omega_n)$. 
The probability $\ca P(\omega)$ that the random walk selects the path $\omega$ and then stops is
\begin{eqnarray}\label{1}
\ca P(\omega) &=& \hat \lambda_{\omega_n} q(\omega ) , \\
 q(\omega ) &=& \hat \beta_{\omega_1 \omega_2} \hat\beta_{\omega_2 \omega_3} \dots\hat\beta_{\omega_{n-1} \omega_n}. 
\end{eqnarray}

\smallskip
\noindent{\underline{\bf Weight function}:}
We call $q(\omega)$ the {\em weight function}. 
This weight factorizes: if $\omega^{(a)} = (\omega_1^{(a)},\dots,\omega_n^{(a)})$ and $\omega^{(b)} = (\omega_1^{(b)},\dots,\omega_m^{(b)})$, with $\omega_n^{(a)} = \omega_1^{(b)}$, then the composition 
 $ \omega := \omega^{(a)} \circ \omega^{(b)} = (\omega_1^{(a)},\dots,\omega_n^{(a)} =\omega_1^{(b)},\dots,\omega_m^{(b)})$
 of the paths $ \omega^{(a)}$ and $ \omega^{(b)} $ has weight $q(\omega) = q (\omega^{(a)}) q (\omega^{(b)}) $.

\smallskip
\noindent{\underline{\bf Loop-erased random walk}:}
Given a path $\omega=(\omega_1,\dots,\omega_n)$, we define the \emph{loop-erasure} procedure. The loop erasure is obtained by iterating the \emph{one-loop erasure}: look for the first time $i$ at which the path repeats a vertex, so ${\omega_i = \omega_j}$ for some $j<i$. The one-loop erasure of $\omega$ is the path $(\omega_1,\dots,\omega_j,\omega_{i+1},\dots,\omega_n)$. We then apply the one-loop erasure to this new path, and continue until the path has no repeating vertices. The resulting path is the {\em loop erasure} of $\omega$, denoted $\LE(\omega)$. It is self-avoiding, and when the initial path $\omega$ is a random walk, it is called the {\em loop-erased random walk} (LERW).
If $\gamma$ is a LERW ending at $x$, the probability to generate it is 
\be
\mathcal P(\gamma ) = \hat \lambda_x \sum_{\omega:\LE(\omega)=\gamma} q(\omega) \ .
\ee

\smallskip
\noindent{\underline{\bf Lattice Laplacian}:} Be $f(x)$ a function defined for $x\in\ca V$. 
Then
\be\label{lat-Laplace}
\Delta^\beta f(x):= \sum _y \beta_{xy} \left[f(y) - f(x)\right] 
\ee
is the lattice Laplacian induced by $\beta$, applied to $f$. When $\beta_{xy}$ is indicator of neighbors, this becomes the standard Laplacian and is denoted $\nabla^2 $.

\subsection{Viennot's theorem}
\label{s:Helmut-Shapira-Wiese theorem}
We now state Viennot's theorem, closely following \cite{ShapiraWiese2020}.

\smallskip
\noindent{\underline{\bf Loop}:}
A {\em loop} $C$ is a path $C=(\omega_1,\dots,\omega_{n-1},\omega_n=\omega_1)$ where the first and last points are identical. We also require all vertices to be distinct (except $\omega_1$ and $\omega_n$), so it cannot be decomposed into smaller loops. Loops obtained from each other via cyclic permutations (dropping the repeated vertex $\omega_n$) are considered identical.

\smallskip
\noindent{\underline{\bf Collection of disjoint loops}:}
 By a {\em collection of disjoint loops} we mean a set $L=\{C_1,C_2,\dots\}$, each of whose elements is a loop, and the intersection of any pair of loops in $L$ is empty. We denote the {\em set of all such collections} by $\cc L$.

\smallskip
In order to formulate the theorem, we fix a self-avoiding path $\gamma$.
We define the set $\cc L_\gamma$ to consist of the collections of disjoint loops in which no loop intersects $\gamma$.
Then Viennot's theorem can be written as ($|L|$ being the number of loops)
\begin{equation}\label{eq:Viennot}
\ca A(\gamma):=q(\gamma) \sum_{L\in \cc L_\gamma} (-1)^{|L|}\prod_{C\in L}q(C) =\quad \sum_{\mathclap{\omega:\LE(\omega)=\gamma}} q(\omega) \times \sum_{L\in \cc L}(-1)^{|L|}\prod_{C\in L}q(C).
\end{equation}
On the l.h.s.\ one sums over the ensemble of collections of loops which do not intersect $\gamma$, giving each collection a weight $(-1)^{|L|}\prod_{C\in L}q(C)$.
We later calculate this object using field-theory. The r.h.s.\ contains two factors. The first is the weight to find the LERW path $\gamma$, our object of interest. The second is the partition function 
\be\label{Z}
{\cal Z }:= \sum_{L\in \cc L}(-1)^{|L|}\prod_{C\in L}q(C)
\ee of the loop model on the left-hand side.
Assuming the walk to stop at $\xt$, this relation can be read as 
\be\label{Pgamma}
\mathcal P(\gamma ) = \hat \lambda_{\xt} \sum_{\mathclap{\omega:\LE(\omega)=\gamma}} q(\omega) = \hat \lambda_{\xt} \frac{q(\gamma) \sum_{L\in \cc L_\gamma} (-1)^{|L|}\prod_{C\in L}q(C)}{\cal Z}.
\ee
The probability to go from the root $\xr$ to the target $\xt$ is the sum over all non-selfintersecting paths $\gamma$ starting at $\xr$ and going to $\xt$, 
\be\label{Prt}
\ca P_{\xr ,\xt} = \sum_{\gamma= \{ \xr, ..., \xt\} } \quad ~\mathllap{\mathcal P(\gamma )} .
\ee
For a proof and examples see section 3 of \cite{ShapiraWiese2020}.

\subsection{Action for LERW field theory}
\label{Action for LERW field theory}

\noindent{\underline{\bf Shapira-Wiese action}:}
Here we give the action of Ref.~\cite{ShapiraWiese2020}. 
To this aim, consider an interacting field theory with a standard pair of complex conjugate fields $\{\chi_i(x)\}_{x\in \ca G}$, with $i=1$ bosonic and $i=2,3$ fermionic fields. 
The action on the graph $\ca G$ is
\begin{equation} \label{eq:S_boson}
\rme^{-\ca S'_{\ca G}} = \prod_{x\in\ca V} \Big[ \rme^{- \sum_i r_x \tilde\chi_i(x)\chi_i(x)} \Big(1 + \sum_{y \in \ca V} \sum_i\beta_{xy} \tilde\chi_i(y)\chi_i(x) \Big)\Big].
\end{equation}
We can eliminate $r_x$ by rescaling $\chi_i(x) \to \chi_i(x)/r_x$. This leads to 
\begin{equation} \label{eq:S_boson2}
\rme^{-\ca S_{\ca G}} = \prod_{x \in \ca V} \Big[ \rme^{- \sum_i \tilde\chi_i(x)\chi_i(x)} \Big(1 + \sum_{y\in \ca V}\sum_i {\hat \beta_{xy}} \tilde\chi_i(y)\chi_i(x) \Big)\Big].
\end{equation}
This should be compared to Eq.~(44) of \cite{ShapiraWiese2020}, with the caveat that the latter was written for a single bosonic field. 
Note that we use different normalizations to eliminate the factor of $r_x$ from the action\footnote{We mostly set $\lambda_x\to 0$. A constant $\lambda_x>0$ gives a massive field theory, useful for IR regularization.}.

This action is easy to understand: the Gaussian measure $\rme^{- \sum_i \tilde\chi_i(x)\chi_i(x)}$ implies that expectations factorize at each vertex $x$, and that 
\be
\begin{aligned}
\left< \chi_j(x)^m \tilde\chi_i(x)^n \right>_0 &= n! \delta_{m,n}\delta_{i,j} && \mbox{~(bosons)},\\
\left< \chi_j(x)^m \tilde\chi_i(x)^n \right>_0 &= \delta_{n,m}\delta_{i,j} \Theta(n<2) && \mbox{(fermions).}
\end{aligned}
\ee
The factor inside the big round brackets in \Eq{eq:S_boson2} is such that each vertex has an out-degree of maximally one, i.e.\ $\tilde\chi_i(x)$ appears at most with a power of one, and for a single index $i$. 
As a consequence, action \eq{eq:S_boson2} produces the paths for \Eq{Pgamma} with the proper weights. 

Finally note that on a finite graph, we can use $N$ complex bosons, to calculate the partition function $\ca Z$ or an observable $\ca O$ for any $N$. The result will be a polynomial in $N$, which can be continued analytically to $N=-1$. This construction with $N\to -1$ complex bosons, which corresponds to $n=2N\to -2$ real bosons, is equivalent to the one bosonic and two fermionic fields used in \Eq{eq:S_boson2}.
Since bosons are easier to handle, we will use this formulation for practical calculations, both on a finite graph as for the field theory. 

\smallskip
\noindent{\underline{\bf Transition probability and observable}:}
According to \Eqs{Pgamma} and \eq{Prt} 
\be\label{12}
\ca P_{\xr, \xt} = \frac{\left< \chi_1(\xt) \tilde \chi_1(\xr) \right> }{\ca Z}, 
\ee
where $\ca Z = \left< 1\right>$ is the partition function given in \Eq Z. If $\lambda_x=0$ for all $x$, then $\ca P_{\xr, \xt} =1$.

In order to assess whether a point $x$ belongs to a loop-erased random walk from $\xr$ to $\xt$ after erasure, we fix the three vertices $\xr, x$ and $\xt$ and consider the observable
\begin{equation}\label{eq:U}
\ca O(\xr, x, \xt):= \frac{ \left< \chi_2(\xt) \tilde \chi_2(x) \chi_1(x) \tilde \chi_1(\xr) \right>}{\ca Z} .
\end{equation}

\subsection{Solving the Laplace equation with excluded vertices}
\label{Solving the Laplace equation with excluded vertices}
To treat Laplacian growth models we want to solve the Laplace-equation
\bea \label{eq:laplace}
\nn \Delta^\beta_x \Phi_{\ca V | \ca E}(x)&=&0, \qquad ~\,x \in \ca V \setminus \ca E, \\
\Phi(x) &=& a_x, \qquad x \in \ca E.
\eea
We claim that similarly to \Eq{eq:S_boson2} this can be obtained from the action
\be \label{eq:S_boson4}
\rme^{-\ca S_{\ca V| \ca E}}= \prod_{x \in \ca V } \rme^{- \sum_i \tilde\chi_i(x)\chi_i(x)} \prod_{x \in \ca V \backslash \ca E} \Big(1 {+} \sum_{y\in \ca V} \sum_i {\hat\beta_{xy}} \tilde\chi_i(y)\chi_i(x) \Big). 
\ee
Then
\bea \label{potential-Phi}
\nn \Phi_{\ca V | \ca E}(x) &:=& \frac1{\ca Z_{\ca V| \ca E}} \sum_{y \in \ca E} \left< \tilde\chi_1(x) \chi_1 (y) \right> a_y \\
&=& \frac1{\ca Z_{\ca V| \ca E}} 
\int_{\tilde \chi,\chi } \rme^{-\ca S_{\ca V| \ca E}}\tilde\chi_1 (x) \left[\sum_y a_y \chi_1(y)\right]
\eea
solves \Eq{eq:laplace}.
To prove this, we first show
\be
\Delta^\beta_x \left< \tilde\chi_1(x) \chi_1 (y) \right> = 0 \quad \forall x\in \ca V\backslash \ca E, ~y \in \ca V. 
\ee
This is a result of the 2-point function being the sum over all LERWs from $x$ to $y$ in the $\ca O(-2)$-theory restricted to $\ca V\backslash \ca E$, or equivalently the sum over all random walks from $x$ to $y$ in the free theory. 
As a result, $\Phi_{\ca V | \ca E}$ satisfies the Laplace equation for $x\not \in \ca E$. 
Let us now evaluate it for $x \in \ca E$, 
\bea
\Phi_{\ca V | \ca E}(x) &=& \frac1{\ca Z_{\ca V| \ca E}} \sum_{y \in \ca E} \left< \tilde\chi_1(x) \chi_1 (y) \right> a_y \nn\\
&=& \frac1{\ca Z_{\ca V| \ca E}} \left< \tilde\chi_1(x) \chi_1 (x) \right> a_x 
= a_x.
\eea
We used that the action \eq{eq:S_boson4} has no  outgoing edge for $y \in \ca E$, thus the only non-vanishing contraction of $\chi_1(y)$ with $y \in \ca E$ is with $\tilde\chi_1(x)$, which leads to the second equality. What remains are loops in the theory \eq{eq:S_boson4}, i.e.\ $\ca Z_{\ca V| \ca E}$, proving the last equality.

To conclude this section, let us make a specific choice for $\ca E$ and $a_x$: we take $\ca E = \gamma \cup \ca T$, 
where $\gamma$ is a path, and $\ca T$ the target, specified above to $\ca T = \{ \xt \}$. We set 
 $a_x=0$ for $x\in \gamma$, and $a_x=1$ for $x\in \ca T$. We use the notation $\left< \cdot \right>_{\ca V\backslash \gamma} := \left< \cdot \right>_{\ca V' | \ca T}$ for $\ca V' := \ca V \setminus \gamma$.
Then 
\be
 \left< \tilde\chi_1(x) \chi_1 (\xt) \right>_{\ca V\backslash \gamma}=
 \left< \tilde\chi_1(x) \chi_1 (\xt) \right>_{\ca V| \ca E}, \qquad \mbox{and}\qquad 
\ca Z_{\ca V\backslash \gamma} = \ca Z_{\ca V| \ca E}.
\ee
Formula \eq{potential-Phi} can now be written as 
\be
\label{potential-Phi-2}
 \Phi_{\ca V | \ca E}(x) = \frac1{\ca Z_{\ca V| \ca E}} \left< \tilde\chi_1(x) \chi_1 (\xt) \right>_{\ca V| \ca E} = \frac1{\ca Z_{\ca V\backslash \gamma }} \left< \tilde\chi_1(x) \chi_1 (\xt) \right>_{\ca V \backslash \gamma} .
\ee

\subsection{Proof of Lawler's theorem from the lattice field theory}
\label{Proof of Lawler's theorem from the lattice field theory}
The field theory \eq{eq:S_boson2} calculates the probability for all paths $\gamma$ from $\xr$ to $\xt$. This allows us to evaluate the probability to have a path beginning with $\omega_1=(\xt, ...,v )$ and compare it to the probability for a path of one more step $\omega_2 = ( \xt, ..., v,x )$. 
Denote $\ca V_1 := \ca V\backslash (\omega_1 \backslash \{v\})$ and $\ca V_2 := \ca V \backslash \omega_1$.
Then 
\be
\ca P(\omega_1) = \sum_{\gamma \supset \omega_1} \ca P(\gamma) = q(\omega_1) \frac{\left< \tilde \chi_1(v) \chi_1(\xt) \right>_{\ca V_1} }{\ca Z} .
\ee
The same holds for $\omega_2$. 
The ratio of the two probabilities is given by 
\be\label{35}
 \frac{\ca P(\omega_2)}{\ca P(\omega_1)}
= \frac{\hat \beta_{vx}\int_\chi \rme^{-\ca S_{\ca V_2 }} \tilde\chi_1(x) \chi_1(\xt )}{\int_\chi \rme^{-\ca S_{\ca V_1}} \tilde\chi_1(v) \chi_1(\xt)} .
\ee
Let us write the denominator in terms of the action in the numerator, making explicit the additional factor in the interaction for vertex $v$, 
\bea
{\int_\chi \rme^{-\ca S_{\ca V_1}} \tilde\chi_1(v) \chi_1(\xt)}&=& \int_{\chi|\ca V_2} \int_{\chi_v}\rme^{-\ca S_{\ca V_2} - \sum_i \tilde\chi_i(v)\chi_i(v)} \Big(1 + \sum_{y\in \ca V}\sum_i {\hat \beta_{vy}} \tilde\chi_i(y)\chi_i(v) \Big) \tilde\chi_1(v) \chi_1(\xt) \nn \\
&=& \int_{\chi|\ca V_2} \rme^{-\ca S_{\ca V_2}} \sum_{y\in \ca V} {\hat \beta_{vy}} \tilde\chi_1(y)
 \chi_1(\xt)\nn\\
 &=& \sum_y \hat \beta_{v y} \left< \tilde\chi_1(y) \chi_1 (\xt) \right>_{\ca V_2} .
 \label{36}
\eea
Summing the numerator of the r.h.s.~of \Eq{35} over all possible $x$ (i.e.\ different choices of $\omega_2$) gives its denominator, 
\be\label{37}
\ca P (\omega_1) = \sum_{x\in \ca V} \ca P(\omega_1+\{x\}).
\ee
As expected, $\ca P(\omega_2)$ is {\em properly normalized}. 
With the definition given in \Eq{potential-Phi-2}, 
\Eq{35} can be written as 
\be \label{24}
 \frac{\ca P(\omega_2 )}{ \ca P(\omega_1)} = \frac{ \hat \beta_{vx} \Phi_{\ca V|\ca E}(x)}
{ \sum_{y\in \ca V} \hat \beta_{vy} \Phi_{\ca V|\ca E}(y)},
\ee
with $\ca E = \omega_1 \cup \ca T$. The boundary condition is $0$ on $\omega_1$ and $1$ on $\ca T$.

This is Lawler's theorem: it states that LERW configurations have the same probability as the Laplacian random walk defined by \Eq{24}.

\section{Laplacian random walks}
\label{Laplacian Random Walks}

\subsection{LRW as a growth process}
\label{LRW as a growth process}
In the last section we studied the generating function of LERWs, or equivalently LRWs starting at the root $\xr$ and stopping at the target $\xt$. The aim of this section is to reformulate LRWs as a dynamic process in time $t$. According to \Eq{24}, the probability to grow from the tip at $v$ is proportional to the solution of the Laplace equation at time $t$,
\be
\ca P _{v \to x}|_t = \frac{ \beta_{vx} \Phi_{\ca V| \ca E}(x,t)}{\ca N_v(t)} .
\ee
\Eq{24} uses the normalization $\ca N_v(t) = \sum_{y\in \ca V} \beta_{vy} \Phi_{\ca V| \ca E}(y,t)$.
When formulating LRWs as a growth process, it is not only difficult to calculate the inverse of $\ca N_v$ in a field theory but even futile. Instead we use transition rates in continuous time, 
\be\label{rate}
\rho _{v \to x}|_t = \gamma \beta_{vx} \Phi_{\ca V| \ca E}(x,t) .
\ee
As long as $\beta_{xy}=\beta_{yx}$, \Eq{rate} is also the appropriate rate for DLA: $\Phi_{\ca V| \ca E}(x) $ is the probability for a particle to be at $x$ without having been absorbed by $\ca E$, while $\gamma \beta_{xv}$ is the rate to then be absorbed at $v$. 
In addition, the rate \eq{rate} will allow us to address the temporal evolution of the cluster. 
For all these reasons, it is this rate we shall use in our field theory.

\subsection{A first lattice field theory $\ca L_1^{\textsc{LRW}}$}
\label{A first lattice field theory}
To construct a lattice field theory for the growth of Laplacian random walks, we introduce three pairs of complex fields. Each of these depends both on space $x$ (the vertex) and time $t$. The latter increases in steps of $\delta t=1$. The object we construct is the generating function $\mathbb G$, which encodes the ensemble of generated paths and their probabilities: each term is a probability times a product of (bosonic) fields for the path. 
The fields at time $t$ are: $\tilde \phi(x,t)$ indicating that the tip of the path is at $x$, and 
$\tilde \psi(x,t)$ indicating that the path passed through $x$ at an earlier time. 
Let us give an example ($\lambda_1$ and $\lambda_2$ are two numbers, and the root $\xr$ has two neighbors $x$ and $y$)
\bea\label{eq1}
\mathbb G (t) &=& \tilde \phi(\xr ,t), \\
\label{eq2}
\mathbb G(t{+}1) &=& \left[1- \lambda_1 - \lambda_2 \right] \tilde\phi(\xr ,t{+}1) + \lambda_1 \tilde\psi(\xr ,t{+}1) \tilde\phi(x,t{+}1) + \lambda_2 \tilde\psi(\xr ,t{+}1) \tilde \phi(y,t{+}1).
\eea
\Eq{eq1} represents the root from which we start. \Eq{eq2} represents the state of the system after one iteration, where the 
path has been prolonged from $\xr$ to $x$ and $y$ respectively. 
It is normalized, as can be verified by setting $\tilde \psi\to1$ and $\tilde \phi \to 1$. We can further include information about the edges used, by introducing a factor of $R_{xy}$ for an edge from $x$ to $y$, 
\bea
\mathbb G(t{+}1) &=& \left[1- \lambda_1 - \lambda_2 \right] \tilde\phi(\xr ,t{+}1) \nn\\
&& + \lambda_1 R_{\xr x} \tilde\psi(\xr ,t{+}1) \tilde\phi(x,t{+}1) + \lambda_2 R_{\xr y} \tilde\psi(\xr ,t{+}1) \tilde \phi(y,t{+}1).
\label{eq3}
\eea
The evolution of each term in $\mathbb G(t)$ is given by the following rules
\bea\label{24bis}
\tilde \phi(x,t) &\longrightarrow& \tilde \phi(x,t{+}1) + \gamma \sum_{y\in \ca V} \beta_{xy} \Big[ R_{x y} \tilde\phi(y,t{+}1) \tilde \psi(x,t{+}1) - \tilde\phi(x,t{+}1)\Big]\Phi_{\ca V| \ca E}(y,t),\qquad \\
\label{25}
\tilde \psi(x,t) &\longrightarrow& \tilde\psi(x,t{+}1).
\eea
The first term in \Eq{24bis} propagates a configuration without changing it. The second term is the 
growth from $x$ to $y$ with rate $\gamma \beta _{xy}\Phi_{\ca V| \ca E}(y,t)$ in a time interval $\delta t=1$. 
While the first term in the square brackets is the configuration with an additional vertex added, the second subtracts with the same weight the original configuration. This iteration works even for large $\gamma \delta t$ as long as one avoids numerical instabilities. 

The action which implements \Eqs{24bis}-\eq{25} contains $\tilde \psi$ and $\tilde \phi$, as well as their complex conjugates. It further contains 
\noindent $\tilde \chi (x,t)$ and $\chi(x,t)$ to evaluate $\Phi_{\ca V| \ca E}(y,t)$ on each time slice with an action close to \Eq{eq:S_boson2}. 
We claim that 
\be
\mathbb G(t{+}1) = \int_{\tilde \chi,\chi } \int_{\tilde \phi,\phi }\int_{\tilde \psi,\psi }
\rme^{-\ca L_1^{\textsc{LRW}}} \mathbb G(t)
\ee
with Lagrangian 
\bea \label{action-LRW-dynamic-1}
\rme^{-\ca L_1^{\textsc{LRW}}} &=& \prod_{x \in \ca V} \exp\!\Big( - \tilde \phi (x,t)\phi (x,t)- \tilde \psi (x,t)\psi (x,t) - \sum_{\alpha=1}^{n_{\alpha}}\sum_{j=1}^{n_\chi} \tilde\chi_j^\alpha (x,t)\chi_j^\alpha(x,t)\Big) \nn\\
& \times& \prod_{x \in \ca V\backslash \xt } \bigg\{ \prod_{\alpha=1}^{n_\alpha}\Big[ 1+ \sum_{y \in \ca V}\hat \beta_{xy} \sum_{j=1}^{n_\chi} \tilde \chi_j^\alpha (y,t)\chi_j^\alpha(x,t)\Big] + \tilde \phi (x,t{+}1)\phi(x,t)+ \tilde \psi (x,t{+}1)\psi (x,t) \nn\\
&& \qquad~~~~~ + 
 \gamma \sum_{y\in \ca V} { \beta_{xy}}
 \Big[ R_{xy} \tilde \phi (y,t{+}1)\tilde \psi(x,t{+}1) -\tilde \phi(x,t{+}1)\Big] \phi (x,t) \tilde \chi_2^1 (y,t) \bigg\} \nn\\
 & \times& \Big[ 1+ \tilde \phi(\xt ,t{+}1) \phi(\xt,t) + \chi_{2}^1 (\xt,t) \Big].
\eea
The first line is the Gaussian measure for each field, ensuring that $\langle \tilde \phi(x,t) \phi(x,t) \rangle_0 = 
\langle \tilde \psi(x,t) \psi(x,t)\rangle_0$ $=1$, $\langle \tilde \chi_i^\alpha(x,t) \chi_j^\beta(x,t)\rangle_0 =\delta_{ij}\delta^{\alpha \beta} $, and all other expectation values vanish. The additional replica index $\alpha$ will be explained below; for the moment the reader may think of setting $n_\alpha=1$ and $\alpha=1$ (effectively dropping the index).

The second line is a product over all vertices other than the target $\xt$. 
The first term comes from \Eq{eq:S_boson2}. The second term says that if 
 there is a field $\tilde \phi(x,t)$ it can be propagated to time $t{+}1$ by 
$\tilde \phi (x,t{+}1)\phi(x,t)$; an equivalent term is there to propagate $\tilde\psi(x,t)$. As only one term per site can be picked, the boundary conditions for $\chi$ are implemented, as it cannot touch the $\phi$ or $\psi$ fields. 

Let us postpone the discussion of the third line proportional to $\gamma$; the last line says that if a path arrives at the target $\xt$, it stays there. 
It also contains a field $ \chi_{2}^1 (\xt,t) $ for the evaluation of the Laplace equation. The way it is written ensures that once the target is reached and the task finished, we need to use $\tilde \phi (x,t{+}1)\phi(x,t)$, and as a result stop solving the Laplace equation, and thus growth. 

Let us now check normalizations for $\gamma=0$. If a vertex is occupied by $\tilde \phi$ or $\tilde \psi$, no outgoing line for $\tilde \chi$ exists. Therefore integrating over $\tilde \chi$ and $\chi$ at $n_\chi=-1$ gives the partition function $\ca Z_{\ca V\backslash \ca E}$, where the excluded vertices $\ca E$ are all those vertices for which there is either a $\tilde \phi$ or $\tilde \psi$, or which are the target $\xt$. 
Thus, in order to propagate the configuration properly, one needs to normalize by $\ca Z_{\ca V\backslash \ca E}$, as is done in \Eq{potential-Phi-2}. 
Let us now take $n_\alpha$ arbitrary; given the factorization over $\alpha$, the result is $(\ca Z_{\ca V\backslash \ca E})^{n_\alpha}$. 
In perturbative calculations, we can use the replica trick to set $n_\alpha\to 0$, eliminating the need for a normalization. 
This is not possible when working on a graph and performing the Wick contractions; in order to check the action \eq{action-LRW-dynamic-1}, we need to explicitly divide each term by its partition function $\ca Z_{\ca V\backslash \ca E}$, before propagating it. 

We can finally address the terms proportional to $\gamma$. Together with $ \chi_{2}^1 (\xt,t) $ the term $\tilde \chi_{2}^1 (y,t)$ evaluates the solution $\Phi_{\ca V| \ca E}(y,t)$ to the Laplace equation. If, and only if $\tilde \phi(x,t)$ is present, $\phi (x,t) $ can be Wick-contracted with it, which produces the terms in \Eq{24bis}.

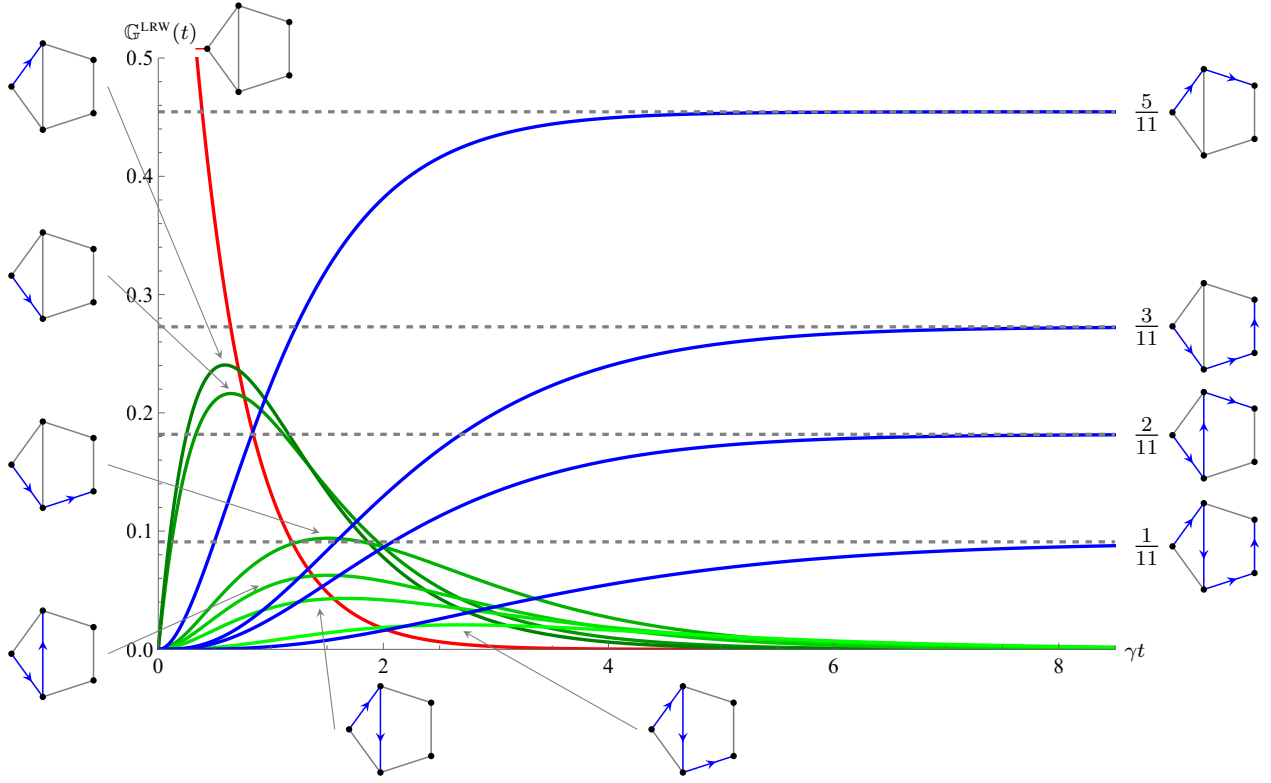
\begin{figure}[t]
\ifTikzExternal{\tikzsetnextfilename{convPlot}}
\[{{\vcenter{\hbox{\begin{tikzpicture}[scale=2]
\coordinate (plot) at (0,0) ;
\coordinate (seed) at (-3,2) ;
\coordinate (1/11) at (3.2,-1.29) ;
\coordinate (2/11) at (3.2,-0.555) ;
\coordinate (3/11) at (3.2,0.165) ;
\coordinate (5/11) at (3.2,1.58) ;
\coordinate (R15) at (-3.5,1.75) ;
\coordinate (R12) at (-3.5,0.5) ;
\coordinate (R123) at (-3.5,-0.75) ;
\coordinate (R125) at (-3.5,-2) ;
\coordinate (R152) at (-2,-2.5) ;
\coordinate (R1523) at (0,-2.5) ;
\node at (plot) {\fig{0.88888888888888}{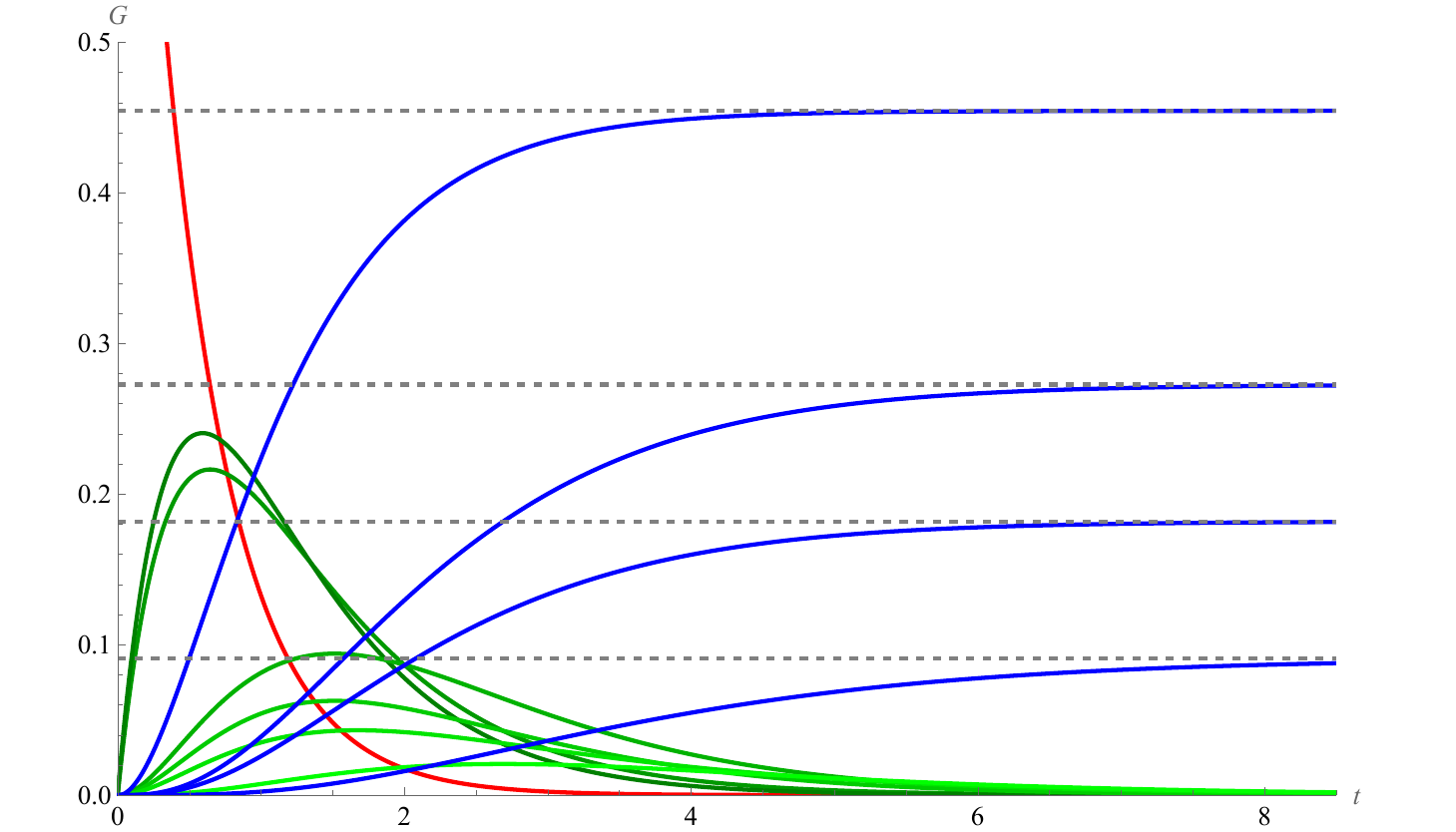}};
\node[font=\scriptsize, fill=white, inner sep=1pt] at ($(seed) + (-0.15,0.1)$) {$\mathbb G^{\textsc{lrw}}(t)$};
\node[anchor=west,font=\scriptsize, fill=white, inner sep=1pt] at ($(1/11) + (0,-0.7)$) {$\gamma t$};
\node[anchor=west] at (seed) {$\vcenter{\hbox{%
\begin{tikzpicture}[x=0.12mm,y=0.12mm]
 \draw [line width=0.2mm, gray] (-50.,0.) -- (-15.4508,-47.5528) ;
 \draw [line width=0.2mm, gray] (-15.4508,-47.5528) -- (40.4508,-29.3893) ;
 \draw [line width=0.2mm, gray] (40.4508,-29.3893) -- (40.4508,29.3893) ;
 \draw [line width=0.2mm, gray] (-15.4508,47.5528) -- (-15.4508,-47.5528) ;
 \draw [line width=0.2mm, gray] (-50.,0.) -- (-15.4508,47.5528) ;
 \draw [line width=0.2mm, gray] (-15.4508,47.5528) -- (40.4508,29.3893) ;
 \draw [line width=0.2mm, red] (-50.,0.) -- (-63.,0.) ;
 \fill(-50.,0.) circle (1.2pt);
 \fill(-15.4508,-47.5528) circle (1.2pt);
 \fill(40.4508,-29.3893) circle (1.2pt);
 \fill(40.4508,29.3893) circle (1.2pt);
 \fill(-15.4508,47.5528) circle (1.2pt);
\end{tikzpicture}%
 }}$};
\node[anchor=west] at (1/11) {$\frac{1}{11}~\vcenter{\hbox{%
 \begin{tikzpicture}[x=0.12mm,y=0.12mm]
 \draw [line width=0.2mm, gray] (-50.,0.) -- (-15.4508,-47.5528) ;
 \draw [line width=0.2mm, blue,directed] (-15.4508,-47.5528) -- (40.4508,-29.3893) ;
 \draw [line width=0.2mm, blue,directed] (40.4508,-29.3893) -- (40.4508,29.3893) ;
 \draw [line width=0.2mm, blue,directed] (-15.4508,47.5528) -- (-15.4508,-47.5528) ;
 \draw [line width=0.2mm, blue,directed] (-50.,0.) -- (-15.4508,47.5528) ;
 \draw [line width=0.2mm, gray] (-15.4508,47.5528) -- (40.4508,29.3893) ;
 \fill(-50.,0.) circle (1.2pt);
 \fill(-15.4508,-47.5528) circle (1.2pt);
 \fill(40.4508,-29.3893) circle (1.2pt);
 \fill(40.4508,29.3893) circle (1.2pt);
 \fill(-15.4508,47.5528) circle (1.2pt);
 \end{tikzpicture}%
 }}$};
\node[anchor=west] at (2/11) {$\frac{2}{11}~\vcenter{\hbox{%
 \begin{tikzpicture}[x=0.12mm,y=0.12mm]
 \draw [line width=0.2mm, blue,directed] (-50.,0.) -- (-15.4508,-47.5528) ;
 \draw [line width=0.2mm, gray] (-15.4508,-47.5528) -- (40.4508,-29.3893) ;
 \draw [line width=0.2mm, gray] (40.4508,-29.3893) -- (40.4508,29.3893) ;
 \draw [line width=0.2mm, blue,directed] (-15.4508,-47.5528) -- (-15.4508,47.5528) ;
 \draw [line width=0.2mm, gray] (-50.,0.) -- (-15.4508,47.5528) ;
 \draw [line width=0.2mm, blue,directed] (-15.4508,47.5528) -- (40.4508,29.3893) ;
 \fill(-50.,0.) circle (1.2pt);
 \fill(-15.4508,-47.5528) circle (1.2pt);
 \fill(40.4508,-29.3893) circle (1.2pt);
 \fill(40.4508,29.3893) circle (1.2pt);
 \fill(-15.4508,47.5528) circle (1.2pt);
 \end{tikzpicture}%
 }}$};
\node[anchor=west] at (3/11) {$\frac{3}{11}~\vcenter{\hbox{%
 \begin{tikzpicture}[x=0.12mm,y=0.12mm]
 \draw [line width=0.2mm, blue,directed] (-50.,0.) -- (-15.4508,-47.5528) ;
 \draw [line width=0.2mm, blue,directed] (-15.4508,-47.5528) -- (40.4508,-29.3893) ;
 \draw [line width=0.2mm, blue,directed] (40.4508,-29.3893) -- (40.4508,29.3893) ;
 \draw [line width=0.2mm, gray] (-15.4508,47.5528) -- (-15.4508,-47.5528) ;
 \draw [line width=0.2mm, gray] (-50.,0.) -- (-15.4508,47.5528) ;
 \draw [line width=0.2mm, gray] (-15.4508,47.5528) -- (40.4508,29.3893) ;
 \fill(-50.,0.) circle (1.2pt);
 \fill(-15.4508,-47.5528) circle (1.2pt);
 \fill(40.4508,-29.3893) circle (1.2pt);
 \fill(40.4508,29.3893) circle (1.2pt);
 \fill(-15.4508,47.5528) circle (1.2pt);
 \end{tikzpicture}%
 }}$};
\node[anchor=west] at (5/11) {$\frac{5}{11}~\vcenter{\hbox{%
\begin{tikzpicture}[x=0.12mm,y=0.12mm]
 \draw [line width=0.2mm, gray] (-50.,0.) -- (-15.4508,-47.5528) ;
 \draw [line width=0.2mm, gray] (-15.4508,-47.5528) -- (40.4508,-29.3893) ;
 \draw [line width=0.2mm, gray] (40.4508,-29.3893) -- (40.4508,29.3893) ;
 \draw [line width=0.2mm, gray] (-15.4508,47.5528) -- (-15.4508,-47.5528) ;
 \draw [line width=0.2mm, blue,directed] (-50.,0.) -- (-15.4508,47.5528) ;
 \draw [line width=0.2mm, blue,directed] (-15.4508,47.5528) -- (40.4508,29.3893) ;
 \fill(-50.,0.) circle (1.2pt);
 \fill(-15.4508,-47.5528) circle (1.2pt);
 \fill(40.4508,-29.3893) circle (1.2pt);
 \fill(40.4508,29.3893) circle (1.2pt);
 \fill(-15.4508,47.5528) circle (1.2pt);
\end{tikzpicture}%
}}$};
\node[anchor=east] at (R15) {$\vcenter{\hbox{%
\begin{tikzpicture}[x=0.12mm,y=0.12mm]
 \draw [line width=0.2mm, gray] (-50.,0.) -- (-15.4508,-47.5528) ;
 \draw [line width=0.2mm, gray] (-15.4508,-47.5528) -- (40.4508,-29.3893) ;
 \draw [line width=0.2mm, gray] (40.4508,-29.3893) -- (40.4508,29.3893) ;
 \draw [line width=0.2mm, gray] (-15.4508,47.5528) -- (-15.4508,-47.5528) ;
 \draw [line width=0.2mm, blue,directed] (-50.,0.) -- (-15.4508,47.5528) ;
 \draw [line width=0.2mm, gray] (-15.4508,47.5528) -- (40.4508,29.3893) ;
 \fill(-50.,0.) circle (1.2pt);
 \fill(-15.4508,-47.5528) circle (1.2pt);
 \fill(40.4508,-29.3893) circle (1.2pt);
 \fill(40.4508,29.3893) circle (1.2pt);
 \fill(-15.4508,47.5528) circle (1.2pt);
\end{tikzpicture}%
 }}$};
\draw[gray,enddirected] (R15) -- (-2.75,-0.05);
\node[anchor=east] at (R12) {$\vcenter{\hbox{%
\begin{tikzpicture}[x=0.12mm,y=0.12mm]
 \draw [line width=0.2mm, blue,directed] (-50.,0.) -- (-15.4508,-47.5528) ;
 \draw [line width=0.2mm, gray] (-15.4508,-47.5528) -- (40.4508,-29.3893) ;
 \draw [line width=0.2mm, gray] (40.4508,-29.3893) -- (40.4508,29.3893) ;
 \draw [line width=0.2mm, gray] (-15.4508,47.5528) -- (-15.4508,-47.5528) ;
 \draw [line width=0.2mm, gray] (-50.,0.) -- (-15.4508,47.5528) ;
 \draw [line width=0.2mm, gray] (-15.4508,47.5528) -- (40.4508,29.3893) ;
 \fill(-50.,0.) circle (1.2pt);
 \fill(-15.4508,-47.5528) circle (1.2pt);
 \fill(40.4508,-29.3893) circle (1.2pt);
 \fill(40.4508,29.3893) circle (1.2pt);
 \fill(-15.4508,47.5528) circle (1.2pt);
\end{tikzpicture}%
 }}$};
\draw[gray,enddirected] (R12) -- (-2.7,-0.25);
\node[anchor=east] at (R123) {$\vcenter{\hbox{%
\begin{tikzpicture}[x=0.12mm,y=0.12mm]
 \draw [line width=0.2mm, blue,directed] (-50.,0.) -- (-15.4508,-47.5528) ;
 \draw [line width=0.2mm, blue,directed] (-15.4508,-47.5528) -- (40.4508,-29.3893) ;
 \draw [line width=0.2mm, gray] (40.4508,-29.3893) -- (40.4508,29.3893) ;
 \draw [line width=0.2mm, gray] (-15.4508,-47.5528) -- (-15.4508,47.5528) ;
 \draw [line width=0.2mm, gray] (-50.,0.) -- (-15.4508,47.5528) ;
 \draw [line width=0.2mm, gray] (-15.4508,47.5528) -- (40.4508,29.3893) ;
 \fill(-50.,0.) circle (1.2pt);
 \fill(-15.4508,-47.5528) circle (1.2pt);
 \fill(40.4508,-29.3893) circle (1.2pt);
 \fill(40.4508,29.3893) circle (1.2pt);
 \fill(-15.4508,47.5528) circle (1.2pt);
\end{tikzpicture}%
 }}$};
\draw[gray,enddirected] (R123) -- (-2.1,-1.2);
\node[anchor=east] at (R125) {$\vcenter{\hbox{%
\begin{tikzpicture}[x=0.12mm,y=0.12mm]
 \draw [line width=0.2mm, blue,directed] (-50.,0.) -- (-15.4508,-47.5528) ;
 \draw [line width=0.2mm, gray] (-15.4508,-47.5528) -- (40.4508,-29.3893) ;
 \draw [line width=0.2mm, gray] (40.4508,-29.3893) -- (40.4508,29.3893) ;
 \draw [line width=0.2mm, blue,directed] (-15.4508,-47.5528) -- (-15.4508,47.5528) ;
 \draw [line width=0.2mm, gray] (-50.,0.) -- (-15.4508,47.5528) ;
 \draw [line width=0.2mm, gray] (-15.4508,47.5528) -- (40.4508,29.3893) ;
 \fill(-50.,0.) circle (1.2pt);
 \fill(-15.4508,-47.5528) circle (1.2pt);
 \fill(40.4508,-29.3893) circle (1.2pt);
 \fill(40.4508,29.3893) circle (1.2pt);
 \fill(-15.4508,47.5528) circle (1.2pt);
\end{tikzpicture}%
 }}$};
\draw[gray,enddirected] (R125) -- (-2.5,-1.55);
\node[anchor=west] at (R152) {$\vcenter{\hbox{%
\begin{tikzpicture}[x=0.12mm,y=0.12mm]
 \draw [line width=0.2mm, gray] (-50.,0.) -- (-15.4508,-47.5528) ;
 \draw [line width=0.2mm, gray] (-15.4508,-47.5528) -- (40.4508,-29.3893) ;
 \draw [line width=0.2mm, gray] (40.4508,-29.3893) -- (40.4508,29.3893) ;
 \draw [line width=0.2mm, blue,directed] (-15.4508,47.5528) -- (-15.4508,-47.5528) ;
 \draw [line width=0.2mm, blue,directed] (-50.,0.) -- (-15.4508,47.5528) ;
 \draw [line width=0.2mm, gray] (-15.4508,47.5528) -- (40.4508,29.3893) ;
 \fill(-50.,0.) circle (1.2pt);
 \fill(-15.4508,-47.5528) circle (1.2pt);
 \fill(40.4508,-29.3893) circle (1.2pt);
 \fill(40.4508,29.3893) circle (1.2pt);
 \fill(-15.4508,47.5528) circle (1.2pt);
\end{tikzpicture}%
 }}$};
\draw[gray,enddirected] (R152) -- (-2.1,-1.7);
\node[anchor=west] at (R1523) {$\vcenter{\hbox{%
\begin{tikzpicture}[x=0.12mm,y=0.12mm]
 \draw [line width=0.2mm, gray] (-50.,0.) -- (-15.4508,-47.5528) ;
 \draw [line width=0.2mm, blue,directed] (-15.4508,-47.5528) -- (40.4508,-29.3893) ;
 \draw [line width=0.2mm, gray] (40.4508,-29.3893) -- (40.4508,29.3893) ;
 \draw [line width=0.2mm, blue,directed] (-15.4508,47.5528) -- (-15.4508,-47.5528) ;
 \draw [line width=0.2mm, blue,directed] (-50.,0.) -- (-15.4508,47.5528) ;
 \draw [line width=0.2mm, gray] (-15.4508,47.5528) -- (40.4508,29.3893) ;
 \fill(-50.,0.) circle (1.2pt);
 \fill(-15.4508,-47.5528) circle (1.2pt);
 \fill(40.4508,-29.3893) circle (1.2pt);
 \fill(40.4508,29.3893) circle (1.2pt);
 \fill(-15.4508,47.5528) circle (1.2pt);
\end{tikzpicture}%
 }}$};
\draw[gray,enddirected] (R1523) -- (-1.15,-1.85);
 \end{tikzpicture}%
 }}}%
 }
\]
\caption{The contributions to $\mathbb G^{\textsc{LRW}}(t)$ of each path are plotted against $\gamma t$ to show how $\mathbb G^{\textsc{LRW}}(t=0) = \tilde\phi(1,0)$ evolves under the action \eq{action-LRW-dynamic-1}, and converges to $\mathbb G^{\textsc{LERW}} $ given in \Eq{G-grpah1}. It is obtained numerically, using $\gamma \delta t=0.01$; each curve is associated to the respective path. Starting from an empty graph with a seed (red curve), the graph is explored by intermediate paths (green curves) which finally reach the target, converging to the non-self-intersecting paths with LRW statistics (blue cruves). The same trajectories are obtained using $\ca L_2^{\textsc{lrw}}$.}
\label{f:evolution}
\end{figure}

\subsection{Test of lattice field theory} 
\label{Test of lattice field theory} 
We tested the lattice field theory on some simple graphs. Here is an example for graph $\ca G_1$ with five vertices; its root is at vertex $\xr=1$ (indicated with a red line), and the target is at vertex $\xt=4$ (indicated with a green line).
\be\label{graph-G1}
\ca G_1 = {\parbox{2.8cm}{\ifTikzExternal{\tikzsetnextfilename{tikzAux1}}\begin{tikzpicture}[x=0.2mm,y=0.2mm]
\draw [line width=0.15mm, gray] (-50.,0.) -- (-15.4508,-47.5528) ;
\draw [line width=0.15mm, gray] (-50.,0.) -- (-15.4508,47.5528) ;
\draw [line width=0.15mm, gray] (-15.4508,-47.5528) -- (40.4508,-29.3893) ;
\draw [line width=0.15mm, gray] (40.4508,-29.3893) -- (40.4508,29.3893) ;
\draw [line width=0.15mm, gray] (40.4508,29.3893) -- (-15.4508,47.5528) ;
\draw [line width=0.15mm, gray] (-15.4508,47.5528) -- (-15.4508,-47.5528) ;
\draw [line width=0.6mm, red] (-50.,0.) -- (-60.,0.) ;
\draw [line width=0.6mm, Green] (40.4508,29.3893) -- (48.541, 35.2671); 
\fill(-50.,0.) circle (1.5pt);
\node [font=\scriptsize] at (-62.5,0.) {1};
\fill(-15.4508,-47.5528) circle (1.5pt);
\node [font=\scriptsize] at(-19.3136,-59.441) {2};
\fill(40.4508,-29.3893) circle (1.5pt);
\node [font=\scriptsize] at(50.5636,-36.7366) {3};
\fill(40.4508,29.3893) circle (1.5pt);
\node [font=\scriptsize] at(51,40) {4};
\fill(-15.4508,47.5528) circle (1.5pt);
\node [font=\scriptsize] at(-19.3136,59.441) {5};
\end{tikzpicture}}} ~~.
\ee 
The generating function $\mathbb G$ of LERWs on $\ca G_1$ via the $O(-2)$ theory defined via action \eq{eq:S_boson2} reads
\be
\mathbb G ^{\textsc{LERW}}= \frac{3}{11}\parbox{20mm}{\ifTikzExternal{\tikzsetnextfilename{tikzAux2}}\begin{tikzpicture}[x=0.2mm,y=0.2mm]
\draw [line width=0.2mm, gray] (-50.,0.) -- (-15.4508,47.5528) ;
\draw [line width=0.2mm, gray] (40.4508,29.3893) -- (-15.4508,47.5528) ;
\draw [line width=0.2mm, gray] (-15.4508,47.5528) -- (-15.4508,-47.5528) ;
\draw [line width=0.3mm, blue,directed] (-50.,0.) -- (-15.4508,-47.5528) ;
\draw [line width=0.3mm, blue,directed] (-15.4508,-47.5528) -- (40.4508,-29.3893) ;
\draw [line width=0.3mm, blue,directed] (40.4508,-29.3893) -- (40.4508,29.3893) ;
\fill(-50.,0.) circle (1.5pt);
\fill(-15.4508,-47.5528) circle (1.5pt);
\fill(40.4508,-29.3893) circle (1.5pt);
\fill(40.4508,29.3893) circle (1.5pt);
\fill(-15.4508,47.5528) circle (1.5pt); 
\end{tikzpicture}
}
 + \frac{5}{11}\parbox{20mm}{\ifTikzExternal{\tikzsetnextfilename{tikzAux3}}{\begin{tikzpicture}[x=0.2mm,y=0.2mm]
\draw [line width=0.2mm, gray] (-50.,0.) -- (-15.4508,-47.5528) ;
\draw [line width=0.2mm, gray] (-15.4508,-47.5528) -- (40.4508,-29.3893) ;
\draw [line width=0.2mm, gray] (40.4508,-29.3893) -- (40.4508,29.3893) ;
\draw [line width=0.2mm, gray] (-15.4508,47.5528) -- (-15.4508,-47.5528) ;
\draw [line width=0.3mm, blue,directed] (-50.,0.) -- (-15.4508,47.5528) ;
\draw [line width=0.3mm, blue,directed] (-15.4508,47.5528) -- (40.4508,29.3893) ;
\fill(-50.,0.) circle (1.5pt);
\fill(-15.4508,-47.5528) circle (1.5pt);
\fill(40.4508,-29.3893) circle (1.5pt);
\fill(40.4508,29.3893) circle (1.5pt);
\fill(-15.4508,47.5528) circle (1.5pt);
\end{tikzpicture}}
}
 + \frac{2}{11}\parbox{20mm}{\ifTikzExternal{\tikzsetnextfilename{tikzAux4}}\begin{tikzpicture}[x=0.2mm,y=0.2mm]
\draw [line width=0.2mm, gray] (-50.,0.) -- (-15.4508,47.5528) ;
\draw [line width=0.2mm, gray] (-15.4508,-47.5528) -- (40.4508,-29.3893) ;
\draw [line width=0.2mm, gray] (40.4508,-29.3893) -- (40.4508,29.3893) ;
\draw [line width=0.3mm, blue,directed] (-50.,0.) -- (-15.4508,-47.5528) ;
\draw [line width=0.3mm, blue,directed] (-15.4508,-47.5528) -- (-15.4508,47.5528) ;
\draw [line width=0.3mm, blue,directed] (-15.4508,47.5528) -- (40.4508,29.3893) ;
\fill(-50.,0.) circle (1.5pt);
\fill(-15.4508,-47.5528) circle (1.5pt);
\fill(40.4508,-29.3893) circle (1.5pt);
\fill(40.4508,29.3893) circle (1.5pt);
\fill(-15.4508,47.5528) circle (1.5pt);
\end{tikzpicture}
}
 + \frac{1}{11}\parbox{20mm}{\ifTikzExternal{\tikzsetnextfilename{tikzAux5}}\begin{tikzpicture}[x=0.2mm,y=0.2mm]
\draw [line width=0.2mm, gray] (-50.,0.) -- (-15.4508,-47.5528) ;
\draw [line width=0.2mm, gray] (40.4508,29.3893) -- (-15.4508,47.5528) ;
\draw [line width=0.3mm, blue,directed] (-50.,0.) -- (-15.4508,47.5528) ;
\draw [line width=0.3mm, blue,directed] (-15.4508,-47.5528) -- (40.4508,-29.3893) ;
\draw [line width=0.3mm, blue,directed] (40.4508,-29.3893) -- (40.4508,29.3893) ;
\draw [line width=0.3mm, blue,directed] (-15.4508,47.5528) -- (-15.4508,-47.5528) ;
\fill(-50.,0.) circle (1.5pt);
\fill(-15.4508,-47.5528) circle (1.5pt);
\fill(40.4508,-29.3893) circle (1.5pt);
\fill(40.4508,29.3893) circle (1.5pt);
\fill(-15.4508,47.5528) circle (1.5pt);
\end{tikzpicture}
}.
\label{G-grpah1}
\ee
We numerically checked this against an LERW starting at $\xr$ and stopping at $\xt$.

To obtain the evolution via the lattice action, we evolve $\mathbb G(t=0) = \tilde\phi(1,0)$ with the action \eq{action-LRW-dynamic-1}. This is shown in Fig.~\ref{f:evolution}. One sees that configurations with one edge appear and are progressively replaced by configurations with more and more edges until only configurations that end at the target are left, and the process stops. 
We checked numerically to machine precision that 
\be
\lim_{t\to \infty} \mathbb G^{\textsc{LRW}}(t) = \mathbb G^{\textsc{LERW}}. 
\ee

\subsection{Continuum limit for $\ca L^{\textsc{LRW}}_1$}
\label{Continuum limit for L-LERW-1}
We now expand $\ca L_1^{\textsc{LRW}}$ as defined in \Eq{action-LRW-dynamic-1} in the fields, retaining up to four fields. 
Power counting done later will show that this is sufficient. 
\bea\label{action-LRW-dynamic-1-expanded}
\ca L_1^{\textsc{LRW}} &\simeq& \sum_{x\in \ca V} \big[ \tilde \phi(x,t)-\tilde \phi(x,t{+}1)\big] \phi(x,t) 
 + \big[ \tilde \psi(x,t)-\tilde \psi(x,t{+}1)\big] \psi(x,t) 
\nn\\
&+& \sum_{x,y\in \ca V} \sum_{j,\alpha}\hat \beta_{xy} \big[\tilde \chi_j^\alpha (x,t)-\tilde \chi_j^\alpha (y,t)\big] \chi_j^\alpha(x,t) \nn\\
&+& \sum_{x\in \ca V}g \bigg\{ \frac{1}2 \big[ \tilde \phi(x,t{+}1) \phi(x,t) +\tilde \psi(x,t{+}1)\psi(x,t)\big]^2 + \frac{1}2 \sum_{\alpha} \Big[\sum_j \sum_{ y\in\ca V} \hat \beta_{xy}\tilde \chi_j^\alpha (y,t) \chi_j^\alpha(x,t) \Big]^2 
\nn\\
&&\qquad~~~+ \sum_{j,\alpha} \big[ \tilde \phi(x,t{+}1) \phi(x,t) +\tilde \psi(x,t{+}1)\psi(x,t)\big] \sum_{ y\in\ca V} \hat \beta_{xy}\tilde \chi_j^\alpha (y,t) \chi_j^\alpha(x,t)\bigg\} \nn\\
&-& \gamma 
 \sum_{x,y\in \ca V} {   \beta_{xy}} \Big[ \tilde \phi (y,t{+}1) \tilde \psi(x,t{+}1) - \tilde \phi(x,t{+}1) \Big] \tilde \chi_2^1 (y,t) \phi (x,t) \nn\\
 &-& \chi_{2}^1 (\xt,t) .
\eea
The factor of $g$, which is $g=1$ in the microscopic theory, was introduced to facilitate bookkeeping, while 
$R_{xy}\to 1$ since it is no longer useful. 
The last line provides the target for the solution of the Laplace equation. 

In the next step, we interpret fields at time $t+1 $ as $t+\delta t$, and expand them around $t$. 
Similarly, fields at $y$ are Taylor expanded around $x$, where we used the Laplacian defined in \Eq{lat-Laplace}.

Once one reaches the target, all auxiliary fields $\tilde \psi$ are set to $1$. (Formally, this is done for final time $T\to \infty$.) Technically, this gives a vacuum expectation value to $\tilde \psi$. Similarly, the solution of the Laplace equation is obtained as $1$ minus corrections when paths intersect. This leads to the following definitions
\bea\label{delta-psi-def}
\tilde \psi(x,t) &=& 1 + \delta \tilde \psi(x,t),\\
\label{delta-chi-def}
\tilde \chi(x,t) &=& 1 + \delta \tilde \chi(x,t).
\eea
With this, the continuum limit of \Eq{action-LRW-dynamic-1-expanded} reads 
\bea\label{action-LRW-dynamic-1-continuum}
 \!\!\! {\ca L_1^{\textsc{LRW}}} &\simeq& \ds {\int_x} \tilde \phi(x,t) [\partial_t-\gamma\nabla^2] \phi(x,t)+ \tilde \psi(x,t) \partial_t \psi(x,t) - 
 \sum_{j,\alpha} \tilde \chi_j^\alpha (x,t)\nabla^2 \chi_j^\alpha(x,t) \nn\\
&&~~+g \bigg\{ \frac{1}2 \big[ \tilde \phi(x,t) \phi(x,t) +\tilde \psi(x,t)\psi(x,t)\big]^2 
+ \frac{1}2 \sum_{\alpha} \Big[\sum_j \tilde \chi_j^\alpha (x,t) \chi_j^\alpha(x,t) \Big]^2 
 \nn\\
&&\qquad ~~ + \sum_{j,\alpha} \big[ \tilde \phi(x,t) \phi(x,t) +\tilde \psi(x,t)\psi(x,t)\big] \tilde \chi_j^\alpha (x,t) \chi_j^\alpha(x,t) \bigg\} \nn\\
&& ~~- \gamma \Big\{  \tilde \phi (x,t) \tilde \chi_2^1 (x,t) \delta \tilde \psi(x,t) + \nabla^2\Big[\tilde \phi (x,t) \delta \tilde \chi_2^1 (x,t) \Big] \nn\\
&& \qquad ~ +  \delta\tilde \psi(x,t)\nabla^2 \Big[\tilde \phi (x,t) \tilde \chi_2^1 (x,t) \Big] {-} \tilde \phi(x,t) \nabla^2 \delta \tilde \chi_2^1(x,t)\Big\} \phi (x,t) \nn\\
 && ~~- \chi_{2}^1 (\xt,t) .
\eea
To arrive at this equation, we used \Eqs{delta-psi-def} and \eq{delta-chi-def}; the leading term $\tilde \psi=\tilde \chi=1$ gives 
the diffusion term $-\gamma \tilde \phi \nabla^2 \phi$ in the first line above. 
The term of order $\delta \tilde \psi(x,t)$ in the last curly brackets of \Eq{action-LRW-dynamic-1-continuum} is the leading term in the Taylor expansion of \Eq{action-LRW-dynamic-1-expanded}. What remains is proportional to $\delta \tilde \chi(x,t)$.

\subsection{Propagators and vertices}
\label{Propagators and vertices}
The free-theory propagators are
\begin{align}
\label{propRed}
R^{\phi}(x,t,y,t') =\propRed &= \left< \tilde \phi(x,t) \phi(y,t')\right>_0 = \int_k \rme^{ik(y-x)} \rme^{- k^2 (t'-t) \gamma}\Theta(t'>t), \\
\label{propBlue}
R^{\psi}(x,t,y,t') =\hspace{0.2cm}\propBlue \hspace{0.2cm}&=\left< \tilde \psi(x,t) \psi(y,t')\right>_0 = \delta_{xy} \Theta(t'>t), \\
\label{propGreen}
R^{\chi}(x,t,y,t') =\propGreen &=\left< \tilde \chi_i^\alpha (x,t) \chi_j^\beta(y,t') \right>_0 = \delta_{ij} \delta^{\alpha \beta} \delta_{tt'}\int_k \rme^{ik(x-y)} \frac1{k^2}.
\end{align}
Our convention is to draw time vertically from the bottom to the top, and space horizontally; the above symbols reflect this. 
All $\delta $-functions in space and time are Kronecker-$\delta$'s: the $\chi$-theory lives on a {\em single time-slice} $t$, while 
the $\psi$-theory propagates the occupation on a {\em single site}. 
This means one should think of the theory as discretized, with 
\be
 \delta_{tt'}^2 = \delta_{tt'} ,\qquad\delta_{xy}^2 = \delta_{xy}.
 \ee
To get more compact expressions, we may write $ R^{\phi}(x,t,y,t')=: R_{y-x,t'-t}$, 
$R^{\psi}(x,t,y,t') = \Theta(t'>t)\delta_{xy}$, $R^{\chi}(x,t,y,t') =: C_{x-y} \delta_{tt'} $.
 
The interaction terms are (see also appendix \ref{a:notations})
\bea\label{relevant-g-terms}
 \ca L_1^{g} &=& \frac{g}2 \left[\, \vertexIntPhi + \vertexIntPsi + \vertexIntChi \, \right]
 + g \left[ \,\vertexIntPhiPsi + \vertexIntPhiChi + \vertexIntPsiChi \,\right],
 \\
 \label{relevant-gamma-terms}
-\ca L_1^\gamma &=& \gamma_1 \vertexOneDashed +\gamma_2^+ \vertexTwoRel + \gamma_2^{+\psi} \vertexTwoBlueDashed - \gamma_2^-\,\vertexTwoMinus, \qquad \gamma_1 = \gamma_2^+= \gamma_2^{+\psi}=\gamma_2^- = \gamma.
\eea
We use dashed lines to denote the fluctuating parts $\delta \tilde \chi$ and $\delta \tilde\psi$, and dots to separate parts of the vertex. 
Even though all $\gamma$-terms are equal, we added indices to be able to refer to them as $\gamma_1$, $\gamma_2^+$, $\gamma_2^{+\psi}$ and $\gamma_2^-$ terms or vertices. 
The extra $\gamma_2^{+\psi}$ vertex will be dropped thanks to power counting. It generates tadpoles that cancel, see appendix \ref{a:tadpoles}.
To see these tadpole cancelations, we need the $g\gamma$ vertices coming from the expansion of the logarithm: at second order in the expansion, there is a cross term between $\gamma_1$ and $\psi \tilde\psi+ \sum_{\alpha,j}\chi_j^\alpha \tilde\chi_j^\alpha$.
These vertices are
\begin{equation}\label{eq:gGammaVertices}
	\ca L_1^{g\gamma} = g\gamma^{\psi} \vertexGGammaPsi + g\gamma^{\chi} \vertexGGammaChi.
\end{equation}

\subsection{Power counting}
\label{Power counting}
Power counting is confusing. There are two sources for this: 
\begin{enumerate}
\itemsep0em 
\item[({\em i})] the theory contains both a static part, living on a single time slice (the theory of $\chi$ and $\tilde \chi$), and a dynamic part ($\tilde \phi$, $\phi$ and $\tilde \psi$, $\psi$). While $\phi$ diffuses, $\psi$ does not evolve in space; 
\item[({\em ii})] both $\tilde \chi= 1+\delta \tilde \chi$ and $\tilde \psi=1+\delta \tilde \psi$ contain a background field $1$, and a fluctuating part, which scale differently.
\end{enumerate}
Let us start from \Eq{action-LRW-dynamic-1-continuum}.
First, the $\chi$-theory is instantaneous, i.e.\ does not contain time. We assign (in momentum units)
\be\label{dim-chibarchi}
0= \bigg[\int_x \sum_{j,\alpha} \tilde \chi_j^\alpha  \nabla^2 \chi_j^\alpha\bigg] \quad \Longrightarrow \quad 
\left[\tilde \chi \chi\right] := \left[\tilde \chi \chi\right]_\mu = d-2 .
\ee
The coupling $g$ behaves as a standard $\phi^4$ coupling, 
\be\label{dim-g}
0 = \left[ g \int_x \left(\tilde \chi \chi\right) ^2 \right] \quad \Longrightarrow \quad [g]= 4-d =:\epsilon.
\ee
There is diffusion of $\phi$, so 
\be
0 = \bigg[ \int_{x,t}\tilde \phi \partial_t \phi \bigg] =\bigg[ \gamma \int_{x,t} \tilde \phi \nabla^2 \phi\bigg] \quad \Longrightarrow \quad[ t] = -2, ~[x]=-1, ~ [\gamma]=0, ~ \big[ \tilde \phi \phi\big] = d. 
\ee
The action \eq{action-LRW-dynamic-1-continuum} propagates $\psi$, 
\be\label{222}
0 = \bigg[ \int_{x,t}\tilde \psi \partial_t \psi \bigg] \quad\Longrightarrow \quad\big[ \delta \tilde \psi \psi\big] = d . 
\ee
Above, we found the dimensions of $\tilde \phi \phi$, $\tilde \psi \psi$ and $\tilde \chi \chi$. How this is distributed over the single fields is at our disposal. 
For consistency of scaling for $\tilde \chi$, it is easiest to consider the limit of $\tilde \psi \to 1$, 
\be
\bigg[\gamma\int_{x,t}\left\{ \tilde \psi \nabla^2 (\tilde \phi  \tilde \chi ) - \tilde \phi \nabla^2 \tilde \chi\right\} \phi  \bigg] = 0 \quad\Longrightarrow \quad\big[ \tilde \chi \big] = \big[ \delta \tilde \chi \big] = 0. 
\ee
Consistency between the $\gamma$-terms with and without derivatives implies 
\be\label{224}
 \delta \tilde \psi(x,t) \sim \nabla^2 \quad\Longrightarrow \quad\big[ \delta \tilde \psi \big] = 2.
\ee
\Eqs{222} and \eq{224} imply that 
\be
 [\psi] = d-2. 
\ee
Thus both $\tilde \psi$ and $\psi$ should not appear as additional factors for couplings deemed marginal in $d=4$.
The same argument applies to $\tilde \phi$, $\phi$ and $\tilde \chi$, $\chi$. 
Considering \Eq{224}, we expect that when we wish to keep only relevant diagrams, we can remove
$\delta\tilde \psi$ from the interaction, leading to a simplification for the green-blue interaction,
\be
 \vertexIntPsiChi \to \vertexIntPsiChiRel . 
\ee
The $\gamma$ vertex simplifies even more, 
\bea
\label{relevant-gamma-vertices}
-\ca L_1^\gamma &=& 
\gamma_1 \vertexOneDashed +\gamma_2^+ \vertexTwoRel + \gamma_2^{+\psi} \vertexTwoBlueDashed - \gamma_2^-\,\vertexTwoMinus \nn\\
&\to& \gamma_1\, \vertexOneDashed +\gamma_2^+ \,\vertexTwoRel - \gamma_2^-\,\vertexTwoMinus \nn \\
&=& \gamma_1\, \vertexOneDashed + \gamma_{\textsc l} \vertexTwoBred + 2\gamma_{\textsc d} \vertexTwoCred.
\eea
The first line shows all vertices; in the second line, those irrelevant by power counting are suppressed. The third line is 
an exact rewriting, with a natural interpretation, see section \ref{Diffusion and drift}. 

Difficulties in the perturbative calculations with the action \eq{action-LRW-dynamic-1-continuum} are shown in appendix \ref{a:Perturbative calculations}. What is problematic are the interactions between $\phi$ and $\psi$, as well as $\psi$ with itself. 
In the evolution by the lattice action these are irrelevant, as a single site cannot be occupied by more than one ``particle''. 
In the perturbative expansion in appendix \ref{a:LRW1-perturbative} and \ref{a:Tracer-tracer interactions} we see chains of these interactions suppressing the latter. 
Another problem is the plethora of tadpoles one has to handle for the current action, even though they cancel in the end. This is discussed in appendix \ref{a:tadpoles}.

For these reasons, in the next section we construct a lattice action that in the continuum limit 
gives no tadpole diagrams, and retains only vertices relevant for renormalization.

\subsection{A simpler lattice action $\ca L_2^{\textsc{LRW}}$}
\label{A simpler lattice action}
While the lattice action \eq{action-LRW-dynamic-1} does the job, it induces many interactions in its perturbative treatment and many tadpoles, which cancel in the end but complicate the calculations. We therefore wish to construct an action where these terms are absent by construction. The key is to realize that one can add many terms which do not change the outcome: an example is a term which annihilates two $\tilde \phi$-particles at the same lattice site, i.e.\ $\tilde \phi(x,t)^2$. As this never appears, whether the lattice action contains it or not is irrelevant. We can use this property to simplify the lattice action with the goal of reducing or simplifying as much as possible the interaction vertices. 
The simplest working action we found is 
\bea \label{action-LRW-dynamic-2}
&&\!\!\rme^{-\ca L_2^{\textsc{LRW}}} 
= \prod_{x \in \ca V\backslash \xt } \Bigg\{ \exp\!\Big( \big[ \tilde \phi (x,t{+}1) {-} \tilde \phi (x,t)\big]\phi (x,t) 
 {-} \tilde \psi (x,t) \psi (x,t) { -} \sum_{\alpha=1}^{n_{\alpha}}\sum_j \tilde\chi_j^\alpha (x,t)\chi_j^\alpha(x,t)\Big)
 \nn\\
& & \times
\bigg[ \prod_{\alpha=1}^{n_\alpha} \Big( 1+ \sum_{y \in \ca V}\hat \beta_{xy} \sum_{j=1}^{n_\chi} \tilde \chi_j^\alpha (y,t)\chi_j^\alpha(x,t)\Big)-1 \nn\\
&& \qquad~~~ + \exp\!\Big( \tilde \psi (x,t{+}1) \psi (x,t) {+} \gamma\sum_{y\in \ca V} { \beta_{xy}}
 \big[ R_{xy} \tilde \phi (y,t{+}1)\tilde \psi(x,t{+}1) {-}\tilde \phi(x,t{+}1)\big] \phi (x,t) \tilde \chi_2^1 (y,t) \Big) \bigg] \Bigg\} \nn\\
 & & \times \exp\!\Big( {-} \tilde \phi (\xt,t) \phi (\xt,t) { -} \sum_{\alpha=1}^{n_{\alpha}}\sum_j \tilde\chi_j^\alpha (\xt,t)\chi_j^\alpha(\xt,t) \Big) \Big[1+\tilde \phi (\xt ,t{+}1)\phi (\xt,t) + \chi_{2}^1 (\xt,t) \Big].
\eea
The action is split into contributions for a generic point $x\in \ca V\backslash \xt$ (first three lines), and a contribution for the target $\xt$ (last line). 
This action now propagates any number of $\phi$ or $\psi$-particles. 
The part for $\chi$ is unchanged from \Eq{action-LRW-dynamic-1}. Since $\gamma$ should be thought of as being small, and 
only one $\psi$ and $\chi$ particle can exist per site, we consider the linear terms in the expansion of the exponential on the third line.
Unlike in \Eq{action-LRW-dynamic-1} there is no explicit exclusion of $\phi$, $\psi$ and $\chi$. 
However, our process does not produces these occurrences.

\subsection{Continuum limit of simplified lattice theory $\ca L^{\textsc{LRW}}_2$}
\label{Continuum limit for L-LERW-2}
Following the procedure of section \ref{Continuum limit for L-LERW-1}, we now expand $\ca L_2^{\textsc{LRW}}$ as defined in \Eq{action-LRW-dynamic-2} in the fields, retaining up to four fields. In a second step, we interpret fields at time $t+1 $ as $t+\delta t$, and expand them around $t$. 
Similarly, fields at $y$ are Taylor expanded around $x$. 
This gives (with target terms implicit)
\be \label{L-LRW-2}
\begin{array}{rcl}
 \ca L_2^{\textsc{LRW}} &\simeq& \ds {\int_x} \tilde \phi(x,t) [\partial_t-\gamma\nabla^2] \phi(x,t)+ \tilde \psi(x,t) \partial_t \psi(x,t) - 
 \sum_{j,\alpha} \tilde \chi_j^\alpha (x,t)\nabla^2 \chi_j^\alpha(x,t) \\
&+& \ds g \int_x \bigg\{ \sum_{j,\alpha} \tilde \psi(x,t)\psi(x,t) \tilde \chi_j^\alpha (x,t) \chi_j^\alpha(x,t)+ \frac{1}2 \sum_{\alpha} \Big[\sum_j \tilde \chi_j^\alpha (x,t) \chi_j^\alpha(x,t) \Big]^2 \bigg\} \\
&-& \ds
 \gamma \int _x \left\{ \tilde \phi (x,t) \tilde \chi^1_2 (x,t) \delta \tilde \psi(x,t) + \nabla^2\Big[\tilde \phi (x,t) \delta \tilde \chi^1_2 (x,t) \Big] - \tilde \phi(x,t) \nabla^2 \delta \tilde \chi^1_2(x,t)\right\} \phi (x,t) .
 \end{array} 
\ee
The free propagators are as given in \Eqs{propRed} to \eq{propGreen}.
There are much fewer interaction terms than in \Eq{relevant-g-terms}, namely 
\be \label{g-terms-2}
 \ca L_2^{g} = g \,\vertexIntPsiChi + \frac g2 \,\vertexIntChi \ .
\ee
The $\gamma$-terms are unchanged, 
\be
-\ca L_2^\gamma = \gamma\, \vertexOneDashed +\gamma\,\vertexTwoRel- \gamma\,\vertexTwoMinus = \gamma \,\vertexOneDashed +\gamma \vertexTwoBred + 2\gamma\,\vertexTwoCred.
\ee
To more easily refer to these terms, we write 
\be
-\ca L_2^\gamma = \gamma_1\, \vertexOneDashed +\gamma_2^+\,\vertexTwoRel- \gamma_2^{-}\,\vertexTwoMinus
\ee
and refer to the different terms as the $\gamma_1$, $\gamma_2^+$ and $\gamma_2^-$ terms. (The signs in the last two terms refer to the sign in the perturbative expansion of $- \ca L_2^{ \gamma}$.)

\subsection{Diffusion and drift}
\label{Diffusion and drift}
Both $\tilde \phi$ and $\phi$ appear linearly in the Lagrangian \eq{L-LRW-2} and the action $\ca S_2^{\textsc{LRW}}=\int_t \ca L_2^{\textsc{LRW}}$. 
Therefore, the variations w.r.t.~$\phi$ and $\tilde \phi$ lead to exact equations of motion. 
If we are not interested in the trace left, we can set $\tilde \psi=1$; accordingly the first equation of motion (EOM) is 
\bea\label{backward-Fokker-Planck}
\frac{\delta\ca S_2^{\textsc{LRW}} }{\delta \phi(x,t)} &=& 0, \\
-\partial_t \tilde \phi(x,t) &\simeq& \gamma \left[ \nabla^2\Big( \tilde \phi (x,t) \tilde \chi^1_2 (x,t) \Big) - \tilde \phi(x,t) \nabla^2 \tilde \chi^1_2(x,t) \right] \nn\\
&=& \gamma \left[ \tilde \chi^1_2 (x,t) \nabla^2 \tilde \phi (x,t) + 2 \nabla \tilde \phi (x,t) \nabla \tilde \chi^1_2 (x,t)\right] .
\eea
This is the backward Fokker-Planck equation for $\tilde \phi(x,t)$ (see e.g.~Eq.~(949) of \cite{Wiese2021} ), with 
\bea
\label{D-def}
D(x,t)&:=& \gamma \left< \tilde \chi^1_2(x,t)\chi_2^1 (\xt,t) \right> ,\\
\label{F-def}
F(x,t)&:=& 2 \gamma \nabla_x \left< \tilde \chi^1_2(x,t)\chi^1_2 (\xt,t) \right>.
\eea
This gives the probability to observe the path at time $t$ at $x$, while arriving at a later time $t'$ at $y$. 
\be
P_{\rm b}(x,t) := \left< \tilde \phi(x,t) \phi(y,t') \right> .
\ee
The second EOM is (again at $\tilde \psi=1$)
\bea
\frac{\delta \ca S_2^{\textsc{LRW}} }{\delta \tilde \phi(x,t)} &=& 0 ,\\
\partial_t \phi(x,t) &\simeq& \gamma \Big[ \tilde \chi_2^1(x,t) \nabla^2 \phi (x,t) 
- \phi (x,t) \nabla^2 \tilde \chi_2^1(x,t) \Big].
\label{64}
\eea
Define the probability to observe the path at time $t$ in $x$, given it started at time $t=0$ in $y$, 
\be
P_{\rm f}(x,t):= \left< \phi(x,t) \tilde \phi(y,0) \right> .
\ee
With the definitions in \Eqs{D-def} and \eq{F-def}, \Eq{64} can be rewritten as 
\bea\label{forward-Fokker-Planck}
\partial_t P_{\rm f}(x,t) &=& D (x,t) \nabla^2 P_{\rm f}(x,t) - P_{\rm f}(x,t) \nabla^2 D(x,t ) \nn \\
&= & \nabla^2\big[ D (x,t) P_{\rm f}(x,t) \big] - \nabla \left[ 2\nabla D(x,t) P_{\rm f}(x,t)\right] . 
\eea
This is the forward Fokker-Planck equation for $P_{\rm f}(x,t)$ given diffusion coefficient $D$ and drift (force) $F$ (see e.g.~Eq.~(944) of \cite{Wiese2021}). 

Both the backward Fokker-Planck equation \eq{backward-Fokker-Planck} and its forward version \eq{forward-Fokker-Planck} encode the same random process. They teach us that diffusion coefficient $D$ and drift $F$ are related, 
\be
F(x,t) = 2 \nabla D(x,t). 
\ee
The drift points in the direction where the potential $\Phi$ grows.

\subsection{Observable: passing through a point}
\label{Observable: passing through a point}
What we are mostly interested in is the configuration of the trace left by the LRW at time $t=T$. 
This trace is encoded in $\tilde \psi(x,T)$. Suppose we initialize the path at $t=0$ in $x=0$; then the generating functional is
\be
\left< \rme^{\int_x \lambda_x\psi(x,T) + \rho_x \phi(x,T) } \right> := 
\int_{\tilde \chi \chi}\int_{\tilde \phi \phi} \int_{\tilde \psi \psi} \rme^{-\int_t \ca L(t)} \tilde \phi(0,0) \rme^{\int_x \lambda_x\psi(x,T) + \rho_x \phi(x,T) }.
\ee 
If we want to know the probability that the path passed through $y$, and went on to $z$ at time $t=T$, we need to evaluate 
\bea
 \left< \ca O(y,z|T) \right> &:=& \frac{\delta }{\delta \lambda_y} \frac{\delta }{\delta \rho_z} \left< \rme^{\int_x \lambda_x\psi(x,T) + \rho_x \phi(x,T) } \right>\Big|_{\lambda,\rho=1} \nn\\
&=& \int_{\tilde \chi \chi}\int_{\tilde \phi \phi} \int_{\tilde \psi \psi} \rme^{-\int_t \ca L(t)}
\rme^{\int_x \psi(x,T) + \phi(x,T) } \tilde \phi(0,0) \psi(y,T) \phi(z,T) .
\label{9}
\eea
At the tree level, the only contribution is provided by the $\gamma_1$ vertex. Evaluating \Eq{9} step by step, we obtain 
\bea\label{56}
\left< \ca O(y,z|T)\right>^{\rm tree} &=& \gamma \left< \tilde \phi(0,0) \psi(y,T) \phi(z,T) 
 \int_{x,t}\vertexOneDashed {\!\!\!\!\!\!{}_x\,\,\,\,\,\,} \chi^1_2(\xt,t) \right>_{\!\!0} \nn\\
&=& \gamma\int_{x,t}
\left<\tilde \phi(0,0) \phi(x,t) \right>_{\!0} 
\left< \tilde \chi^1_2 (x,t) \chi^1_2(\xt,t)\right>_{\!0} 
\left< \delta \tilde \psi(x,t) \psi(y,T)\right>_{\!0} 
\left<\tilde \phi(x,t) \phi(z,T) \right>_{\!0} \nn\\
&=& \gamma\int_{x,t} R_{x,t} \times C_{x-\xt} \times \Theta(T>t)\delta_{x,y} \times R_{x-z,T-t} \nn\\
&=& \gamma\int_{t} R_{y,t} \times C_{y-\xt} \times R_{y-z,T-t} .
\eea
In the evaluation of $ \left< \tilde \chi^1_2 (x,t) \chi^1_2(\xt,t)\right>_0$ we only kept the leading term. 
\Eq{56} can further be simplified by doing a Laplace transform w.r.t.\ to $\gamma T$, and a Fourier transform w.r.t.\ $y$ and $z$. Assuming that the target is uniformly distributed, $ C_{y-\xt}$ is a constant we put to $1$. This yields
\bea\label{57}
{{{\parbox{2.7cm}{{\begin{tikzpicture}[scale=1]
\node (k+p+q) at (-1.05,0) {$\!{}_{k+p+q}$};
\node (k) at (-.15,1.05) {${}_{k}$};
\node (p) at (0.67,1.) {${}_{p}$};
\node (q) at (0.8,0.5) {${}_{q=0}$};
\coordinate (x1) at (-.5,0) ;
\coordinate (x2) at (0,0.5) ;
\coordinate (x2a) at (.5,1) ;
\coordinate (x2b) at (0,1) ;
\coordinate (x2c) at (0.5,.5) ;
\draw [thick,red,directed] (x1) -- (x2) ;
\draw [thick,blue,directed] (x2) -- (x2b);
\draw [thick,red,directed] (x2) -- (x2a);
\draw [thick,Green,enddirected] (x2) -- (x2c);
\fill (x1) circle (1.5pt);
\fill (x2) circle (1.5pt);
\fill (x2a) circle (1.5pt);
\fill (x2b) circle (1.5pt);
\end{tikzpicture}}}}}}
&=& \left< \ca O_\mu(p,k)\right>^{\rm tree} 
 = \int_{y,z} \int_{\gamma T}\rme^{i (k y+pz) -\mu \gamma T}
\left< \ca O(y,z|T)\right>^{\rm tree} \nn\\
&=& \gamma \int_{y,z} \int_{t,\gamma T}\rme^{i (k y+pz) -\mu \gamma T} R_{y,t}R_{y-z,T-t} C_{y-\xt} = \frac{1}{(k+p)^2 +\mu} \frac{1}{p^2 +\mu}.
\eea
We use the rule that momenta follow the arrows. 
The explicit $\gamma$ dependence has disappeared, canceling with the $\gamma$ from the exponential in \Eq{propRed} of the first red propagator. The second one is canceled explicitly in the definition of the Laplace transform.

\subsection{Perturbative corrections to the observable}
\label{1-loop corrections to the observer}
The tree diagram in \Eqs{56}-\eq{57} has three 1-loop corrections. 
With the momentum conventions in \Eq{57}, the first one is 
\be
 \diagOne \to \diagOnesimp =: \ca D_1.
\ee
The fluctuating part $\delta \tilde \chi$ cannot be used, thus $\tilde \chi \to 1$. 
(We show later in appendix \ref{app:pOperation Proof} that this is exact to all orders.)
This diagram is evaluated as 
\bea\label{13}
\diagOnesimp &=& \gamma^2 (-g)\int_q \int_{ t} \rme^{-[(k+p+q)^2 +\mu] \gamma t } \frac1{q^2}\vertexOneDashed \nn\\
&=& - \gamma g \int_q \frac1{(k{+}p{+}q)^2{+}\mu}\frac1{q^2} \vertexOneDashed
\xrightarrow{\mu\to 0}-\gamma g \,\propdiagExt_{k+p} \vertexOneDashed.
\eea
There are two more contributions, 
\be
 \diagOneB \to \diagOneBsimp =: \ca D_2 ~, \qquad 
- \diagOneC \to -\diagOneCsimp =: \ca D_3.
\ee
They evaluate to 
\begin{align}\label{15}
\diagOneBsimp &= \gamma^3 (-g)\int_q \int_{ t,t'} \rme^{-[(k+p+q)^2+\mu] \gamma t -[(p{+}q)^2+\mu]\gamma t'
} \frac{-(p{+}q)^2}{q^2} \vertexOneDashed\nn\\
& 
=\gamma g \int_q \frac1{(k{+}p{+}q)^2{+}\mu} \frac{(p{+}q)^2}{(p{+}q)^2{+}\mu}\frac1{q^2}
\xrightarrow{\mu\to 0}  \gamma g \,\propdiagExt_{k+p}\vertexOneDashed.
\end{align}

\begin{align}
\diagOneCsimp &= -\gamma^3 (-g)\int_q \int_{ t,t'} \rme^{-[(k{+}p{+}q)^2+\mu] \gamma t -[(p{+}q)^2{+}\mu]\gamma t'
} \frac{-q^2}{q^2} \vertexOneDashed\nn\\
& 
=- \gamma g \int_q \frac1{(k{+}p{+}q)^2{+}\mu}    \frac1{(p{+}q)^2{+}\mu} 
\xrightarrow{\mu\to 0}  -\gamma g \,\propdiagExt_{k} \vertexOneDashed.
\end{align}
We see that the derivatives have canceled against a propagator, leaving us with a factor of $(-1)$. 
This can be written both in momentum and in position space, 
\be
(-k^2) C(k) = -1 \quad \Longleftrightarrow \quad \nabla^2 C(x-y) = - \delta^d (x-y).
\ee 
Graphically, this yields
\be
\diagOneCsimp = - \diagOneCsimpSimp\ . 
\ee
\Eqs{13} and \eq{15} show that in the limit of $\mu \to 0$, for all possible external momenta
\be\label{D1+D2=0}
\ca D_1+ \ca D_2 = \diagOnesimp + \diagOneBsimp = 0 .
\ee
If we want to know what is the probability that starting at $x=0$ the trace passes through $y$, before going to a uniformly distributed target, then we can evaluate these diagrams at $p=0$, and obtain (again for $\mu \to 0$)
\be
\ca D_2+ \ca D_3 \big|_{p=0} = 0. 
\ee
This cancelation, if valid to higher loop-orders, would be very useful, as it takes out a pair of diagrams with derivatives, leaving us with a perturbative expansion in $\gamma_1$ only. While this is correct at leading order (1 loop), if we insert $\ca D_1$ to $\ca D_3$ into a larger diagram, then the subtle change in momentum dependence leads to different diagrams.
We show in section \ref{s:Cancellations of diagrams} and appendix \ref{app:pOperation Proof} one effective cancelation scheme. 
 
After time-integration, these diagrams give
\be\label{21}
{\diagOnesimp} + \diagOneBsimp - \diagOneCsimp = \diagOneCsimpSimp \; \to {\LERWdiagOneBlack} \ .
\ee
The diagram we show on the r.h.s.\ is the diagram remaining in $\phi^4$-theory, i.e.\ for LERWs. 
Colors are chosen to make this equivalence suggestive.

\subsection{Simplifying the $\chi$-theory: $\ca L_3^{\textsc{LRW}}$}
\label{Simplifying the chi-theory}
In the calculations of the preceding sections, we saw that the $\chi$-theory is only used to evaluate a 2-point function, namely $\left< \tilde \chi_2^1(x,t) \chi_2^1(\xt,t)\right>$. We know that this 2-point function is the same as that of a free theory. Therefore, we can simplify our theory one step further to 
\be \label{L-LERW-3}
\!\!\!{\begin{array}{rcl}
 \ca L_3^{\textsc{LRW}} &\simeq& \ds {\int_x} \tilde \phi(x,t) [\partial_t-\gamma\nabla^2] \phi(x,t)+ \tilde \psi(x,t) \partial_t \psi(x,t) - 
 \sum_{\alpha=1}^{n_\alpha} \tilde \chi^{\a} (x,t)\nabla^2 \chi^\a (x,t) \\
&+& \ds g \sum_{\alpha=1}^{n_\alpha} \int_x \tilde \psi(x,t)\psi(x,t) \tilde \chi^\a (x,t) \chi^\a (x,t) \\
&-& \ds
 \gamma \int _x\left\{ \tilde \phi (x,t) \tilde \chi^1 (x,t) \delta \tilde \psi(x,t) + \nabla^2\Big[\tilde \phi (x,t) \delta \tilde \chi^1 (x,t) \Big] - \tilde \phi(x,t) \nabla^2 \delta \tilde \chi^1(x,t)\right\} \phi (x,t) .
 \end{array} }
\ee
This theory only has a single vector field $\chi^\a$ with $n_\alpha=0$ components.
The purpose of this index is to eliminate loops made of $\sum_{\alpha=1}^{n_\alpha} \tilde \chi^{\a} (x,t)\nabla^2 \chi^\a (x,t)$, as for self-avoiding polymers. 
The vertices reduce to 
\bea \label{g-terms-3}
 \ca L_{3}^g &=& g \,\vertexIntPsiChi \to g \, \vertexIntPsiChiRel, 
\\
\label{relevant-gamma-vertices-3}
-\ca L_3^\gamma &=& \gamma_1\, \vertexOneDashed +\gamma_2^+\,\vertexTwoRel- \gamma_2^-\,\vertexTwoMinus = 
\gamma_1\, \vertexOneDashed +\gamma_{\textsc l}\, {\vertexTwoBred} + 2\gamma_{\textsc d}\, {\vertexTwoCred}.
\eea

\subsection{More cancelations of diagrams}
\label{s:Cancellations of diagrams}
In Section \ref{1-loop corrections to the observer} we saw that some diagrams cancel exactly. 
This happens more generally. While the precise explanation is given in appendix \ref{app:pOperation Proof}, we try to give the general idea here. The only surviving diagrams that contribute to the observable $\ca O(y,z|T)$ are those that contain no independent sub-loop (a red loop which does not emit a blue line, see definitions \ref{+sub-loop} and \ref{-sub-loop} in appendix \ref{app:pOperation Proof}), no $\gamma_2^+$, and as interaction only $g^1$. In particular, all diagrams containing $\gamma_2^+$, $\gamma_1^{\delta \chi}$ or $g^{\delta \chi}$ (in the notations of appendix \ref{a:notations}) vanish.

Let us give some examples: after time-integration, pairs of diagrams containing the following pattern cancel 
\begin{eqnarray}
	\interactingGammaOneBis &+& \GammaOneSubbedByGammaTwopbis \simeq 0, \\
	\interactingGvertexbis &+& \GvertexSubbedbyGammaTwoBis \simeq 0 .
\end{eqnarray}
A special case of the first identity is \Eq{D1+D2=0}. The second identity first shows up at 2-loop order. 

With just $\gamma_1^{1}$, $\gamma_2^-$ and $g^{1}$ vertices (see \Eq{L-A} for the definition), one can draw diagrams containing \emph{independent sub-loops}
without a $\gamma_1^{1}$ inside the loop, see the first diagram in \Eq{89} below. These cancel against the second type of diagrams in \Eq{89}, 
\be\label{89}
\GammaTwominusMakingLoop+ \GammaTwoplusMakingLoop \simeq 0
\ee
These simplifications allow us to get contributions to the observable by using $\gamma_1^{1}$, $\gamma_2^-$ and $g^{1}$ vertices only, subject to not forming any independent sub-loop.
For example, the first diagram below is not allowed due to an independent sub-loop (highlighted), while the second has no independent sub-loop and is allowed,
\begin{equation}
	\NOTallowedTwoLoop = \mbox{not allowed, } \quad 
	\allowedTwoLoop = \mbox{allowed.} 
\end{equation}
How these cancelations work in general, and how this allows one to show that our perturbation theory for LRWs agrees with that of LERWs is shown in Appendix \ref{app:pOperation Proof}.

\section{$b$-Laplacian random walks}
\label{b-Laplacian Random Walks}
The action \eq{L-LERW-3} is simple enough, such that we can generalize it to the $b$-Laplacian random walk. 
In the latter, instead of weighting the rate to go to a site adjacent to the tip with the solution of the Laplace equation, it is weighted by its $b$-th power. In the corresponding field theory, we solve the Laplace equation by $b$ independent $\chi$-fields, and use as weight the product over these $b$ fields.
\bea \nn
 \ca L^{b\textsc{-LRW}} &\simeq& \ds {\int_x} \tilde \phi(x,t) [\partial_t-\gamma\nabla^2] \phi(x,t)+ \tilde \psi(x,t) \partial_t \psi(x,t) - 
\sum_{{\alpha=1}}^{n_\alpha} \sum_{\mathfrak b=1}^b \tilde \chi^\alpha_{ \mathfrak b} (x,t)\nabla^2 \chi^\alpha_{ \mathfrak b} (x,t) \\
 \label{L-b-LRW} 
&+& \ds g \int_x \tilde \psi(x,t)\psi(x,t) \sum_{{\alpha=1}}^{n_\alpha} \sum_{\mathfrak b=1}^b\tilde \chi^\alpha_{ \mathfrak b} (x,t) \chi^\alpha_{ \mathfrak b} (x,t) \\
&-& \ds
 \gamma \int _x \left\{ \tilde \phi (x,t) \tilde \bchi (x,t) \delta \tilde \psi(x,t) + \nabla^2\Big[\tilde \phi (x,t) \delta \tilde \bchi (x,t) \Big] - \tilde \phi(x,t) \nabla^2 \delta \tilde \bchi(x,t)\right\} \phi (x,t) , \nn \\
\!\!\! \tilde \bchi(x,t) &=& \prod_{{\mathfrak b=1}}^b \tilde \chi_{\mathfrak b}^1(x,t), \qquad \delta \tilde \bchi (x,t) = \tilde \bchi (x,t)-1,
 \qquad \tilde \chi_{\mathfrak b}^1(x,t) = 1+ \delta \tilde \chi_{\mathfrak b}^1(x,t).
\eea
The last terms in the first line and the second line are the generalization of a single $\chi^{\alpha}$-field to $b$ copies $\chi_{\mathfrak b}^{\alpha}$, indexed by $\mathfrak b=1,...,b$.
When the index $b$ is omitted, we mean the first index $\mathfrak b=1$.

The effective potential steering the Laplacian random walk is $\tilde\bchi(x,t)$. When acting with the Laplacian on it, one gets a new type of interaction, since (with $b=2$ for simplification)
\be
\nabla^2\bchi(x,t) = \chi_1(x,t) \nabla^2 \chi_2(x,t) +\chi_2(x,t) \nabla^2\chi_1(x,t) +2 \nabla\chi_1(x,t) \nabla \chi_2(x,t).
\ee
It is the last vertex that complicates the diagrammatics for the $b$-Laplacian random walk \cite{PisapiaWiese2026}.

\section{Diffusion Limited Aggregation}
\label{DLA}
\subsection{Definition}\label{Definition}
Diffusion limited aggregation is a process where particles start a random walk at the target $\xt$ before hitting the seed or another point in the already constructed cluster, upon which they are placed at the last vertex before hitting it. The process stops when the cluster has reached the target. The probabilities for the different configurations may as for Laplacian random walks be obtained recursively, solving the Laplace equation. The only difference is that the cluster can now grow from any point, not only the tip of a path. 
We start with an example.

\subsection{Example}
\label{DLA-example}
Let us start our analysis by studying DLA growth on graph $\ca G_1$ defined in \Eq{graph-G1}. We find for the generating function
\bea\label{DLA-on-G1}
 \mathbb G^{\textsc{dla}} &=&
 \frac{9}{77}\parbox{20mm}{\ifTikzExternal{\tikzsetnextfilename{tikzAux6}}\begin{tikzpicture}[x=0.2mm,y=0.2mm]
\draw [line width=0.2mm, gray] (-50.,0.) -- (-15.4508,47.5528) ;
\draw [line width=0.2mm, gray] (40.4508,29.3893) -- (-15.4508,47.5528) ;
\draw [line width=0.2mm, gray] (-15.4508,47.5528) -- (-15.4508,-47.5528) ;
\draw [line width=0.3mm, blue,directed] (-50.,0.) -- (-15.4508,-47.5528) ;
\draw [line width=0.3mm, blue,directed] (-15.4508,-47.5528) -- (40.4508,-29.3893) ;
\draw [line width=0.3mm, blue,directed] (40.4508,-29.3893) -- (40.4508,29.3893) ;
\fill(-50.,0.) circle (1.5pt);
\fill(-15.4508,-47.5528) circle (1.5pt);
\fill(40.4508,-29.3893) circle (1.5pt);
\fill(40.4508,29.3893) circle (1.5pt);
\fill(-15.4508,47.5528) circle (1.5pt);
\end{tikzpicture}
}
 + \frac{30}{77}\parbox{20mm}{\ifTikzExternal{\tikzsetnextfilename{tikzAux7}}\begin{tikzpicture}[x=0.2mm,y=0.2mm]
\draw [line width=0.2mm, gray] (-50.,0.) -- (-15.4508,-47.5528) ;
\draw [line width=0.2mm, gray] (-15.4508,-47.5528) -- (40.4508,-29.3893) ;
\draw [line width=0.2mm, gray] (40.4508,-29.3893) -- (40.4508,29.3893) ;
\draw [line width=0.2mm, gray] (-15.4508,47.5528) -- (-15.4508,-47.5528) ;
\draw [line width=0.3mm, blue,directed] (-50.,0.) -- (-15.4508,47.5528) ;
\draw [line width=0.3mm, blue,directed] (-15.4508,47.5528) -- (40.4508,29.3893) ;
\fill(-50.,0.) circle (1.5pt);
\fill(-15.4508,-47.5528) circle (1.5pt);
\fill(40.4508,-29.3893) circle (1.5pt);
\fill(40.4508,29.3893) circle (1.5pt);
\fill(-15.4508,47.5528) circle (1.5pt);
\end{tikzpicture}
}
 + \frac{32}{231}\parbox{20mm}{\ifTikzExternal{\tikzsetnextfilename{tikzAux8}}\begin{tikzpicture}[x=0.2mm,y=0.2mm]
\draw [line width=0.2mm, gray] (-15.4508,-47.5528) -- (40.4508,-29.3893) ;
\draw [line width=0.2mm, gray] (40.4508,-29.3893) -- (40.4508,29.3893) ;
\draw [line width=0.2mm, gray] (-15.4508,47.5528) -- (-15.4508,-47.5528) ;
\draw [line width=0.3mm, blue,directed] (-50.,0.) -- (-15.4508,-47.5528) ;
\draw [line width=0.3mm, blue,directed] (-50.,0.) -- (-15.4508,47.5528) ;
\draw [line width=0.3mm, blue,directed] (-15.4508,47.5528) -- (40.4508,29.3893) ;
\fill(-50.,0.) circle (1.5pt);
\fill(-15.4508,-47.5528) circle (1.5pt);
\fill(40.4508,-29.3893) circle (1.5pt);
\fill(40.4508,29.3893) circle (1.5pt);
\fill(-15.4508,47.5528) circle (1.5pt);
\end{tikzpicture}
}
 + \frac{20}{231}\parbox{20mm}{\ifTikzExternal{\tikzsetnextfilename{tikzAux9}}\begin{tikzpicture}[x=0.2mm,y=0.2mm]
\draw [line width=0.2mm, gray] (-50.,0.) -- (-15.4508,47.5528) ;
\draw [line width=0.2mm, gray] (-15.4508,-47.5528) -- (40.4508,-29.3893) ;
\draw [line width=0.2mm, gray] (40.4508,-29.3893) -- (40.4508,29.3893) ;
\draw [line width=0.3mm, blue,directed] (-50.,0.) -- (-15.4508,-47.5528) ;
\draw [line width=0.3mm, blue,directed] (-15.4508,-47.5528) -- (-15.4508,47.5528) ;
\draw [line width=0.3mm, blue,directed] (-15.4508,47.5528) -- (40.4508,29.3893) ;
\fill(-50.,0.) circle (1.5pt);
\fill(-15.4508,-47.5528) circle (1.5pt);
\fill(40.4508,-29.3893) circle (1.5pt);
\fill(40.4508,29.3893) circle (1.5pt);
\fill(-15.4508,47.5528) circle (1.5pt);
\end{tikzpicture}
}\nn\\
 &+& \frac{4}{77}\parbox{20mm}{\ifTikzExternal{\tikzsetnextfilename{tikzAux10}}\begin{tikzpicture}[x=0.2mm,y=0.2mm]
\draw [line width=0.2mm, gray] (-50.,0.) -- (-15.4508,-47.5528) ;
\draw [line width=0.2mm, gray] (-15.4508,-47.5528) -- (40.4508,-29.3893) ;
\draw [line width=0.2mm, gray] (40.4508,-29.3893) -- (40.4508,29.3893) ;
\draw [line width=0.3mm, blue,directed] (-50.,0.) -- (-15.4508,47.5528) ;
\draw [line width=0.3mm, blue,directed] (-15.4508,47.5528) -- (-15.4508,-47.5528) ;
\draw [line width=0.3mm, blue,directed] (-15.4508,47.5528) -- (40.4508,29.3893) ;
\fill(-50.,0.) circle (1.5pt);
\fill(-15.4508,-47.5528) circle (1.5pt);
\fill(40.4508,-29.3893) circle (1.5pt);
\fill(40.4508,29.3893) circle (1.5pt);
\fill(-15.4508,47.5528) circle (1.5pt);
\end{tikzpicture}
}
 + \frac{25}{462}\parbox{20mm}{\ifTikzExternal{\tikzsetnextfilename{tikzAux11}}\begin{tikzpicture}[x=0.2mm,y=0.2mm]
\draw [line width=0.2mm, gray] (40.4508,29.3893) -- (-15.4508,47.5528) ;
\draw [line width=0.2mm, gray] (-15.4508,47.5528) -- (-15.4508,-47.5528) ;
\draw [line width=0.3mm, blue,directed] (-50.,0.) -- (-15.4508,-47.5528) ;
\draw [line width=0.3mm, blue,directed] (-50.,0.) -- (-15.4508,47.5528) ;
\draw [line width=0.3mm, blue,directed] (-15.4508,-47.5528) -- (40.4508,-29.3893) ;
\draw [line width=0.3mm, blue,directed] (40.4508,-29.3893) -- (40.4508,29.3893) ;
\fill(-50.,0.) circle (1.5pt);
\fill(-15.4508,-47.5528) circle (1.5pt);
\fill(40.4508,-29.3893) circle (1.5pt);
\fill(40.4508,29.3893) circle (1.5pt);
\fill(-15.4508,47.5528) circle (1.5pt);
\end{tikzpicture}
}
 + \frac{19}{462}\parbox{20mm}{\ifTikzExternal{\tikzsetnextfilename{tikzAux12}}\begin{tikzpicture}[x=0.2mm,y=0.2mm]
\draw [line width=0.2mm, gray] (-50.,0.) -- (-15.4508,47.5528) ;
\draw [line width=0.2mm, gray] (40.4508,29.3893) -- (-15.4508,47.5528) ;
\draw [line width=0.3mm, blue,directed] (-50.,0.) -- (-15.4508,-47.5528) ;
\draw [line width=0.3mm, blue,directed] (-15.4508,-47.5528) -- (40.4508,-29.3893) ;
\draw [line width=0.3mm, blue,directed] (-15.4508,-47.5528) -- (-15.4508,47.5528) ;
\draw [line width=0.3mm, blue,directed] (40.4508,-29.3893) -- (40.4508,29.3893) ;
\fill(-50.,0.) circle (1.5pt);
\fill(-15.4508,-47.5528) circle (1.5pt);
\fill(40.4508,-29.3893) circle (1.5pt);
\fill(40.4508,29.3893) circle (1.5pt);
\fill(-15.4508,47.5528) circle (1.5pt);
\end{tikzpicture}
}
 + \frac{1}{77}\parbox{20mm}{\ifTikzExternal{\tikzsetnextfilename{tikzAux13}}\begin{tikzpicture}[x=0.2mm,y=0.2mm]
\draw [line width=0.2mm, gray] (-50.,0.) -- (-15.4508,-47.5528) ;
\draw [line width=0.2mm, gray] (40.4508,29.3893) -- (-15.4508,47.5528) ;
\draw [line width=0.3mm, blue,directed] (-50.,0.) -- (-15.4508,47.5528) ;
\draw [line width=0.3mm, blue,directed] (-15.4508,-47.5528) -- (40.4508,-29.3893) ;
\draw [line width=0.3mm, blue,directed] (40.4508,-29.3893) -- (40.4508,29.3893) ;
\draw [line width=0.3mm, blue,directed] (-15.4508,47.5528) -- (-15.4508,-47.5528) ;
\fill(-50.,0.) circle (1.5pt);
\fill(-15.4508,-47.5528) circle (1.5pt);
\fill(40.4508,-29.3893) circle (1.5pt);
\fill(40.4508,29.3893) circle (1.5pt);
\fill(-15.4508,47.5528) circle (1.5pt);
\end{tikzpicture}
}\nn\\
 &+& \frac{25}{462}\parbox{20mm}{\ifTikzExternal{\tikzsetnextfilename{tikzAux14}}\begin{tikzpicture}[x=0.2mm,y=0.2mm]
\draw [line width=0.2mm, gray] (40.4508,-29.3893) -- (40.4508,29.3893) ;
\draw [line width=0.2mm, gray] (-15.4508,47.5528) -- (-15.4508,-47.5528) ;
\draw [line width=0.3mm, blue,directed] (-50.,0.) -- (-15.4508,-47.5528) ;
\draw [line width=0.3mm, blue,directed] (-50.,0.) -- (-15.4508,47.5528) ;
\draw [line width=0.3mm, blue,directed] (-15.4508,-47.5528) -- (40.4508,-29.3893) ;
\draw [line width=0.3mm, blue,directed] (-15.4508,47.5528) -- (40.4508,29.3893) ;
\fill(-50.,0.) circle (1.5pt);
\fill(-15.4508,-47.5528) circle (1.5pt);
\fill(40.4508,-29.3893) circle (1.5pt);
\fill(40.4508,29.3893) circle (1.5pt);
\fill(-15.4508,47.5528) circle (1.5pt);
\end{tikzpicture}
}
 + \frac{19}{462}\parbox{20mm}{\ifTikzExternal{\tikzsetnextfilename{tikzAux15}}\begin{tikzpicture}[x=0.2mm,y=0.2mm]
\draw [line width=0.2mm, gray] (-50.,0.) -- (-15.4508,47.5528) ;
\draw [line width=0.2mm, gray] (40.4508,-29.3893) -- (40.4508,29.3893) ;
\draw [line width=0.3mm, blue,directed] (-50.,0.) -- (-15.4508,-47.5528) ;
\draw [line width=0.3mm, blue,directed] (-15.4508,-47.5528) -- (40.4508,-29.3893) ;
\draw [line width=0.3mm, blue,directed] (-15.4508,-47.5528) -- (-15.4508,47.5528) ;
\draw [line width=0.3mm, blue,directed] (-15.4508,47.5528) -- (40.4508,29.3893) ;
\fill(-50.,0.) circle (1.5pt);
\fill(-15.4508,-47.5528) circle (1.5pt);
\fill(40.4508,-29.3893) circle (1.5pt);
\fill(40.4508,29.3893) circle (1.5pt);
\fill(-15.4508,47.5528) circle (1.5pt);
\end{tikzpicture}
}
 + \frac{1}{77}\parbox{20mm}{\ifTikzExternal{\tikzsetnextfilename{tikzAux16}}\begin{tikzpicture}[x=0.2mm,y=0.2mm]
\draw [line width=0.2mm, gray] (-50.,0.) -- (-15.4508,-47.5528) ;
\draw [line width=0.2mm, gray] (40.4508,-29.3893) -- (40.4508,29.3893) ;
\draw [line width=0.3mm, blue,directed] (-50.,0.) -- (-15.4508,47.5528) ;
\draw [line width=0.3mm, blue,directed] (-15.4508,-47.5528) -- (40.4508,-29.3893) ;
\draw [line width=0.3mm, blue,directed] (-15.4508,47.5528) -- (-15.4508,-47.5528) ;
\draw [line width=0.3mm, blue,directed] (-15.4508,47.5528) -- (40.4508,29.3893) ;
\fill(-50.,0.) circle (1.5pt);
\fill(-15.4508,-47.5528) circle (1.5pt);
\fill(40.4508,-29.3893) circle (1.5pt);
\fill(40.4508,29.3893) circle (1.5pt);
\fill(-15.4508,47.5528) circle (1.5pt);
\end{tikzpicture}
}.
\eea
How were the weights calculated? 
The tricky part is that each final configuration can be constructed in different orders. As an example consider
\be
\frac{32}{231}\parbox{20mm}{\ifTikzExternal{\tikzsetnextfilename{tikzAux17}}\begin{tikzpicture}[x=0.2mm,y=0.2mm]
\draw [line width=0.2mm, gray] (-15.4508,-47.5528) -- (40.4508,-29.3893) ;
\draw [line width=0.2mm, gray] (40.4508,-29.3893) -- (40.4508,29.3893) ;
\draw [line width=0.2mm, gray] (-15.4508,47.5528) -- (-15.4508,-47.5528) ;
\draw [line width=0.3mm, blue,directed] (-50.,0.) -- (-15.4508,-47.5528) ;
\draw [line width=0.3mm, blue,directed] (-50.,0.) -- (-15.4508,47.5528) ;
\draw [line width=0.3mm, blue,directed] (-15.4508,47.5528) -- (40.4508,29.3893) ;
\fill(-50.,0.) circle (1.5pt);
\fill(-15.4508,-47.5528) circle (1.5pt);
\fill(40.4508,-29.3893) circle (1.5pt);
\fill(40.4508,29.3893) circle (1.5pt);
\fill(-15.4508,47.5528) circle (1.5pt);
\end{tikzpicture}
} =\frac{20}{231} \parbox{20mm}{\ifTikzExternal{\tikzsetnextfilename{tikzAux18}}\begin{tikzpicture}[x=0.2mm,y=0.2mm]
\draw [line width=0.2mm, gray] (-15.4508,-47.5528) -- (40.4508,-29.3893) ;
\draw [line width=0.2mm, gray] (40.4508,-29.3893) -- (40.4508,29.3893) ;
\draw [line width=0.2mm, gray] (-15.4508,47.5528) -- (-15.4508,-47.5528) ;
\draw [line width=0.3mm, blue,directed] (-50.,0.) -- (-15.4508,-47.5528) ;
\draw [line width=0.3mm, blue,directed] (-50.,0.) -- (-15.4508,47.5528) ;
\draw [line width=0.3mm, blue,directed] (-15.4508,47.5528) -- (40.4508,29.3893) ;
\fill(-50.,0.) circle (1.5pt);
\fill(-15.4508,-47.5528) circle (1.5pt);
\fill(40.4508,-29.3893) circle (1.5pt);
\fill(40.4508,29.3893) circle (1.5pt);
\fill(-15.4508,47.5528) circle (1.5pt);
\node at (-40,30) {\color{red}\small 2};
\node at (-40,-30) {\color{red}\small 1};
\node at (30,45) {\color{red}\small 3};
\end{tikzpicture}
}+\frac{4}{77}\parbox{20mm}{\ifTikzExternal{\tikzsetnextfilename{tikzAux19}}\begin{tikzpicture}[x=0.2mm,y=0.2mm]
\draw [line width=0.2mm, gray] (-15.4508,-47.5528) -- (40.4508,-29.3893) ;
\draw [line width=0.2mm, gray] (40.4508,-29.3893) -- (40.4508,29.3893) ;
\draw [line width=0.2mm, gray] (-15.4508,47.5528) -- (-15.4508,-47.5528) ;
\draw [line width=0.3mm, blue,directed] (-50.,0.) -- (-15.4508,-47.5528) ;
\draw [line width=0.3mm, blue,directed] (-50.,0.) -- (-15.4508,47.5528) ;
\draw [line width=0.3mm, blue,directed] (-15.4508,47.5528) -- (40.4508,29.3893) ;
\fill(-50.,0.) circle (1.5pt);
\fill(-15.4508,-47.5528) circle (1.5pt);
\fill(40.4508,-29.3893) circle (1.5pt);
\fill(40.4508,29.3893) circle (1.5pt);
\fill(-15.4508,47.5528) circle (1.5pt);
\node at (-40,30) {\color{red}\small 1};
\node at (-40,-30) {\color{red}\small 2};
\node at (30,45) {\color{red}\small 3};
\end{tikzpicture}
}.
\ee
Here the numbers in red indicate in which order the edges were added. Each order is a {\em history}. 
For each history, we can evaluate its probability by solving the Laplace equation. 
This is how we arrived at the probabilities in \Eq{DLA-on-G1}. 
We checked our procedure with a Monte Carlo simulation, starting the particles at the target, and recording the graphs including their history. 
It is important to realize that while DLA is {\em history dependent}, LRW is {\em history independent}, since each path can only be grown in a single order. 
The question therefore arises whether a simple static theory can be constructed. We believe this to be impossible. Below we show that a generalization of the action \eq{eq:S_boson2} for LERWs to DLA gives the wrong weights.

\begin{figure}[t]
\centerline{\fig{0.5}{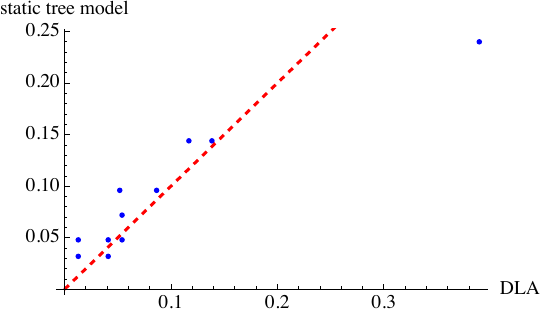}}
\caption{Plot of the weights in \Eq{conjecture-on-G1} against those in \Eq{DLA-on-G1}. Equal weights are indicated by a densely dashed red line.}
\label{compare-conjecture}
\end{figure}

\subsection{A conjecture}
\label{DLA-conjecture}
A natural conjecture is to take the action \eq{eq:S_boson2}, and allow for branching. This is a theory of trees in a background of loops with weight $-1$ per oriented loop. Its action reads
\bea \label{action-1-DLAc}
\!\!\!\rme^{-\ca S_{\ca G}^{\textsc{TL}}} \!&=& \prod_{x \in \ca V} \rme^{- \sum_i \tilde \phi_i(x)\phi_i(x) - \sum_j \tilde \chi_j(x)\chi_j(x)} 
 \Big\{ 1{+}\! \sum_{i}\! \prod_{y\in \ca V} \Big[ 1{+} {\beta_{xy}} \tilde \phi_i(y)\Big] \phi_i(x) {+}\!\sum_{y\in \ca V} \sum_{j} {\beta_{xy}} \tilde \chi_j(y)\chi_j(x) \Big\} \nn \\
 &\equiv & \prod_{x \in \ca V} \rme^{- \sum_i \tilde \phi_i(x)\phi_i(x) - \sum_j \tilde \chi_j(x)\chi_j(x)} 
 \Big\{ 1 {+} \!\sum_{i} \rme^{\sum_{y \in \ca V}\beta_{xy} \tilde \phi_i(y)} \phi_i(x) {+}\!\sum_{y\in \ca V} \sum_{j} {\beta_{xy}} \tilde \chi_j(y)\chi_j(x) \Big\} .
\eea
Here $i=1,...,0$ denotes the components of the cluster; having the index run to $0$ ensures that we get a single connected object.
The second index
$j=1,...,-1$ denotes the complex boson background (equivalent to one complex fermion) that solves the Laplace equation. The first and second line are equal, since even if more $\tilde\phi$ are offered in the second, they cannot be contracted as the out-degree of any vertex is either zero or one. We can choose whichever gives the better field theory.

We may conjecture that 
$
\langle \tilde \phi_1(\xr) \phi_1 (\xt) \rangle 
$
is the probability to have a DLA seeded at $\xr$ and ending at $\xt$. (In a simulation $\xt$ is the point where particles are injected.) 
We studied this conjecture on graph $\ca G_1$. Our first concern is the normalization; if we take the normalization of the $\ca O(-2)$ theory, 
then all LRW configurations are given as before, while branching trees are extra. Instead we normalize such that the prefactors add up to $1$. 
\bea
\label{conjecture-on-G1}
\mathbb G &=& \frac{18}{125}\parbox{20mm}{\ifTikzExternal{\tikzsetnextfilename{tikzAux20}}\begin{tikzpicture}[x=0.2mm,y=0.2mm]
\draw [line width=0.2mm, gray] (-50.,0.) -- (-15.4508,47.5528) ;
\draw [line width=0.2mm, gray] (40.4508,29.3893) -- (-15.4508,47.5528) ;
\draw [line width=0.2mm, gray] (-15.4508,47.5528) -- (-15.4508,-47.5528) ;
\draw [line width=0.3mm, blue,directed] (-50.,0.) -- (-15.4508,-47.5528) ;
\draw [line width=0.3mm, blue,directed] (-15.4508,-47.5528) -- (40.4508,-29.3893) ;
\draw [line width=0.3mm, blue,directed] (40.4508,-29.3893) -- (40.4508,29.3893) ;
\fill(-50.,0.) circle (1.5pt);
\fill(-15.4508,-47.5528) circle (1.5pt);
\fill(40.4508,-29.3893) circle (1.5pt);
\fill(40.4508,29.3893) circle (1.5pt);
\fill(-15.4508,47.5528) circle (1.5pt);
\end{tikzpicture}
}
 + \frac{36}{125}\parbox{20mm}{\ifTikzExternal{\tikzsetnextfilename{tikzAux21}}\begin{tikzpicture}[x=0.2mm,y=0.2mm]
\draw [line width=0.2mm, gray] (-50.,0.) -- (-15.4508,-47.5528) ;
\draw [line width=0.2mm, gray] (-15.4508,-47.5528) -- (40.4508,-29.3893) ;
\draw [line width=0.2mm, gray] (40.4508,-29.3893) -- (40.4508,29.3893) ;
\draw [line width=0.2mm, gray] (-15.4508,47.5528) -- (-15.4508,-47.5528) ;
\draw [line width=0.3mm, blue,directed] (-50.,0.) -- (-15.4508,47.5528) ;
\draw [line width=0.3mm, blue,directed] (-15.4508,47.5528) -- (40.4508,29.3893) ;
\fill(-50.,0.) circle (1.5pt);
\fill(-15.4508,-47.5528) circle (1.5pt);
\fill(40.4508,-29.3893) circle (1.5pt);
\fill(40.4508,29.3893) circle (1.5pt);
\fill(-15.4508,47.5528) circle (1.5pt);
\end{tikzpicture}
}
 -\frac{6}{125}\parbox{20mm}{\ifTikzExternal{\tikzsetnextfilename{tikzAux22}}\begin{tikzpicture}[x=0.2mm,y=0.2mm]
\draw [line width=0.2mm, gray] (-50.,0.) -- (-15.4508,-47.5528) ;
\draw [line width=0.2mm, gray] (40.4508,-29.3893) -- (40.4508,29.3893) ;
\draw [line width=0.2mm, gray] (-15.4508,47.5528) -- (-15.4508,-47.5528) ;
\draw [line width=0.3mm, blue,directed] (-50.,0.) -- (-15.4508,47.5528) ;
\draw [line width=0.3mm, blue,directed] plot [smooth,tension=1] coordinates { (-15.4508,-47.5528) (14.0451,-43.2263) (40.4508,-29.3893) };
\draw [line width=0.3mm, blue,directed] plot [smooth,tension=1] coordinates { (40.4508,-29.3893) (10.9549,-33.7158) (-15.4508,-47.5528) };
\draw [line width=0.3mm, blue,directed] (-15.4508,47.5528) -- (40.4508,29.3893) ;
\fill(-50.,0.) circle (1.5pt);
\fill(-15.4508,-47.5528) circle (1.5pt);
\fill(40.4508,-29.3893) circle (1.5pt);
\fill(40.4508,29.3893) circle (1.5pt);
\fill(-15.4508,47.5528) circle (1.5pt);
\end{tikzpicture}
}
 + \frac{18}{125}\parbox{20mm}{\ifTikzExternal{\tikzsetnextfilename{tikzAux23}}\begin{tikzpicture}[x=0.2mm,y=0.2mm]
\draw [line width=0.2mm, gray] (-15.4508,-47.5528) -- (40.4508,-29.3893) ;
\draw [line width=0.2mm, gray] (40.4508,-29.3893) -- (40.4508,29.3893) ;
\draw [line width=0.2mm, gray] (-15.4508,47.5528) -- (-15.4508,-47.5528) ;
\draw [line width=0.3mm, blue,directed] (-50.,0.) -- (-15.4508,-47.5528) ;
\draw [line width=0.3mm, blue,directed] (-50.,0.) -- (-15.4508,47.5528) ;
\draw [line width=0.3mm, blue,directed] (-15.4508,47.5528) -- (40.4508,29.3893) ;
\fill(-50.,0.) circle (1.5pt);
\fill(-15.4508,-47.5528) circle (1.5pt);
\fill(40.4508,-29.3893) circle (1.5pt);
\fill(40.4508,29.3893) circle (1.5pt);
\fill(-15.4508,47.5528) circle (1.5pt);
\end{tikzpicture}
}
 + \frac{12}{125}\parbox{20mm}{\ifTikzExternal{\tikzsetnextfilename{tikzAux24}}\begin{tikzpicture}[x=0.2mm,y=0.2mm]
\draw [line width=0.2mm, gray] (-50.,0.) -- (-15.4508,47.5528) ;
\draw [line width=0.2mm, gray] (-15.4508,-47.5528) -- (40.4508,-29.3893) ;
\draw [line width=0.2mm, gray] (40.4508,-29.3893) -- (40.4508,29.3893) ;
\draw [line width=0.3mm, blue,directed] (-50.,0.) -- (-15.4508,-47.5528) ;
\draw [line width=0.3mm, blue,directed] (-15.4508,-47.5528) -- (-15.4508,47.5528) ;
\draw [line width=0.3mm, blue,directed] (-15.4508,47.5528) -- (40.4508,29.3893) ;
\fill(-50.,0.) circle (1.5pt);
\fill(-15.4508,-47.5528) circle (1.5pt);
\fill(40.4508,-29.3893) circle (1.5pt);
\fill(40.4508,29.3893) circle (1.5pt);
\fill(-15.4508,47.5528) circle (1.5pt);
\end{tikzpicture}
} \nn\\
& +& 
 \frac{12}{125}\parbox{20mm}{\ifTikzExternal{\tikzsetnextfilename{tikzAux25}}\begin{tikzpicture}[x=0.2mm,y=0.2mm]
\draw [line width=0.2mm, gray] (-50.,0.) -- (-15.4508,-47.5528) ;
\draw [line width=0.2mm, gray] (-15.4508,-47.5528) -- (40.4508,-29.3893) ;
\draw [line width=0.2mm, gray] (40.4508,-29.3893) -- (40.4508,29.3893) ;
\draw [line width=0.3mm, blue,directed] (-50.,0.) -- (-15.4508,47.5528) ;
\draw [line width=0.3mm, blue,directed] (-15.4508,47.5528) -- (-15.4508,-47.5528) ;
\draw [line width=0.3mm, blue,directed] (-15.4508,47.5528) -- (40.4508,29.3893) ;
\fill(-50.,0.) circle (1.5pt);
\fill(-15.4508,-47.5528) circle (1.5pt);
\fill(40.4508,-29.3893) circle (1.5pt);
\fill(40.4508,29.3893) circle (1.5pt);
\fill(-15.4508,47.5528) circle (1.5pt);
\end{tikzpicture}
} + \frac{9}{125}\parbox{20mm}{\ifTikzExternal{\tikzsetnextfilename{tikzAux26}}\begin{tikzpicture}[x=0.2mm,y=0.2mm]
\draw [line width=0.2mm, gray] (40.4508,29.3893) -- (-15.4508,47.5528) ;
\draw [line width=0.2mm, gray] (-15.4508,47.5528) -- (-15.4508,-47.5528) ;
\draw [line width=0.3mm, blue,directed] (-50.,0.) -- (-15.4508,-47.5528) ;
\draw [line width=0.3mm, blue,directed] (-50.,0.) -- (-15.4508,47.5528) ;
\draw [line width=0.3mm, blue,directed] (-15.4508,-47.5528) -- (40.4508,-29.3893) ;
\draw [line width=0.3mm, blue,directed] (40.4508,-29.3893) -- (40.4508,29.3893) ;
\fill(-50.,0.) circle (1.5pt);
\fill(-15.4508,-47.5528) circle (1.5pt);
\fill(40.4508,-29.3893) circle (1.5pt);
\fill(40.4508,29.3893) circle (1.5pt);
\fill(-15.4508,47.5528) circle (1.5pt);
\end{tikzpicture}
}
 + \frac{6}{125}\parbox{20mm}{\ifTikzExternal{\tikzsetnextfilename{tikzAux27}}\begin{tikzpicture}[x=0.2mm,y=0.2mm]
\draw [line width=0.2mm, gray] (-50.,0.) -- (-15.4508,47.5528) ;
\draw [line width=0.2mm, gray] (40.4508,29.3893) -- (-15.4508,47.5528) ;
\draw [line width=0.3mm, blue,directed] (-50.,0.) -- (-15.4508,-47.5528) ;
\draw [line width=0.3mm, blue,directed] (-15.4508,-47.5528) -- (40.4508,-29.3893) ;
\draw [line width=0.3mm, blue,directed] (-15.4508,-47.5528) -- (-15.4508,47.5528) ;
\draw [line width=0.3mm, blue,directed] (40.4508,-29.3893) -- (40.4508,29.3893) ;
\fill(-50.,0.) circle (1.5pt);
\fill(-15.4508,-47.5528) circle (1.5pt);
\fill(40.4508,-29.3893) circle (1.5pt);
\fill(40.4508,29.3893) circle (1.5pt);
\fill(-15.4508,47.5528) circle (1.5pt);
\end{tikzpicture}
}
 + \frac{6}{125}\parbox{20mm}{\ifTikzExternal{\tikzsetnextfilename{tikzAux28}}\begin{tikzpicture}[x=0.2mm,y=0.2mm]
\draw [line width=0.2mm, gray] (-50.,0.) -- (-15.4508,-47.5528) ;
\draw [line width=0.2mm, gray] (40.4508,29.3893) -- (-15.4508,47.5528) ;
\draw [line width=0.3mm, blue,directed] (-50.,0.) -- (-15.4508,47.5528) ;
\draw [line width=0.3mm, blue,directed] (-15.4508,-47.5528) -- (40.4508,-29.3893) ;
\draw [line width=0.3mm, blue,directed] (40.4508,-29.3893) -- (40.4508,29.3893) ;
\draw [line width=0.3mm, blue,directed] (-15.4508,47.5528) -- (-15.4508,-47.5528) ;
\fill(-50.,0.) circle (1.5pt);
\fill(-15.4508,-47.5528) circle (1.5pt);
\fill(40.4508,-29.3893) circle (1.5pt);
\fill(40.4508,29.3893) circle (1.5pt);
\fill(-15.4508,47.5528) circle (1.5pt);
\end{tikzpicture}
}\nn\\
& +& \frac{6}{125}\parbox{20mm}{\ifTikzExternal{\tikzsetnextfilename{tikzAux29}}\begin{tikzpicture}[x=0.2mm,y=0.2mm]
\draw [line width=0.2mm, gray] (40.4508,-29.3893) -- (40.4508,29.3893) ;
\draw [line width=0.2mm, gray] (-15.4508,47.5528) -- (-15.4508,-47.5528) ;
\draw [line width=0.3mm, blue,directed] (-50.,0.) -- (-15.4508,-47.5528) ;
\draw [line width=0.3mm, blue,directed] (-50.,0.) -- (-15.4508,47.5528) ;
\draw [line width=0.3mm, blue,directed] (-15.4508,-47.5528) -- (40.4508,-29.3893) ;
\draw [line width=0.3mm, blue,directed] (-15.4508,47.5528) -- (40.4508,29.3893) ;
\fill(-50.,0.) circle (1.5pt);
\fill(-15.4508,-47.5528) circle (1.5pt);
\fill(40.4508,-29.3893) circle (1.5pt);
\fill(40.4508,29.3893) circle (1.5pt);
\fill(-15.4508,47.5528) circle (1.5pt);
\end{tikzpicture}
}
+\frac{4}{125}\parbox{20mm}{\ifTikzExternal{\tikzsetnextfilename{tikzAux30}}\begin{tikzpicture}[x=0.2mm,y=0.2mm]
\draw [line width=0.2mm, gray] (-50.,0.) -- (-15.4508,47.5528) ;
\draw [line width=0.2mm, gray] (40.4508,-29.3893) -- (40.4508,29.3893) ;
\draw [line width=0.3mm, blue,directed] (-50.,0.) -- (-15.4508,-47.5528) ;
\draw [line width=0.3mm, blue,directed] (-15.4508,-47.5528) -- (40.4508,-29.3893) ;
\draw [line width=0.3mm, blue,directed] (-15.4508,-47.5528) -- (-15.4508,47.5528) ;
\draw [line width=0.3mm, blue,directed] (-15.4508,47.5528) -- (40.4508,29.3893) ;
\fill(-50.,0.) circle (1.5pt);
\fill(-15.4508,-47.5528) circle (1.5pt);
\fill(40.4508,-29.3893) circle (1.5pt);
\fill(40.4508,29.3893) circle (1.5pt);
\fill(-15.4508,47.5528) circle (1.5pt);
\end{tikzpicture}
}
 + \frac{4}{125}\parbox{20mm}{\ifTikzExternal{\tikzsetnextfilename{tikzAux31}}\begin{tikzpicture}[x=0.2mm,y=0.2mm]
\draw [line width=0.2mm, gray] (-50.,0.) -- (-15.4508,-47.5528) ;
\draw [line width=0.2mm, gray] (40.4508,-29.3893) -- (40.4508,29.3893) ;
\draw [line width=0.3mm, blue,directed] (-50.,0.) -- (-15.4508,47.5528) ;
\draw [line width=0.3mm, blue,directed] (-15.4508,-47.5528) -- (40.4508,-29.3893) ;
\draw [line width=0.3mm, blue,directed] (-15.4508,47.5528) -- (-15.4508,-47.5528) ;
\draw [line width=0.3mm, blue,directed] (-15.4508,47.5528) -- (40.4508,29.3893) ;
\fill(-50.,0.) circle (1.5pt);
\fill(-15.4508,-47.5528) circle (1.5pt);
\fill(40.4508,-29.3893) circle (1.5pt);
\fill(40.4508,29.3893) circle (1.5pt);
\fill(-15.4508,47.5528) circle (1.5pt);
\end{tikzpicture}
}.
\eea
Note that the diagrams are the same as in \Eq{DLA-on-G1}, given in the same order, and that the second diagram appears twice, once by itself, once with a loop on the lower two vertices. 
Comparison shows that the weights are different, even though there is some measure of correlation, as can be seen in Fig.~\ref{compare-conjecture}. While the action \eq{action-1-DLAc} is not a theory for DLA, it would still be interesting to analyze it, and understand what physical situation it might describe.

\begin{figure}[t!]
\ifTikzExternal{\tikzsetnextfilename{DLAconvPlot}}
 \[{{\vcenter{\hbox{\begin{tikzpicture}[scale=1.565,x={(1,0)},y={(0,1)}]
\coordinate (plot) at (-0.7,-0.24) ;
\coordinate (Ylabel) at (0,5.2);
\coordinate (Xlabel) at (8.2,-0.1);
\coordinate (seed) at (1,5.2) ;
\coordinate (30/77graph) at ($(8.3,5.8)$) ;
\coordinate (30/77) at ($(8.15,{12.45*30/77})$) ;
\coordinate (32/231graph) at ($(8.3,5.05)$) ;
\coordinate (32/231) at ($(8.15,{12.45*32/231})$) ;
\coordinate (9/77graph) at ($(8.45,4.3)$) ;
\coordinate (9/77) at ($(8.15,{12.45*9/77})$) ;
\coordinate (20/231graph) at ($(8.45,3.55)$) ;
\coordinate (20/231) at ($(8.15,{12.45*20/231})$) ;
\coordinate (25/462graph) at ($(8.45,2.5)$) ;
\coordinate (25/462) at ($(8.13,{12.45*25/462})$) ;
\coordinate (4/77graph) at ($(8.63,1.4)$) ;
\coordinate (4/77) at ($(8.15,{12.45*4/77})$) ;
\coordinate (19/462graph) at ($(8.35,0.25)$) ;
\coordinate (19/462) at ($(8.15,{12.45*19/462})$) ;
\coordinate (1/77graph) at (8.3,-0.17) ;
\coordinate (1/77) at ($(8.15,{12.45*1/77})$) ;
\coordinate (R15) at (-0.25,3) ;
\coordinate (R12) at (-0.25,2) ;
\coordinate (Rmany) at (-0.75,-0.17) ;
\coordinate (R125) at (-3.5,-2) ;
\coordinate (R152) at (-2,-2.5) ;
\coordinate (R1523) at (0,-2.5) ;
\fill (0,0) circle (1pt);
\node[anchor=base west] at (plot) {\fig{0.86}{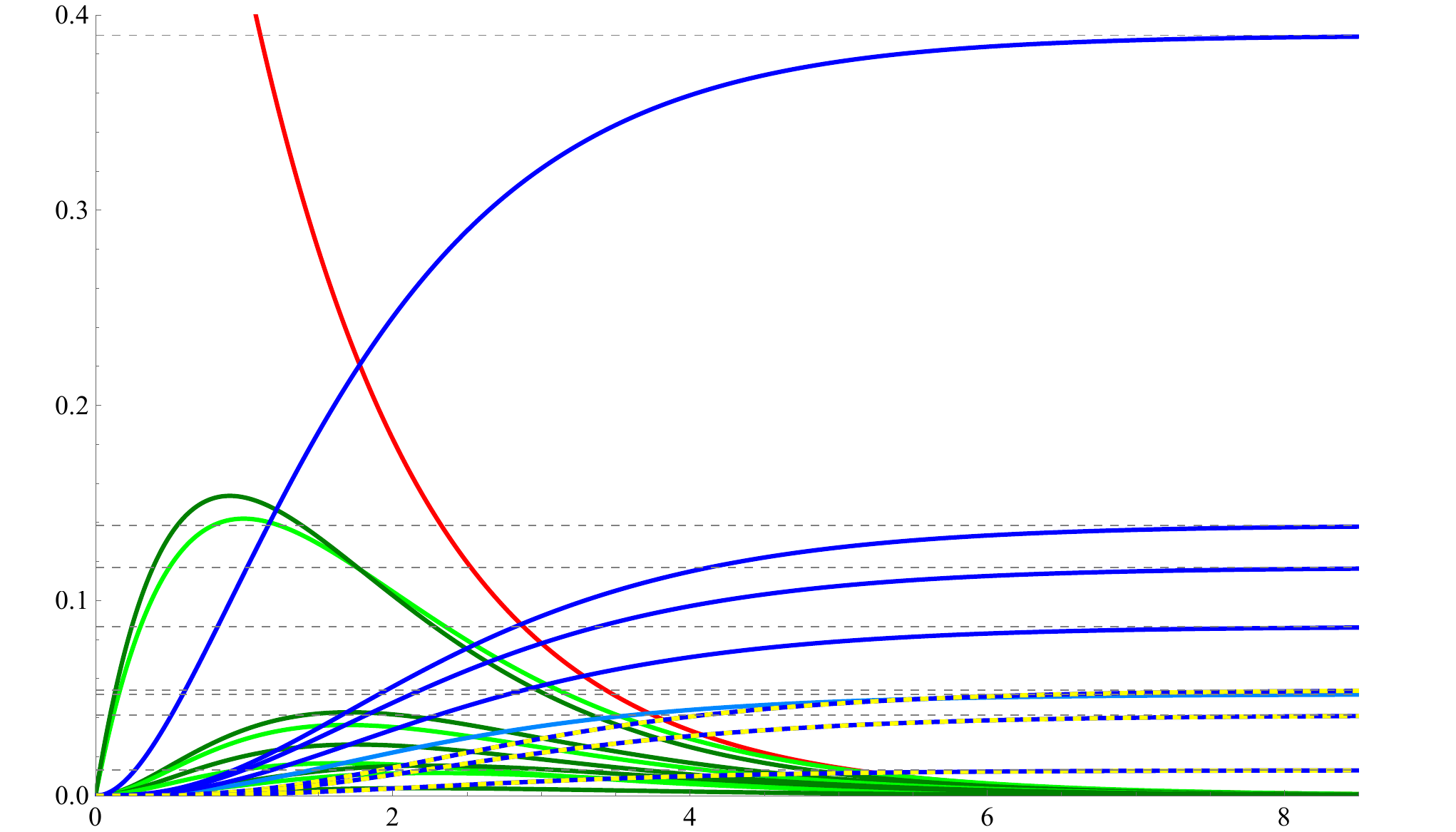}};
\node[font=\scriptsize, fill=white, inner sep=1pt] at (Ylabel) {$\mathbb G^{\textsc{dla}}(t)$};
\node[anchor=west,font=\scriptsize, fill=white, inner sep=1pt] at (Xlabel) {$\gamma t$};
\node[anchor=base west] at (seed) {$\vcenter{\hbox{%
\begin{tikzpicture}[x=0.11mm,y=0.11mm]
 \draw [line width=0.2mm, gray] (-50.,0.) -- (-15.4508,-47.5528) ;
 \draw [line width=0.2mm, gray] (-15.4508,-47.5528) -- (40.4508,-29.3893) ;
 \draw [line width=0.2mm, gray] (40.4508,-29.3893) -- (40.4508,29.3893) ;
 \draw [line width=0.2mm, gray] (-15.4508,47.5528) -- (-15.4508,-47.5528) ;
 \draw [line width=0.2mm, gray] (-50.,0.) -- (-15.4508,47.5528) ;
 \draw [line width=0.2mm, gray] (-15.4508,47.5528) -- (40.4508,29.3893) ;
 \draw [line width=0.2mm, red] (-50.,0.) -- (-63.,0.) ;
 \fill(-50.,0.) circle (1.2pt);
 \fill(-15.4508,-47.5528) circle (1.2pt);
 \fill(40.4508,-29.3893) circle (1.2pt);
 \fill(40.4508,29.3893) circle (1.2pt);
 \fill(-15.4508,47.5528) circle (1.2pt);
\end{tikzpicture}%
 }}$};
\node[anchor=west, inner sep=2pt] at (30/77graph) {$\frac{30}{77}~\vcenter{\hbox{%
 \begin{tikzpicture}[x=0.11mm,y=0.11mm]
 \draw [line width=0.2mm, gray] (-50.,0.) -- (-15.4508,-47.5528) ;
 \draw [line width=0.2mm, gray] (-15.4508,-47.5528) -- (40.4508,-29.3893) ;
 \draw [line width=0.2mm, gray] (40.4508,-29.3893) -- (40.4508,29.3893) ;
 \draw [line width=0.2mm, gray] (-15.4508,47.5528) -- (-15.4508,-47.5528) ;
 \draw [line width=0.2mm, blue,directed] (-50.,0.) -- (-15.4508,47.5528) ;
 \draw [line width=0.2mm, blue,directed] (-15.4508,47.5528) -- (40.4508,29.3893) ;
 \fill(-50.,0.) circle (1.2pt);
 \fill(-15.4508,-47.5528) circle (1.2pt);
 \fill(40.4508,-29.3893) circle (1.2pt);
 \fill(40.4508,29.3893) circle (1.2pt);
 \fill(-15.4508,47.5528) circle (1.2pt);
\end{tikzpicture}%
}}$};
\draw[gray,enddirected] (30/77graph) -- (30/77);
\node[anchor=west, inner sep=2pt] at (32/231graph) {$\frac{32}{231}~\vcenter{\hbox{%
 \begin{tikzpicture}[x=0.11mm,y=0.11mm]
\draw [line width=0.2mm, blue,directed] (-50.,0.) -- (-15.4508,-47.5528) ;
\draw [line width=0.2mm, gray](-15.4508,-47.5528) -- (40.4508,-29.3893) ;
\draw [line width=0.2mm, gray] (40.4508,-29.3893) -- (40.4508,29.3893) ;
\draw [line width=0.2mm, gray] (-15.4508,47.5528) -- (-15.4508,-47.5528) ;
\draw [line width=0.2mm, blue,directed] (-50.,0.) -- (-15.4508,47.5528) ;
\draw [line width=0.2mm, blue,directed] (-15.4508,47.5528) -- (40.4508,29.3893) ;
\fill(-50.,0.) circle (1.2pt);
\fill(-15.4508,-47.5528) circle (1.2pt);
\fill(40.4508,-29.3893) circle (1.2pt);
\fill(40.4508,29.3893) circle (1.2pt);
\fill(-15.4508,47.5528) circle (1.2pt);
 \end{tikzpicture}
}}$};
\draw[gray,enddirected] (32/231graph) -- (32/231);
\node[anchor=west, inner sep=2pt] at (9/77graph) {$\frac{9}{77}~\vcenter{\hbox{%
 \begin{tikzpicture}[x=0.11mm,y=0.11mm]
\draw [line width=0.2mm, blue,directed] (-50.,0.) -- (-15.4508,-47.5528) ;
\draw [line width=0.2mm, blue,directed] (-15.4508,-47.5528) -- (40.4508,-29.3893) ;
\draw [line width=0.2mm, blue,directed] (40.4508,-29.3893) -- (40.4508,29.3893) ;
\draw [line width=0.2mm, gray] (-15.4508,47.5528) -- (-15.4508,-47.5528) ;
\draw [line width=0.2mm, gray] (-50.,0.) -- (-15.4508,47.5528) ;
\draw [line width=0.2mm, gray] (-15.4508,47.5528) -- (40.4508,29.3893) ;
\fill(-50.,0.) circle (1.2pt);
\fill(-15.4508,-47.5528) circle (1.2pt);
\fill(40.4508,-29.3893) circle (1.2pt);
\fill(40.4508,29.3893) circle (1.2pt);
\fill(-15.4508,47.5528) circle (1.2pt);
 \end{tikzpicture}
}}$};
\draw[gray,enddirected] (9/77graph) -- (9/77);
\node[anchor=west, inner sep=2pt] at (20/231graph) {$\frac{20}{231}~\vcenter{\hbox{%
 \begin{tikzpicture}[x=0.11mm,y=0.11mm]
 \draw [line width=0.2mm, blue,directed] (-50.,0.) -- (-15.4508,-47.5528) ;
 \draw [line width=0.2mm, gray] (-15.4508,-47.5528) -- (40.4508,-29.3893) ;
 \draw [line width=0.2mm, gray] (40.4508,-29.3893) -- (40.4508,29.3893) ;
 \draw [line width=0.2mm, blue,directed] (-15.4508,-47.5528) -- (-15.4508,47.5528) ;
 \draw [line width=0.2mm, gray] (-50.,0.) -- (-15.4508,47.5528) ;
 \draw [line width=0.2mm, blue,directed] (-15.4508,47.5528) -- (40.4508,29.3893) ;
 \fill(-50.,0.) circle (1.2pt);
 \fill(-15.4508,-47.5528) circle (1.2pt);
 \fill(40.4508,-29.3893) circle (1.2pt);
 \fill(40.4508,29.3893) circle (1.2pt);
 \fill(-15.4508,47.5528) circle (1.2pt);
 \end{tikzpicture}%
 }}$};
\draw[gray,enddirected] (20/231graph) -- (20/231);
\node[anchor=west, inner sep=2pt] at (25/462graph) {$
 \begin{cases}
 \frac{25}{462}\vcenter{\hbox{%
 \begin{tikzpicture}[x=0.11mm,y=0.11mm]
 \draw [line width=0.2mm, blue,directed] (-50.,0.) -- (-15.4508,-47.5528) ;
 \draw [line width=0.2mm, blue,directed] (-15.4508,-47.5528) -- (40.4508,-29.3893) ;
 \draw [line width=0.2mm, gray] (40.4508,-29.3893) -- (40.4508,29.3893) ;
 \draw [line width=0.2mm, gray] (-15.4508,-47.5528) -- (-15.4508,47.5528) ;
 \draw [line width=0.2mm, blue,directed] (-50.,0.) -- (-15.4508,47.5528) ;
 \draw [line width=0.2mm, blue,directed] (-15.4508,47.5528) -- (40.4508,29.3893) ;
 \fill(-50.,0.) circle (1.2pt);
 \fill(-15.4508,-47.5528) circle (1.2pt);
 \fill(40.4508,-29.3893) circle (1.2pt);
 \fill(40.4508,29.3893) circle (1.2pt);
 \fill(-15.4508,47.5528) circle (1.2pt);
 \end{tikzpicture}%
 }}
 \\[10pt]
 \frac{25}{462}\vcenter{\hbox{%
 \begin{tikzpicture}[x=0.11mm,y=0.11mm]
 \draw [line width=0.2mm, blue,directed] (-50.,0.) -- (-15.4508,-47.5528) ;
 \draw [line width=0.2mm, blue,directed] (-15.4508,-47.5528) -- (40.4508,-29.3893) ;
 \draw [line width=0.2mm, blue,directed] (40.4508,-29.3893) -- (40.4508,29.3893) ;
 \draw [line width=0.2mm, gray] (-15.4508,-47.5528) -- (-15.4508,47.5528) ;
 \draw [line width=0.2mm, blue,directed] (-50.,0.) -- (-15.4508,47.5528) ;
 \draw [line width=0.2mm, gray] (-15.4508,47.5528) -- (40.4508,29.3893) ;
 \fill(-50.,0.) circle (1.2pt);
 \fill(-15.4508,-47.5528) circle (1.2pt);
 \fill(40.4508,-29.3893) circle (1.2pt);
 \fill(40.4508,29.3893) circle (1.2pt);
 \fill(-15.4508,47.5528) circle (1.2pt);
 \end{tikzpicture}%
 }}
 \end{cases}
 $};
\draw[gray,enddirected] (25/462graph) -- (25/462);
\node[anchor=west, inner sep=2pt] at (4/77graph) {$\frac{4}{77}~\vcenter{\hbox{%
 \begin{tikzpicture}[x=0.11mm,y=0.11mm]
\draw [line width=0.2mm, gray] (-50.,0.) -- (-15.4508,-47.5528) ;
\draw [line width=0.2mm, gray] (-15.4508,-47.5528) -- (40.4508,-29.3893) ;
\draw [line width=0.2mm, gray] (40.4508,-29.3893) -- (40.4508,29.3893) ;
\draw [line width=0.2mm, blue,directed] (-15.4508,47.5528) -- (-15.4508,-47.5528) ;
\draw [line width=0.2mm, blue,directed] (-50.,0.) -- (-15.4508,47.5528) ;
\draw [line width=0.2mm, blue,directed] (-15.4508,47.5528) -- (40.4508,29.3893) ;
\fill(-50.,0.) circle (1.2pt);
\fill(-15.4508,-47.5528) circle (1.2pt);
\fill(40.4508,-29.3893) circle (1.2pt);
\fill(40.4508,29.3893) circle (1.2pt);
\fill(-15.4508,47.5528) circle (1.2pt);
 \end{tikzpicture}
}}$};
\draw[gray,enddirected] (4/77graph) -- (4/77);
\node[anchor=west, inner sep=2pt] at (19/462graph) {$
 \begin{cases}
 \frac{19}{462}\vcenter{\hbox{%
 \begin{tikzpicture}[x=0.11mm,y=0.11mm]
 \draw [line width=0.2mm, blue,directed] (-50.,0.) -- (-15.4508,-47.5528) ;
 \draw [line width=0.2mm, blue,directed] (-15.4508,-47.5528) -- (40.4508,-29.3893) ;
 \draw [line width=0.2mm, gray] (40.4508,-29.3893) -- (40.4508,29.3893) ;
 \draw [line width=0.2mm, blue,directed] (-15.4508,-47.5528) -- (-15.4508,47.5528) ;
 \draw [line width=0.2mm, gray] (-50.,0.) -- (-15.4508,47.5528) ;
 \draw [line width=0.2mm, blue,directed] (-15.4508,47.5528) -- (40.4508,29.3893) ;
 \fill(-50.,0.) circle (1.2pt);
 \fill(-15.4508,-47.5528) circle (1.2pt);
 \fill(40.4508,-29.3893) circle (1.2pt);
 \fill(40.4508,29.3893) circle (1.2pt);
 \fill(-15.4508,47.5528) circle (1.2pt);
 \end{tikzpicture}%
 }}
 \\[10pt]
 \frac{19}{462}\vcenter{\hbox{%
 \begin{tikzpicture}[x=0.11mm,y=0.11mm]
 \draw [line width=0.2mm, blue,directed] (-50.,0.) -- (-15.4508,-47.5528) ;
 \draw [line width=0.2mm, blue,directed] (-15.4508,-47.5528) -- (40.4508,-29.3893) ;
 \draw [line width=0.2mm, blue,directed] (40.4508,-29.3893) -- (40.4508,29.3893) ;
 \draw [line width=0.2mm, blue,directed] (-15.4508,-47.5528) -- (-15.4508,47.5528) ;
 \draw [line width=0.2mm, gray] (-50.,0.) -- (-15.4508,47.5528) ;
 \draw [line width=0.2mm, gray] (-15.4508,47.5528) -- (40.4508,29.3893) ;
 \fill(-50.,0.) circle (1.2pt);
 \fill(-15.4508,-47.5528) circle (1.2pt);
 \fill(40.4508,-29.3893) circle (1.2pt);
 \fill(40.4508,29.3893) circle (1.2pt);
 \fill(-15.4508,47.5528) circle (1.2pt);
 \end{tikzpicture}%
 }}
 \end{cases}
 $};
\draw[gray,enddirected] (19/462graph) -- (19/462);
\node[anchor=north east] at (1/77graph) {$\aoverbrace[L9U1R]{{\textstyle\frac{1}{77}}~\vcenter{\hbox{%
 \begin{tikzpicture}[x=0.11mm,y=0.11mm]
\draw [line width=0.2mm, gray] (-50.,0.) -- (-15.4508,-47.5528) ;
\draw [line width=0.2mm, blue,directed] (-15.4508,-47.5528) -- (40.4508,-29.3893) ;
\draw [line width=0.2mm, blue,directed] (40.4508,-29.3893) -- (40.4508,29.3893) ;
\draw [line width=0.2mm, blue,directed] (-15.4508,47.5528) -- (-15.4508,-47.5528) ;
\draw [line width=0.2mm, blue,directed] (-50.,0.) -- (-15.4508,47.5528) ;
\draw [line width=0.2mm, gray] (-15.4508,47.5528) -- (40.4508,29.3893) ;
\fill(-50.,0.) circle (1.2pt);
\fill(-15.4508,-47.5528) circle (1.2pt);
\fill(40.4508,-29.3893) circle (1.2pt);
\fill(40.4508,29.3893) circle (1.2pt);
\fill(-15.4508,47.5528) circle (1.2pt);
 \end{tikzpicture}
}},%
{\textstyle\frac{1}{77}}~\vcenter{\hbox{%
 \begin{tikzpicture}[x=0.11mm,y=0.11mm]
\draw [line width=0.2mm, gray] (-50.,0.) -- (-15.4508,-47.5528) ;
\draw [line width=0.2mm, blue,directed] (-15.4508,-47.5528) -- (40.4508,-29.3893) ;
\draw [line width=0.2mm, gray] (40.4508,-29.3893) -- (40.4508,29.3893) ;
\draw [line width=0.2mm, blue,directed] (-15.4508,47.5528) -- (-15.4508,-47.5528) ;
\draw [line width=0.2mm, blue,directed] (-50.,0.) -- (-15.4508,47.5528) ;
\draw [line width=0.2mm, blue,directed] (-15.4508,47.5528) -- (40.4508,29.3893) ;
\fill(-50.,0.) circle (1.2pt);
\fill(-15.4508,-47.5528) circle (1.2pt);
\fill(40.4508,-29.3893) circle (1.2pt);
\fill(40.4508,29.3893) circle (1.2pt);
\fill(-15.4508,47.5528) circle (1.2pt);
 \end{tikzpicture}
}}}$};
\draw[gray,enddirected] ($(1/77graph) + (-0.445,-0.05)$) .. controls (8.25,-0.15) and (8.25,0.12) .. (1/77); 
%
\node[anchor=east, inner sep=1pt] at (R15) {$\vcenter{\hbox{%
\begin{tikzpicture}[x=0.11mm,y=0.11mm]
 \draw [line width=0.2mm, gray] (-50.,0.) -- (-15.4508,-47.5528) ;
 \draw [line width=0.2mm, gray] (-15.4508,-47.5528) -- (40.4508,-29.3893) ;
 \draw [line width=0.2mm, gray] (40.4508,-29.3893) -- (40.4508,29.3893) ;
 \draw [line width=0.2mm, gray] (-15.4508,47.5528) -- (-15.4508,-47.5528) ;
 \draw [line width=0.2mm, blue,directed] (-50.,0.) -- (-15.4508,47.5528) ;
 \draw [line width=0.2mm, gray] (-15.4508,47.5528) -- (40.4508,29.3893) ;
 \fill(-50.,0.) circle (1.2pt);
 \fill(-15.4508,-47.5528) circle (1.2pt);
 \fill(40.4508,-29.3893) circle (1.2pt);
 \fill(40.4508,29.3893) circle (1.2pt);
 \fill(-15.4508,47.5528) circle (1.2pt);
\end{tikzpicture}%
 }}$};
\draw[gray,enddirected] (R15) -- (0.85,1.95);
\node[anchor=east, inner sep=1pt] at (R12) {$\vcenter{\hbox{%
\begin{tikzpicture}[x=0.11mm,y=0.11mm]
 \draw [line width=0.2mm, blue,directed] (-50.,0.) -- (-15.4508,-47.5528) ;
 \draw [line width=0.2mm, gray] (-15.4508,-47.5528) -- (40.4508,-29.3893) ;
 \draw [line width=0.2mm, gray] (40.4508,-29.3893) -- (40.4508,29.3893) ;
 \draw [line width=0.2mm, gray] (-15.4508,47.5528) -- (-15.4508,-47.5528) ;
 \draw [line width=0.2mm, gray] (-50.,0.) -- (-15.4508,47.5528) ;
 \draw [line width=0.2mm, gray] (-15.4508,47.5528) -- (40.4508,29.3893) ;
 \fill(-50.,0.) circle (1.2pt);
 \fill(-15.4508,-47.5528) circle (1.2pt);
 \fill(40.4508,-29.3893) circle (1.2pt);
 \fill(40.4508,29.3893) circle (1.2pt);
 \fill(-15.4508,47.5528) circle (1.2pt);
\end{tikzpicture}%
 }}$};
\draw[gray,enddirected] (R12) -- (0.9,1.8);
\node[anchor=west] at (-0.5,0.27) {$\Big\{$} ;
\node[anchor=north west] at (Rmany) {$\aoverbrace[LU5R]{\vcenter{\hbox{%
\begin{tikzpicture}[x=0.11mm,y=0.11mm]
 \draw [line width=0.2mm, blue,directed] (-50.,0.) -- (-15.4508,-47.5528) ;
 \draw [line width=0.2mm, gray] (-15.4508,-47.5528) -- (40.4508,-29.3893) ;
 \draw [line width=0.2mm, gray] (40.4508,-29.3893) -- (40.4508,29.3893) ;
 \draw [line width=0.2mm, gray] (-15.4508,-47.5528) -- (-15.4508,47.5528) ;
 \draw [line width=0.2mm, blue,directed] (-50.,0.) -- (-15.4508,47.5528) ;
 \draw [line width=0.2mm, gray] (-15.4508,47.5528) -- (40.4508,29.3893) ;
 \fill(-50.,0.) circle (1.2pt);
 \fill(-15.4508,-47.5528) circle (1.2pt);
 \fill(40.4508,-29.3893) circle (1.2pt);
 \fill(40.4508,29.3893) circle (1.2pt);
 \fill(-15.4508,47.5528) circle (1.2pt);
\end{tikzpicture}%
 }}%
 ,
 \vcenter{\hbox{%
\begin{tikzpicture}[x=0.11mm,y=0.11mm]
 \draw [line width=0.2mm, blue,directed] (-50.,0.) -- (-15.4508,-47.5528) ;
 \draw [line width=0.2mm, blue,directed] (-15.4508,-47.5528) -- (40.4508,-29.3893) ;
 \draw [line width=0.2mm, gray] (40.4508,-29.3893) -- (40.4508,29.3893) ;
 \draw [line width=0.2mm, gray] (-15.4508,-47.5528) -- (-15.4508,47.5528) ;
 \draw [line width=0.2mm, gray] (-50.,0.) -- (-15.4508,47.5528) ;
 \draw [line width=0.2mm, gray] (-15.4508,47.5528) -- (40.4508,29.3893) ;
 \fill(-50.,0.) circle (1.2pt);
 \fill(-15.4508,-47.5528) circle (1.2pt);
 \fill(40.4508,-29.3893) circle (1.2pt);
 \fill(40.4508,29.3893) circle (1.2pt);
 \fill(-15.4508,47.5528) circle (1.2pt);
\end{tikzpicture}%
 }}%
 ,
 \vcenter{\hbox{%
\begin{tikzpicture}[x=0.11mm,y=0.11mm]
 \draw [line width=0.2mm, blue,directed] (-50.,0.) -- (-15.4508,-47.5528) ;
 \draw [line width=0.2mm, gray] (-15.4508,-47.5528) -- (40.4508,-29.3893) ;
 \draw [line width=0.2mm, gray] (40.4508,-29.3893) -- (40.4508,29.3893) ;
 \draw [line width=0.2mm, blue,directed] (-15.4508,-47.5528) -- (-15.4508,47.5528) ;
 \draw [line width=0.2mm, gray] (-50.,0.) -- (-15.4508,47.5528) ;
 \draw [line width=0.2mm, gray] (-15.4508,47.5528) -- (40.4508,29.3893) ;
 \fill(-50.,0.) circle (1.2pt);
 \fill(-15.4508,-47.5528) circle (1.2pt);
 \fill(40.4508,-29.3893) circle (1.2pt);
 \fill(40.4508,29.3893) circle (1.2pt);
 \fill(-15.4508,47.5528) circle (1.2pt);
\end{tikzpicture}%
 }}%
 ,
 \vcenter{\hbox{%
\begin{tikzpicture}[x=0.11mm,y=0.11mm]
 \draw [line width=0.2mm, gray] (-50.,0.) -- (-15.4508,-47.5528) ;
 \draw [line width=0.2mm, gray] (-15.4508,-47.5528) -- (40.4508,-29.3893) ;
 \draw [line width=0.2mm, gray] (40.4508,-29.3893) -- (40.4508,29.3893) ;
 \draw [line width=0.2mm, blue,directed] (-15.4508,47.5528) -- (-15.4508,-47.5528) ;
 \draw [line width=0.2mm, blue,directed] (-50.,0.) -- (-15.4508,47.5528) ;
 \draw [line width=0.2mm, gray] (-15.4508,47.5528) -- (40.4508,29.3893) ;
 \fill(-50.,0.) circle (1.2pt);
 \fill(-15.4508,-47.5528) circle (1.2pt);
 \fill(40.4508,-29.3893) circle (1.2pt);
 \fill(40.4508,29.3893) circle (1.2pt);
 \fill(-15.4508,47.5528) circle (1.2pt);
\end{tikzpicture}%
 }}%
 ,
 \vcenter{\hbox{%
\begin{tikzpicture}[x=0.11mm,y=0.11mm]
 \draw [line width=0.2mm, blue,directed] (-50.,0.) -- (-15.4508,-47.5528) ;
 \draw [line width=0.2mm, blue,directed] (-15.4508,-47.5528) -- (40.4508,-29.3893) ;
 \draw [line width=0.2mm, gray] (40.4508,-29.3893) -- (40.4508,29.3893) ;
 \draw [line width=0.2mm, gray] (-15.4508,47.5528) -- (-15.4508,-47.5528) ;
 \draw [line width=0.2mm, blue,directed] (-50.,0.) -- (-15.4508,47.5528) ;
 \draw [line width=0.2mm, gray] (-15.4508,47.5528) -- (40.4508,29.3893) ;
 \fill(-50.,0.) circle (1.2pt);
 \fill(-15.4508,-47.5528) circle (1.2pt);
 \fill(40.4508,-29.3893) circle (1.2pt);
 \fill(40.4508,29.3893) circle (1.2pt);
 \fill(-15.4508,47.5528) circle (1.2pt);
\end{tikzpicture}%
 }}%
 ,
 \vcenter{\hbox{%
\begin{tikzpicture}[x=0.11mm,y=0.11mm]
 \draw [line width=0.2mm, blue,directed] (-50.,0.) -- (-15.4508,-47.5528) ;
 \draw [line width=0.2mm, blue,directed] (-15.4508,-47.5528) -- (40.4508,-29.3893) ;
 \draw [line width=0.2mm, gray] (40.4508,-29.3893) -- (40.4508,29.3893) ;
 \draw [line width=0.2mm, blue,directed] (-15.4508,-47.5528) -- (-15.4508,47.5528) ;
 \draw [line width=0.2mm, gray] (-50.,0.) -- (-15.4508,47.5528) ;
 \draw [line width=0.2mm, gray] (-15.4508,47.5528) -- (40.4508,29.3893) ;
 \fill(-50.,0.) circle (1.2pt);
 \fill(-15.4508,-47.5528) circle (1.2pt);
 \fill(40.4508,-29.3893) circle (1.2pt);
 \fill(40.4508,29.3893) circle (1.2pt);
 \fill(-15.4508,47.5528) circle (1.2pt);
\end{tikzpicture}%
 }}%
 ,
 \vcenter{\hbox{%
\begin{tikzpicture}[x=0.11mm,y=0.11mm]
 \draw [line width=0.2mm, gray] (-50.,0.) -- (-15.4508,-47.5528) ;
 \draw [line width=0.2mm, blue,directed] (-15.4508,-47.5528) -- (40.4508,-29.3893) ;
 \draw [line width=0.2mm, gray] (40.4508,-29.3893) -- (40.4508,29.3893) ;
 \draw [line width=0.2mm, blue,directed] (-15.4508,47.5528) -- (-15.4508,-47.5528) ;
 \draw [line width=0.2mm, blue,directed] (-50.,0.) -- (-15.4508,47.5528) ;
 \draw [line width=0.2mm, gray] (-15.4508,47.5528) -- (40.4508,29.3893) ;
 \fill(-50.,0.) circle (1.2pt);
 \fill(-15.4508,-47.5528) circle (1.2pt);
 \fill(40.4508,-29.3893) circle (1.2pt);
 \fill(40.4508,29.3893) circle (1.2pt);
 \fill(-15.4508,47.5528) circle (1.2pt);
\end{tikzpicture}%
 }}%
 }$} ;
\draw[gray,enddirected] ($(Rmany) + (0.27,-0.05) $) .. controls (-0.7,-0.09) and (-0.7,0.18) .. (-0.4,0.27) ;
 \end{tikzpicture}%
}}}%
}
\]
\caption{The contributions to the process $\mathbb G^{\textsc{dla}}(t)$ are plotted against time $\gamma t$ to show the convergence of the weights in the field theory \eq{action-DLA-dynamic-1}. It is obtained by using $\gamma=0.01$. Each curve is associated with a tree. Starting from an empty graph with a seed (red curve), the graph is explored by intermediate trees (green curves) that finally reach the target, converging to the possible trees with DLA statistics (blue curves). To make overlapping curves visible, some are in yellow dashed.}
\label{DLAconvPlot}
\end{figure}
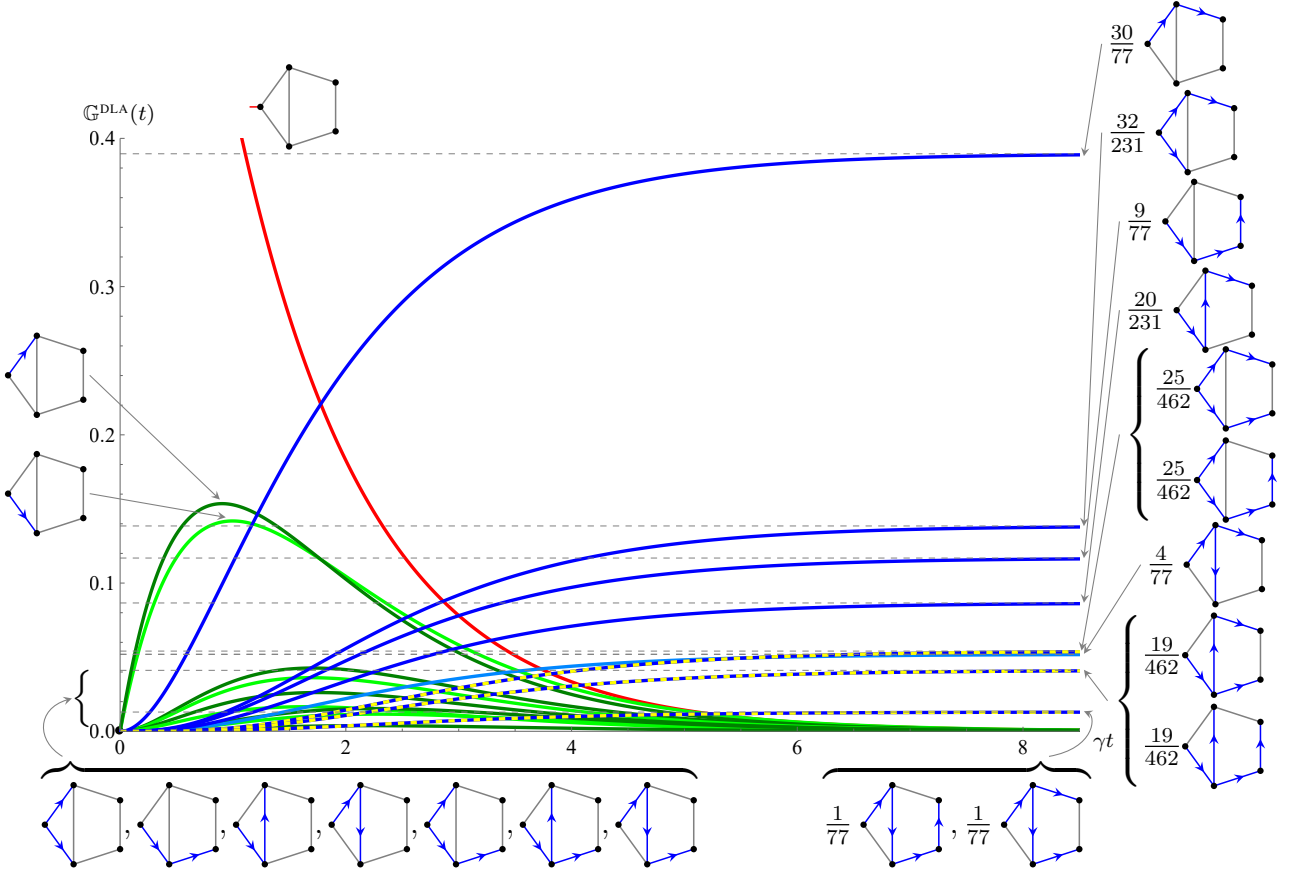

\subsection{Lattice action for DLA as a growth process}
\label{Lattice action for DLA as a growth process}
After the failure of our conjecture in the last section, we now write down a dynamic action following 
what we have done in \Eq{action-LRW-dynamic-1} for Laplacian random walks. This theory contains no longer a tracer field $\psi$, thus is seemingly simpler than the corresponding expression for LRWs. 
\bea \label{action-DLA-dynamic-1}
\!\!\!\rme^{-\ca L_1^{\textsc{DLA}}} &=& \prod_{x \in \ca V} \exp\!\Big( - \tilde \phi (x,t)\phi (x,t) - \sum_{\alpha=1}^{n_{\alpha}}\sum_{j=1}^{n_\chi} \tilde\chi_j^\alpha (x,t)\chi_j^\alpha(x,t)\Big) \nn\\
& \times& \prod_{x \in \ca V\backslash \xt } \bigg\{ \prod_{\alpha=1}^{n_\alpha}\Big[ 1+ \sum_{y \in \ca V}\hat \beta_{xy} \sum_{j=1}^{n_\chi} \tilde \chi_j^\alpha (y,t)\chi_j^\alpha(x,t)\Big] + \tilde \phi (x,t{+}1)\phi(x,t) \nn\\
&& \qquad~~~~~ + 
 \gamma \sum_{y\in \ca V} { \beta_{xy}}
 \Big[ R_{xy} \tilde \phi (y,t{+}1)\tilde \phi(x,t{+}1) -\tilde \phi(x,t{+}1)\Big] \phi (x,t) \tilde \chi_2^1 (y,t) \bigg\} \nn\\
 & \times& \Big[ 1+ \tilde \phi(\xt ,t{+}1) \phi(\xt,t) + \chi_{2}^1 (\xt,t) \Big].
\eea

\subsection{Numerical check}
\label{Numerical check}
Fig.~\ref{DLAconvPlot} shows how this action calculates the DLA weights in \Eq{DLA-on-G1}. Starting with a seed at the leftmost vertex (marked in red), evolution progresses over several intermediate graphs to converge to the weights given in \Eq{DLA-on-G1}.

\subsection{Improved lattice action}
\label{DLA:Improved lattice action}
As we did before for LRWs, the action \eq{action-DLA-dynamic-1} can be improved to contain fewer interaction vertices. 
\bea\label{action-DLA-dynamic-2}
\lefteqn{\rme^{-\ca L_2^{\textsc{DLA}}} 
= \prod_{x \in \ca V\backslash \xt } \Bigg\{ \exp\!\Big( - \tilde \phi (x,t) \phi (x,t) 
 { -} \sum_{\alpha=1}^{n_{\alpha}}\sum_{j=1}^{n_\chi} \tilde\chi_j^\alpha (x,t)\chi_j^\alpha(x,t)\Big)}
 \nn\\
& \times & 
\bigg[ \prod_{\alpha=1}^{n_\alpha} \Big( 1+ \sum_{y \in \ca V}\hat \beta_{xy} \sum_{j=1}^{n_\chi} \tilde \chi_j^\alpha (y,t)\chi_j^\alpha(x,t)\Big)-1 \nn\\
&& + \exp\!\Big( \tilde \phi (x,t{+}1) \phi (x,t) {+} \gamma\sum_{y\in \ca V} { \beta_{xy}}
 \big[ R_{xy} \tilde \phi (y,t{+}1)\tilde \phi(x,t{+}1) {-}\tilde \phi(x,t{+}1)\big] \phi (x,t) \tilde \chi_2^1 (y,t) \Big) \bigg] \Bigg\} \nn\\
 & \times& \exp\!\Big( {-} \tilde \phi (\xt,t) \phi (\xt,t) { -} \sum_{\alpha=1}^{n_{\alpha}}\sum_{j=1}^{n_\chi} \tilde\chi_j^\alpha (\xt,t)\chi_j^\alpha(\xt,t) \Big) \Big[1+\tilde \phi (\xt ,t{+}1)\phi (\xt,t) + \chi_{2}^1 (\xt,t) \Big].
\eea
The continuum limit is similar to \Eq{L-LERW-3}. We will analyze this action and the ensuing field theory in forthcoming work \cite{PisapiaWiese2027}.

\section{Conclusion and Outlook}
In this work we have shown how to construct a lattice action for LRWs and DLA. It propagates the generating function from a given initial configuration. Each term evolves via Laplacian growth until a path or cluster has reached a prescribed target. Alternatively, one can study an ensemble of fixed elapsed times or its Laplace transform. In a second step we went to the continuum limit. We showed how LERW configurations emerge in the large-time limit of Laplacian growth. This formulation is simple enough to allow for a perturbative renormalization-group analysis: first for 
the $b$-Laplacian random walk, which was previously unattainable via the renormalization group. 
We also have preliminary result for DLA which has eluded a proper theoretical treatment since its introduction forty years ago. We will report on these advances in forthcoming publications.

\section*{Acknowledgements}
We thank F.~David, D.~Gaby, G.~Lawler, A.~Nahum, and I.~Procaccia for stimulating discussions.

\appendix

\section*{Appendices}

\section{Perturbative calculations with action $\ca L_1^{\textsc{LRW}}$}
\label{a:Perturbative calculations}

\subsection{Notations}
\label{a:notations}
Throughout the paper we employ the following conventions: full lines for propagators, 
dotted lines to separate interaction vertices so that they are easier to read, dashed lines for the fluctuating part of the fields ($\delta \tilde \chi$ and $\delta\tilde\psi$), and dash-dotted lines for omitted propagators or parts of a diagram. 

Let us start with the vertices we keep after simplifications. We always give $- \ca L $ to avoid additional minus signs in the perturbative expansion; the weight is $\rme^{-\int_t\ca L}$.
The simplified action \eq{g-terms-3}-\eq{relevant-gamma-vertices-3} is
\bea \label{g-terms-3+}
-\ca L_{3}^g &=& - g \,\vertexIntPsiChi \to -g \, \vertexIntPsiChiRel, 
\\
\label{relevant-gamma-vertices-3+}
-\ca L_3^\gamma &=& \gamma_1\, \vertexOneDashed +\gamma_2^+\,\vertexTwoRel- \gamma_2^-\,\vertexTwoMinus = \gamma_1\, \vertexOneDashed +\gamma_{\textsc l}\,\vertexTwoBred + 2\gamma_{\textsc d}\,\vertexTwoCred.
\eea
(The indices $_{\textsc l}$ and $_{\textsc d}$ indicate Laplacian and drift.)
Expanding $\tilde\chi=1+\delta\tilde\chi$, we get 
\begin{equation}
\gamma_1\, \vertexOneDashed = \gamma_1^1\, \vertexOneDashedmean + \gamma_1^{\delta \chi}\, \vertexOneDashedFluctuating, \qquad
- g \vertexIntPsiChiRel = - g^1\, \vertexIntPsiChiRelMean - g^{\delta \chi}\, \vertexIntPsiChiRelFluct . 
\end{equation}
In proposition \ref{A=all} only the following vertices contribute to class $\mathbb A$
\be\label{L-A}
-\ca L_{\mathbb A} =
 \gamma_1^1\, \vertexOneDashedmean -\gamma_2^-\,\vertexTwoMinus - g^1\, \vertexIntPsiChiRelMean. \end{equation}
For class $\mathbb B$ there are in addition
\begin{equation}%
-\ca L_{\mathbb B}=	 \gamma_1^{\delta \chi}\, \vertexOneDashedFluctuating + \gamma_2^+ \vertexTwoRel - g^{\delta \chi} \vertexIntPsiChiRelFluct.
\end{equation}%
Similarly there are irrelevant vertices that exist in action \eq{action-LRW-dynamic-1-continuum}, but which we exclude through power counting in section \ref{Power counting} 
\begin{equation}
\label{106}
-\ca L_{\rm{irrel}}=	\gamma_2^{+\psi} \vertexTwoBlueDashed - g\gamma^{\psi} \vertexGGammaPsi -g\gamma^{\chi} \vertexGGammaChi.
\end{equation}
Their 1-loop contributions are shown to cancel in section \ref{a:tadpoles}.
Finally, some interactions are suppressed by summing chains of them in appendices \ref{a:LRW1-perturbative} and \ref{a:Tracer-tracer interactions},
\begin{equation}
-\ca L_{\phi\psi}^g = - g_{\phi\psi} \RedBlueInt - g_{\psi\psi}\,\BlueBlueInt\ .
\end{equation}

\subsection{Tip-tracer interactions}
\label{a:LRW1-perturbative}
Interactions between $\phi$ (the tip, red) and $\psi$ (the tracer, blue) are present in action \eq{action-LRW-dynamic-1-continuum},
\begin{equation}
 - g_{\phi\psi} \RedBlueInt .
\end{equation}
This interaction is unphysical, as the tracer and the tip never interact directly in the lattice action. Instead, their interaction is mediated by the explorer (green line), which checks if a site is occupied or not before moving the tip there. 
Consistent with this observation, the interaction should be irrelevant. To see this, note that 
 one can stack an arbitrary number of $g_{\phi\psi}$ vertices on top of each other, creating a ``chain of tadpoles''
\begin{align}
 - g_{\phi\psi} \RedBlueInt
 + \big( {-} g_{\phi\psi} \big)^2 \RedBlueIntChainOne
 + \big( {-} g_{\phi\psi} \big)^3 \RedBlueIntChainTwo
 + ... \nn
 \simeq &\, {-} g_{\phi\psi} \RedBlueInt
 \Bigg[ 1 {+} \Big( {-} g_{\phi\psi} \tadpole
 \Big)%
 {+} \Big({-} g_{\phi\psi} \tadpole 
 \Big)^2%
 {+} ...%
 \Bigg]%
 \\
 =& - g_{\phi\psi} \RedBlueInt \frac{1}{1+g_{\phi\psi} \tadpole} .
\end{align}
As the tadpole diagram is very large, the effective interaction vanishes. 

\begin{table}[ht!]%
\[\renewcommand{\arraystretch}{2}%
\begin{array}{>{\displaystyle}c|>{\displaystyle}c|>{\displaystyle}c|>{\displaystyle}c|>{\displaystyle}c|>{\displaystyle}c|>{\displaystyle}c|>{\displaystyle}c}
\hline\hline
\multirow{2}{*}{1st} & \multirow{2}{*}{2nd} & \multirow{2}{*}{3rd} & \multirow{2}{*}{diagrams} & \multirow{2}{*}{integrated diagrams} & \multicolumn{3}{c}{\text{result}}
\\
\cline{6-8}
 & & & & & \text{\begin{minipage}{1.7cm}{\begin{center}vaccuum bubble\end{center}}\end{minipage}} & \text{~tadpole~} & \text{\begin{minipage}{1.5cm}{\begin{center}2-point loop\end{center}}\end{minipage}}
\\
\hline\hline
 	-g\gamma^{\chi} &‑&‑&\diagGGammaChi & -g \gamma^{\chi} \tadpole \vertexOneDashedGreenLess &‑& -\tadpole&‑ 
\\[5pt]
\hline
 \gamma_1 & {-}g \gamma^\psi &‑&\diagGammaOneGGammaPsi & {-}g \frac{\gamma_1 \gamma^\psi}{\gamma} \tadpole \vertexOneDashedGreenLess &‑& -\tadpole&‑
\\[5pt]
\hline
	\gamma_2^{+\psi}& {-}g \gamma^\psi &‑&\diagGammaTwoPsiGGammaPsi & g \frac{\gamma_2^{+\psi} \gamma^\psi}{\gamma} \bubble \vertexOneDashedGreenLess & \bubble &‑&‑
\\[5pt]
\hline
\multirow[c]{3}{*}[-8pt]{$\gamma_1$} & \multirow[c]{3}{*}[-8pt]{$\gamma_2^{+\psi}$} &-g &\diagGammaOneGammaTwoPsi & g \frac{\gamma_1 \gamma_2^{+\psi} }{\gamma} \bigg[\tadpole + \banana p^2 \bigg] \vertexOneDashedGreenLess^{p} &‑&\hphantom{-}\tadpole& \hphantom{-}\banana p^2 
\\[5pt]
	& & {(-g )}\gamma_2^+ &\diagSimplifyingTadpoleGammaOneGammaPlus & -g \frac{\gamma_1 \gamma_2^{+\psi} \gamma_2^+}{\gamma^2} \bigg[\tadpole + \banana p^2 \bigg] \vertexOneDashedGreenLess^{p} &‑& -\tadpole& -\banana p^2
\\[5pt]
	& & {(-g )}({-} \gamma_2^-) &\diagSimplifyingTadpoleGammaOneGammaMinus & g \frac{\gamma_1 \gamma_2^{+\psi} \gamma_2^- }{\gamma^2} \tadpole \vertexOneDashedGreenLess &‑& \hphantom{-}\tadpole& ‑
\\[5pt]
\hline
	\multirow[c]{6}{*}[-14pt]{$\gamma_2^{+\psi}$}& \multirow[c]{3}{*}[-8pt]{$\gamma_2^{+\psi}$} &{-g }&\diagGammaTwoPsiGammaTwoPsi & -g \frac{\gamma_2^{+\psi} \gamma_2^{+\psi} }{\gamma} \bigg[\bubble + \tadpole p^2 \bigg] \vertexOneDashedGreenLess^{p} &-\bubble& -\tadpole p^2& ‑ 
	\\[5pt]
	& & {(-g )}\gamma_2^+ &\diagSimplifyingTadpoleGammaPlus & g \frac{\gamma_2^{+\psi} \gamma_2^{+\psi} \gamma_2^+}{\gamma^2} \bigg[\bubble + \tadpole p^2 \bigg] \vertexOneDashedGreenLess^{p} &\hphantom{-}\bubble&\hphantom{-} \tadpole p^2& ‑
\\[5pt]
	& & {(-g )}({-} \gamma_2^- ) &\diagSimplifyingTadpoleGammaMinus & -g \frac{\gamma_2^{+\psi} \gamma_2^{+\psi} \gamma_2^- }{\gamma^2} \bubble \vertexOneDashedGreenLess &-\bubble& ‑ & ‑
\\[5pt] \cline{2-8}
	& \multirow[c]{3}{*}[-8pt]{$\gamma_1$} &{-g }&\diagSimplifyingTadpoleGammaTwoPsi & g \frac{\gamma_1 \gamma_2^{+\psi} }{\gamma} \tadpole \vertexOneDashedGreenLess & ‑ & \hphantom{-}\tadpole & ‑
	\\[5pt]
	& & {(-g )}\gamma_2^+ &\diagSimplifyingTadpoleGammaTwoPsiGammaPlus & -g \frac{\gamma_1 \gamma_2^{+\psi} \gamma_2^+}{\gamma^2} \tadpole \vertexOneDashedGreenLess & ‑ & -\tadpole & ‑
\\[5pt]
	& & {(-g )}({-}\gamma_2^-) &\diagSimplifyingTadpoleGammaTwoPsiGammaMinus & g \frac{\gamma_2^+\gamma_1 \gamma_2^-}{\gamma^2} \tadpole \vertexOneDashedGreenLess & ‑ & \hphantom{-}\tadpole & ‑
\\ 
\hline\hline
\end{array}\]
\caption{All diagrams originating from irrelevant vertices in the Lagrangian $\ca L_1^{\rm LRW}$ defined in \Eq{action-LRW-dynamic-1}. First three columns: vertices used, ordered in time; the $g$ always occurs with the third $\gamma$. Fourth column: diagram obtained using those vertices; fifth column: result after time-integration and simplifications; last column: summary. 
As is evident from the last column, the sum of all diagrams vanishes.}
\label{tadpole-table}
\end{table}

\subsection{Tracer-tracer interactions}
\label{a:Tracer-tracer interactions}
\begin{equation}
 - g_{\psi\psi}\,\BlueBlueInt
\end{equation}
This blue-blue interaction should be irrelevant as the red-blue interaction (section \ref{a:LRW1-perturbative}), since two blue lines can never be produced at the same site, and while they propagate in time, they do not propagate in space.
To see this, stack $g_{\psi\psi}$ vertices on top of each other to create a ``chain of bubbles''
\begin{align}\nn
	- g_{\psi\psi}\BlueBlueInt
 + \big( {-} g_{\psi\psi} \big)^2~ \BlueBlueIntChainOne
 + \big( {-} g_{\psi\psi} \big)^3~ \BlueBlueIntChainTwo
 + ...%
 \simeq& - g_{\psi\psi} \,\BlueBlueInt 
 \Bigg[ 1 {+} \Big( {-} g_{\psi\psi} \bubble \Big) {+} \Big({-} g_{\psi\psi} \bubble \Big)^2%
 {+} ...%
 \Bigg]%
 \\
 =& - g_{\psi\psi}\,\BlueBlueInt\,
 \frac{1}{1+g_{\psi\psi}\bubble }.
\end{align}
The diagram is again very large, 
\be
\bubble := \int d^d k .
\ee
Thus the blue-blue interaction is also irrelevant. 

\subsection{Tadpole and bubble cancelations}
\label{a:tadpoles}
There are many tadpole and bubble diagrams appearing in the perturbative expansion obtained from the Lagrangian $\ca L_1^{\textsc{LRW}}$. 
While tedious to calculate, they all cancel. To illustrate this, all such 1-loop diagrams are collected in table \ref{tadpole-table}. 
Used are the less relevant vertices $\gamma_2^{+\psi}$ from \Eq{relevant-gamma-terms} and $g\gamma^{\chi}$, $g \gamma^\psi$ from \Eq{eq:gGammaVertices}; they are collected in \Eq{106}.


\section{Simplifications and their proof}\label{app:pOperation Proof}
Here we treat the Lagrangian $\ca L_3^{\textsc{LRW}}$ with the vertices in \Eqs{g-terms-3+}-\eq{relevant-gamma-vertices-3+}. 
The probability 
\be
\left< \ca O(y,\xt)\right> := \lim_{T\to \infty} \left< \ca O(y,\xt|T)\right>
\ee that a path emanating from $\xr=0$ passes through $y$ before ending in $\xt$ was considered at tree level in section \ref{Observable: passing through a point}, 
and to 1-loop order in section \ref{1-loop corrections to the observer}. An important test of our theory 
is that $\left< \ca O(y,\xt)\right>$ has the same perturbative expansion as for LERWs. As shown above, at tree and 1-loop order this is the case. Here we show how this continues to higher orders.

What we will prove first is that the sum over all diagrams can be reduced to the sum over a subset $\mathbb A$, which contains only single-interacting vertices of the $\gamma_2^-$ type, and the similarly non-interacting $\gamma_1^1$, without any independent sub-loop. Its complement $\mathbb B $ contains diagrams that cancel in pairs.
In a second step, we verify to 3-loop order that diagrams in class $\mathbb A$ are equivalent to diagrams for LERWs. 

To formalize the proof, we use the following conventions: 
the red line goes from $(0,0)$ to $(\xt,T)$. When we say \emph{before} or \emph{after}, we use time in the red line as a reference; this also sets the order of the vertices. When the red line continues in green to form a loop, the order continues along the direction of the green line.
\begin{definition}
 An \emph{interacting vertex} $\Gamma$ is a vertex whose outgoing green line interacts. If the green line interacts only once, $\Gamma$ is said to be \emph{single-interacting}, otherwise it is \emph{multi-interacting}. Multi-interacting vertices are denoted $ \Gamma^{{(n)}}$, where $n$ is the number of interactions. For single-interacting vertices, $n=1$ will be omitted.
\end{definition}
\noindent
Examples of a multi-interacting $\gamma_2^+$ and a multi-interacting $\gamma_2^-$ vertex with $n=2$ are
\begin{equation}
 \gamma_2^{+(2)} = \interactingVertexGAmmaPlusExample , 
 \qquad \gamma_2^{-(2)} =\interactingVertexExample \;.
\end{equation}
A multi-interacting vertex is a $\gamma$-vertex followed by a sequence of $g$-vertices. In the previous example $ \gamma_2^{-(2)}= \gamma_2^{-} g^2$.
\begin{definition}\label{+sub-loop}
 Let $\mathbb B_{\rm L}^+$ be the set of diagrams that contains the ordered sequence of vertices
\begin{equation}
 L^+(k,n)=\underset{\LplusExample}{\gamma_1^1 \to \big( \gamma_2^- \big)^{k} \to \gamma_2^{+(n+1)}}.
\end{equation}
This sequence is said to be an \emph{independent plus sub-loop}.  (Note that all $\gamma_2^-$ vertices are single-interacting.)
\end{definition}
\noindent 
Note that we use $(n+1)$ for the last $\gamma_2^+$ vertex to highlight that there are $n$ intermediate interactions, while the last one is used to complete the loop: it interacts with the $\gamma_1$-vertex at the beginning of the sequence. For example 
\begin{equation}
L^+(2,1)=\independetLoopEG
.\end{equation} 
Analogously, we define $\mathbb B_{\rm L}^-$ by using $ \gamma_2^-$ as the last vertex.
\begin{definition}\label{-sub-loop}
Let $\mathbb B_{\rm L}^-$ be the set of diagrams that contains the following ordered sequence of vertices
\begin{equation}
 L^-(k,n)=\underset{\LplusExample}{\gamma_1^1 \to \big( \gamma_2^- \big)^{k} \to \gamma_2^{-(n+1)}}.
\end{equation}
 This sequence is said to be an \emph{independent minus sub-loop}. 
\end{definition}
\noindent
An example is 
\be
 L^-(1,2)= \independetLoopMinusEG .
\ee
To show cancelation of all diagrams in $\mathbb B$, we first show cancelation between $\mathbb B_{\rm L}^+$ and $\mathbb B_{\rm L}^-$.

\begin{definition}
The set of all diagrams with sub-loops is $\mathbb B_{\rm L}:= \mathbb B_{\rm L}^+\cup \mathbb B_{\rm L}^-$. We say a diagram $\ca D$ has a sub-loop if $\ca D\in \mathbb B_{\rm L}$. 
\end{definition}
\begin{proposition}
Distinct sub-loops have no common vertex.
\end{proposition}
\begin{proof}
By definition, a sub-loop contains a single $\gamma_1^1$ vertex at its start. 
\end{proof}

\begin{proposition}\label{prop:independent loops}
 Every independent plus sub-loop is canceled by an independent minus sub-loop, and \emph{vice versa}. The canceling pairs are
\begin{equation}\label{112}
\begin{split}
 L^+(k,n) &\leftrightarrow L^-(k{-}1,n{+}1),\quad k>0,\\
 L^+(0,n) &\leftrightarrow L^-(n,0).
 \end{split}
\end{equation}
\end{proposition}

\begin{proof}
 We need to find a one-to-one correspondence between independent plus and minus sub-loops, and show that these pairs cancel.
 Take an independent plus sub-loop $L^+(k,n)$, and let $I=\{1,\dots,k,{k+1} \dots,{k+n} \}$ be the labels of the interacting incoming blue lines, ordered by following first the red line for $i=1,\dots,k$, and then the green line for $i=k+1,\dots,k+n$. Examples are given in \Eqs{116} and \eq{114} below.

Let us verify these cancelations, first for $k>0$. We can simplify $L^+(k,n)$ by acting with the Laplacian inside the loop, and then integrating over time, 
\begin{equation}\label{116}
L^+(k,n)=
\independetLoopEGlabeled
\simeq
-\independetLoopEGsimp 
\simeq
-\independetLoopEGsimpIntegrated.
\end{equation}
According to \Eq{112} its partner is
\begin{equation}
\label{114}
L^-(k{-}1,n{+}1)=
\independetLoopMinusEGlabeled
\simeq\independetLoopMinusEGsimp
\simeq \independetLoopMinusEGsimpIntegrated.
\end{equation}
Consider the remaining case $k=0$, for which the first diagram reads 
\begin{equation}
L^+(0,n)=
\GammaTwointeractingPrecedingGammaOnebis
\simeq
-\GammaTwointeractingPrecedingGammaOneIntegratedbis .
\end{equation}
Its partner is
\begin{equation}\label{117}
L^-(n,0)=\GammaTwointeractingPrecedingGammaOneSubbedbis
\simeq
\GammaTwointeractingPrecedingGammaOneSubbedIntegratedbis.
\end{equation}
Pairs of diagrams vanish due to the sign difference between the $\gamma_2^+$ and the $\gamma_2^-$ vertices. 
\end{proof}

\begin{proposition}\label{cor:Bl}
 There is a one-to-one correspondence and cancelation between diagrams in $\mathbb B_{\rm L}^+$ and $\mathbb B_{\rm L}^-$.
\end{proposition}
\begin{proof}
 The proof follows from Proposition \ref{prop:independent loops} by matching diagrams according to the independent sub-loop encountered first when following the red line from its starting point. By encounter we mean reaching the first $\gamma$-vertex of the loop. 
\end{proof}

\begin{definition}
 Consider diagrams without independent sub-loops. They can be divided into subsets according to the order in which interacting vertices appear,
\begin{itemize}
 \item $\mathbb A :=$ \{diagrams with no independent sub-loop, and where the only interacting vertices are single-interacting $\gamma_2^-$ vertices\},
 \item $\mathbb B_1 :=$ \{diagrams with no independent sub-loops and whose first interacting vertex that is not a single-interacting $\gamma_2^-$ is $\gamma_1^{\delta \chi}$\}, 
 \item $\mathbb B_2^- :=$ \{diagrams with no independent sub-loops and whose first interacting vertex that is not a single-interacting $\gamma_2^-$ is a multi-interacting $\gamma_2^-$\},
 \item $\mathbb B_2^+ :=$ \{diagrams with no independent sub-loops and whose first interacting vertex that is not a single-interacting $\gamma_2^-$ is $\gamma_2^+$\}.
\end{itemize}
\end{definition}
\noindent 
\begin{proposition}\label{prop:Ddecomposition}
The set $\mathbb D$ of all diagrams can be decomposed into disjoint sets 
\be
\mathbb D = \mathbb A \,\dot \cup\, \mathbb B_{\rm L}^+ \,\dot \cup \, \mathbb B_{\rm L}^- \,\dot \cup \, \mathbb B_1 \,\dot \cup\, \mathbb B_2^+ \dot \cup\,\mathbb B_2^-.
\ee 
\end{proposition}
\begin{proof}
Take a diagram $\ca D$. 
If it has a sub-loop, it is either in $\mathbb B_{\rm L}^+$ or $\mathbb B_{\rm L}^-$.
Otherwise follow the red line from its starting point. If there are only $\gamma_1^1$ and single-interacting $\gamma_2^-$ vertices along this line, then $\ca D \in \mathbb A$. 
Else consider the first vertex other than $\gamma_1^1$ or a single-interacting $\gamma_2^-$. It either is a $\gamma_1^{\delta \chi}$, a multi-interacting $\gamma_2^-$ or a $\gamma_2^+$, which puts $\ca D$ into $\mathbb B_1$, $\mathbb B_2^-$ or $\mathbb B_2^+$ respectively.
\end{proof}

\begin{proposition}\label{prop:Bplus}
 There is a one-to-one correspondence and cancelation between diagrams in $\mathbb B_1 \cup \mathbb B_2^-$ and diagrams in $\mathbb B_2^+$.
\end{proposition}
\begin{proof}
 To prove this, we need to find a one-to-one correspondence and then show that corresponding diagrams cancel.

We start by considering a diagram $\ca D \in \mathbb B_1$, and call $\Gamma$ its first interacting $\gamma_1^{\delta\chi}$ vertex. To construct a diagram $\ca D' \in \mathbb B_2^+$ that cancels it, take $\ca D$ and add a $\gamma_2^+$ vertex after $\Gamma$. Then use the outgoing green line from $\gamma_2^+$ to replace the interacting green line of $\Gamma$. By definition, $\ca D' \in \mathbb B_2^+$. 
This operation can graphically be written as 
\begin{equation}\label{119}
\ca D = \interactingGammaOneBis
\longrightarrow~~
\ca D' = \GammaOneSubbedByGammaTwop
\simeq
- \interactingGammaOne 
.\end{equation}
As is evident from this equation, $\ca D'$ equals $-\ca D$ after time integration, thus $\ca D + \ca D' = 0$. 

We have identified a canceling diagram by starting from $\mathbb B_1$. 
Now consider a diagram $\ca D \in \mathbb B_2^-$. Denote $\Gamma$ its first multi-interacting $\gamma_2^-$ vertex and $\Gamma'$ the first interacting $g$ vertex after $\Gamma$, following the green line. To construct $\ca D' $ take $\ca D$ and add a $\gamma_2^+$ vertex after $\Gamma$. Then use its outgoing green line to replace the interacting green line from $\Gamma'$.
As a result, $\Gamma$ is single-interacting, the next vertex is $\gamma_2^+$, and therefore $\ca D' \in \mathbb B_2^+$. With the corresponding integrations 
\begin{equation}\label{123}
\ca D = \interactingGvertex \simeq
 \interactingGvertexIntegrated
\longrightarrow~~
\ca D' = \GvertexSubbedbyGammaTwo
\simeq
-\interactingGvertexIntegrated
.\end{equation}
As before, $\ca D'$ equals $-\ca D$ after integration, and therefore $\ca D + \ca D' = 0$. 
Let us stress that the order of the interaction vertices is the same in both diagrams, which is important for proper cancelations. 
We thus identified the canceling diagrams by starting from $\mathbb B_2^-$.

Next we need to identify the canceling diagram $\ca D$ starting from $\ca D' \in \mathbb B_2^+$. 
 To this end, look at the vertex preceding the first interacting $\gamma_2^+$ vertex in $\ca D'$. There are two possibilities: it either is a non-interacting $\gamma_1^1$, or a single-interacting $\gamma_2^-$. In the former case, $\ca D$ can be identified by time integration with a diagram in $\mathbb B_1$, see \Eq{119} read from right to left. In the latter case, $\ca D$ can be obtained by ``gluing back'' the green line from $\gamma_2^+$ to the preceding $\gamma_2^- $, see \Eq{123} read from right to left.
 
To complete the proof we have to show that the procedure is a bijection. This is insured since in our construction we used the same \Eqs{119} or \eq{123} both forwards and backwards. 
\end{proof}

\begin{proposition}\label{A=all}
 The observable $\left< \ca O(y,\xt)\right>$ is given by the diagrams in $\mathbb A$.
\end{proposition}
\begin{proof}
The proof follows from the cancelations in propositions \ref{cor:Bl} and \ref{prop:Bplus}, which imply that there are no surviving diagrams other than those in $\mathbb A$.
\end{proof}

\begin{checking}
First, 
to 3-loop order, the observable $\left< \ca O(y,\xt)\right>$ calculated with all vertices from action \eq{L-LERW-3} agrees with the result 
obtained for LERWs in \cite{WieseFedorenko2018}. 
Second, we checked proposition \ref{A=all} to 3-loop order: diagrams in class $\mathbb A$ agree with the result for LERWs in \cite{WieseFedorenko2018}. This is shown in appendix \ref{Observable including sub-loops up to 3-loop order}.
\end{checking}

\section{Observable up to 3-loop order}
\label{Observable including sub-loops up to 3-loop order}
\subsection{1 loop}
In this appendix we give all diagrams that one can construct for $\left< \ca O(y,\xt)\right>$ from action \eq{L-A}. This includes sub-loops as defined in definitions \ref{+sub-loop} and \ref{-sub-loop}. Sub-loops are highlighted; diagrams with sub-loops are to be excluded. We then check that the 2-loop and 3-loop results agree with \cite{WieseFedorenko2018}. For one loop this is done in \Eq{21}.

\subsection{2 loops}

\bea
&&{\parbox{16mm}{

}}.
\eea
This equation shows all diagrams constructed from $\gamma_1^1$ ($\gamma_1$ without outgoing green line) , $\gamma_2^-$ (one outgoing green line with a Laplacian acting on it), and $g^1$ (interaction $g$ without outgoing green line). 
Subloops are highlighted; the corresponding diagrams do not contribute to $\ca O$.
After time integration, the remaining ones are in $\phi^4$ notation
\be {{\parbox{2.1cm}{{%
}}}}~ =\propdiagTwo+ 2 \LHnoExternal \ .
\ee
Multiplying with $(-g)^2$, this agrees with the order $g_0^2$ term in Eq.~(D1) of \cite{WieseFedorenko2018}.

\subsection{3 loops}
\bea
 &&
\hphantom{+}{\vcenter{\hbox{%
%
}}}\qquad 
\eea
Discarding diagrams with sub-loops (highlighted), and integrating over time, we get
\be
 2\;{\parbox{12.2mm}{%
}~.
\ee
Multiplying with $(-g)^3$, 
this agrees with the order $g_0^3$ term in Eq.~(D1) of \cite{WieseFedorenko2018}.


\ifx\doi\undefined
\providecommand{\doi}[2]{\href{http://dx.doi.org/#1}{#2}}
\else
\renewcommand{\doi}[2]{\href{http://dx.doi.org/#1}{#2}}
\fi
\providecommand{\link}[2]{\href{#1}{#2}}
\providecommand{\arxiv}[1]{\href{http://arxiv.org/abs/#1}{#1}}
\providecommand{\hal}[1]{\href{https://hal.archives-ouvertes.fr/hal-#1}{hal-#1}}
\providecommand{\mrnumber}[1]{\href{https://mathscinet.ams.org/mathscinet/search/publdoc.html?pg1=MR&s1=#1&loc=fromreflist}{MR#1}}

\end{document}